\documentclass[11pt,letterpaper]{article}
\usepackage[margin=1in]{geometry}
\usepackage{hyperref}
\usepackage{algorithm}
\usepackage{algorithmic}

\usepackage{hyperref,enumitem}
\hypersetup{
	colorlinks=true,
	citecolor=blue       %
}

\usepackage{natbib}
\usepackage{url}            %
\usepackage{graphicx}
\usepackage{amsmath,amssymb,amsthm}
\usepackage{booktabs} %
\usepackage{color}
\usepackage[font=small,labelfont=bf]{caption}
\usepackage{subcaption}
\usepackage{multirow}
\usepackage{authblk}

\usepackage{comment}

\usepackage{tikz}
\usetikzlibrary{arrows}
\usetikzlibrary{shapes}
\usetikzlibrary{calc}

\newcommand{\R}{{\mathbb R}}

\newcommand{\U}{{\mathcal U}}
\newcommand{\E}{{\mathbf E}}
\newcommand{\N}{{\mathbb N}}
\newcommand{\M}{{\mathcal M}}
\newcommand{\cN}{{\mathcal N}}
\newcommand{\1}{{\mathbb{I}}}
\newcommand{\w}{w}
\newcommand{\h}{{\mathbf{h}}}

\newcommand{\x}{{\mathbf{x}}}

\newcommand{\RB}{\R_{\geq 0}}

\renewcommand{\c}{{\mathbf{c}}}
\newcommand{\C}{{\mathcal{C}}}
\newcommand{\PP}{\mathbf{P}}
\newcommand{\G}{{\mathcal G}}
\newcommand{\fP}{\mathcal{A}}
\newcommand{\fp}{\alpha}
\usepackage{xcolor}
\newcommand{\new}[1]{\textcolor{black}{#1}}
\newcommand{\ZFadd}[1]{\textcolor{black}{#1}}

\newcommand{\ssredit}[1]{\textcolor{black}{#1}}
\newcommand{\pdadd}[1]{\textcolor{black}{#1}}
\newcommand{\todo}[1]{}
\newcommand{\zfedit}[1]{}
\usepackage{color}
\definecolor{english}{rgb}{0.0, 0.5, 0.0}
\newcommand{\dcpadd}[1]{\textcolor{black}{#1}}
\newcommand{\hariadd}[1]{\textcolor{black}{#1}}

\theoremstyle{plain}
\newtheorem{thm}{Theorem}[section]
\newtheorem{definition}{Definition}[section]
\newtheorem{lem}{Lemma}[section]
\newtheorem{cor}{Corollary}[section]
\newtheorem{step}{Step}
\newtheorem{example}{Example}[section]

\newcount\Comments  %
\begin{document}

\title{Optimal Auctions through Deep Learning:\\
Advances in Differentiable Economics\thanks{{\new{David Parkes is currently  on sabbatical at DeepMind as a research scientist. 
This work is supported in part through NSF award
CCF-1841550, as well as a Google Fellowship for Zhe Feng.} We thank Zihe Wang (Shanghai University of Finance and
Economics) for pointing out that the combinatorial feasible definition in the ICML'19 published version of the extended abstract of this paper need not imply an integer decomposition.
      We
would like to thank \new{Dirk Bergemann}, Yang Cai, Vincent
Conitzer, Yannai Gonczarowski, Constantinos Daskalakis, Glenn
Ellison, Sergiu Hart, Ron Lavi, Kevin Leyton-Brown, Shengwu Li,
Noam Nisan, Parag Pathak, Alexander Rush, Karl Schlag, Zihe
Wang, Alex Wolitzky, participants in the Economics and
Computation Reunion Workshop at the Simons Institute, the
NIPS'17 Workshop on Learning in the Presence of Strategic
Behavior, a Dagstuhl Workshop on Computational Learning Theory
meets Game Theory, the EC'18 Workshop on Algorithmic Game Theory
and Data Science, the Annual Congress of the German Economic
Association, participants in seminars at LSE, Technion, Hebrew,
Google, HBS, MIT, and the anonymous reviewers on earlier
versions of this paper for their helpful feedback. The first
version of this paper, originally titled as ``Optimal Auctions through Deep Learning", was posted on arXiv on June 12,
2017.  An extended abstract appeared in ICML'19~\citep{deep19}, 
along with a short Research Highlight in the {\em Comm.~ACM}~\citep{cacm21}. The source code for all experiments is available from Github at \url{https://github.com/saisrivatsan/deep-opt-auctions}.}}}

\author[a]{Paul D\"utting}
\author[a]{Zhe Feng}
\author[a]{Harikrishna Narasimhan}
\author[b]{David C.~Parkes}
\author[b]{\\Sai Srivatsa Ravindranath}

\affil[a]{Google Research \authorcr \texttt{duetting,zhef,hnarasimham@google.com}}
\affil[b]{John A.~Paulson School of Engineering and Applied Sciences, Harvard University \authorcr \texttt{parkes,saisr@g.harvard.edu}}

\date{October 14, 2022}

\maketitle

\begin{abstract}
Designing an incentive compatible auction that maximizes expected revenue is an intricate
task. The single-item case was resolved in a seminal piece of work by
Myerson in 1981, \dcpadd{but more than 40 years later a full analytical understanding of the optimal
design still  remains elusive for settings with two or more items.}
In this work, we initiate the exploration of the use of tools from deep learning for the automated design of optimal auctions.
We model an auction as a multi-layer neural network, frame optimal auction design as a constrained learning problem, and show how it can be solved using standard
\new{machine learning}
pipelines.	In addition to providing generalization bounds, we present extensive experimental results, recovering essentially all known  solutions that come from the theoretical analysis of optimal auction design problems and obtaining novel mechanisms for settings in which the optimal mechanism is unknown.
\end{abstract}

\section{Introduction}

Optimal auction design is one of the cornerstones of economic theory.  
It is of great practical importance as auctions are used across industries and in the public sector to organize the sale of products and services. Concrete examples are the U.S.~FCC Incentive Auction, the sponsored search auctions conducted by search engines such as Google, and the auctions run on platforms such as eBay.
In the standard \emph{independent private valuations} model, each bidder has a valuation function 
over subsets of items, drawn independently from not necessarily identical distributions.
It is assumed that the auctioneer knows the value distributions and can use
this information in designing the auction. A challenge is that valuations are private, and
bidders may not  report their valuations truthfully.

In a seminal piece of work, Myerson resolved the optimal auction
design problem when there is a single item for sale \citep{Myerson81}. Today, after 40 years of intense research, there
are some elegant partial characterizations~\citep{ManelliVincent06,
  Pavlov11, HaghpanahH15,GK18,DaskalakisEtAl17,Yao17}, but the
analytical problem of optimal design is not completely resolved even
for a setting with two bidders and two items.  At the same time, there
have been impressive algorithmic
advances~\citep{CaiDW12a,CaiDW12b,CaiDW13,HartNisan17,BabaioffILW14,Yao15,CaiZhao17,ChawlaHMS10},
although most of them apply to the weaker notion of Bayesian incentive
compatibility (BIC).  Our focus in this paper is on
auctions that satisfy {\em dominant-strategy incentive compatibility}
(DSIC), which is a more robust and desirable notion of incentive
compatibility.

A recent line of work has started to bring in tools from machine
learning and computational learning theory to design auctions from
samples of bidder valuations. Much of the effort has focused on
analyzing the \emph{sample complexity} of designing revenue-maximizing
auctions~\citep{ColeR14,MohriM14,HuangMR15,MorgensternR15,GonczarowskiN17,MorgensternR16,Syrgkanis17,GonczarowskiW18,BalcanSV16}.
A handful of works has leveraged machine learning pipelines to optimize
different aspects of mechanisms \citep{Lahaie11,DuettingFJLLP12,
  Narasimhan_ijcai16}, but none of these provides the generality and
flexibility of our approach.  There have also been other
computational approaches to auction design, under the research program of
\textit{automated mechanism design}
\citep{ConitzerS02,ConitzerS04a,SandholmL15} (to which the present paper contributes), but where scalable, they
are limited to specialized classes of auctions that are already
known to be
incentive compatible.

\subsection{Our Contribution} 

In this work, we provide the first, general purpose, end-to-end
approach for solving the multi-item optimal auction design problem. We use
multi-layer neural networks to encode the rules of auction mechanisms, with bidder
valuations comprising the input to the network and an allocation and
payments comprising the output of the network. We train these neural
networks using samples from bidder value distributions and seek to
maximize expected revenue subject to constraints for incentive
compatibility.  \dcpadd{We refer to the overarching framework as that of {\em
  differentiable economics}, which references the idea of making use
of differentiable representations of economic rules. In this way, we
can  use stochastic gradient descent for economic design,
building on what is a very successful pipeline for deep
learning.}

The central technical challenge  in this work is to achieve {\em incentive compatibility}, so that bidders will report true valuations in the equilibrium of the auction.
We propose two different approaches to handling incentive
compatibility (IC) constraints. In the first, we leverage
characterization results for IC mechanisms, and constrain the network
architecture appropriately. \dcpadd{In the case of single-bidder settings, we  show how to make
  use of menu-based characterizations, which correspond to DSIC
  mechanisms.  We refer to this
  architecture as {\em RochetNet}, reflecting in its naming a
  connection with a characterization due to~\citet{Rochet87}.}

The second approach replaces the IC constraints with the requirement
of zero expected {\em ex post regret}, which is equivalent to DSIC up
to measure zero events.  For this, we make use of \emph{augmented
  Langrangian} optimization during training, which has the effect of introducing into the loss function
  penalty terms that correspond to violations of incentive
compatibility.
In this way, we minimize during training a combination of negated
revenue and a penalty term for IC violations.  We refer to this
neural network architecture as {\em RegretNet}.
This approach is applicable to multi-bidder multi-item settings for
which we do not  have tractable characterizations of IC mechanisms, but
will generally only find mechanisms that are approximately incentive
compatible.

We show through extensive experiments that these two approaches are
capable of recovering the designs of essentially all auctions for which theoretical solutions have been developed   over the past 40 years, and in the case of RegretNet, we show that the degree of approximation to DSIC is very good. We also demonstrate that this deep learning franework is  a
useful tool for  refuting hypotheses or generating supporting evidence in regard to   the conjectured
structure of optimal auctions, and that in the case of RochetNet this 
framework can 
be used to discover   designs that can then be proved to be optimal.
We also give generalization bounds that provide confidence intervals
on the expected revenue and expected ex post regret, in terms of the
empirical revenue and empirical regret achieved during training, the
descriptive complexity of the neural network used to encode the
allocation and payment rules, and the number of samples used to train
the network.

\subsection{Discussion}%
 
 While the original work on automated mechanism design (AMD) framed the problem as a linear program
(LP)~\citep{ConitzerS02,ConitzerS04a}, this has severe scalablility
issues as the formulation scales exponentially in the number of agents
and items~\citep{GuoC10}. We provide a detailed comparison with an
LP-based
framework, and find that even for a small setting with two
bidders and two items (and a discretization of bidder values into eleven bins per
item), the corresponding LP takes 62 hours to complete since the LP
needs to handle $\approx 9 \times 10^5$ decision variables and
$\approx 3.6 \times10^6$ constraints.

In comparison,  differentiable economics leverages the expressive power of neural networks and
the ability to enforce complex constraints using a standard
machine learning pipeline. This provides for optimization over 
a broad class of mechanisms without needing to resort to a discretized function representation, and is 
constrained only by the expressivity of the neural network architecture.     
For the same setting, our approach finds an auction with low regret in
just over 3.7 hours (see Table \ref{tab:regret_vs_lpa}).
\dcpadd{Moreover, the LP based approach fails to scale much beyond
  this point while the neural network-based approach continues to
  scale.}
%

 The optimization problems studied here are non-convex and  gradient-based approaches may, in general, get stuck in local optima. Empirically, however, this has not been an obstacle to the successful application of deep learning in other problem domains, and there is theoretical support for a ``no local optima'' phenomenon (see, e.g.,~\cite{ChoromanskaLA15,Kawaguchi16, pmlr-v97-du19c, pmlr-v97-allen-zhu19a}). \ssredit{We make similar observations for our experiments: our neural network architectures recover optimal solutions, wherever known, despite the formulation being non-convex.}

In the case of RegretNet, our framework only provides a guarantee of approximate DSIC. In this regard, we work  with expected ex post regret, which is a  quantifiable
relaxation of DSIC that was first introduced
in~\citep{DuettingFJLLP12}.  An essential aspect is that it quantifies the regret to bidders for truthful bidding given knowledge of the bids of others (hence ``ex post"), and thus is a quantity that measures the degree of approximation to DSIC. Indeed, our experiments suggest that this
relaxation is a very effective tool for approximating optimal DSIC
auctions, with RegretNet attaining a very good fit to known theoretical results.

\ZFadd{This work also shows that this neural-network based
  pipeline can be
  used to discover new analytical results (see
  Section~\ref{sec:dualityddt}, where we use computational results to
  guess the analytical structure of an optimal design and  duality
  theory to verify its optimality). %
}

\subsection{Further Related Work}   

\dcpadd{Since the first version of this paper, there has been considerable
follow-up work on the topic of differentiable economics,  extending
the approach to budget-constrained bidders \citep{FengEtAl18},
applying specialized architectures for single bidder settings and
using them to derive new analtyical results~\citep{STZ18}, minimizing
agent payments~\citep{Tacchetti19}, applying to multi-facility
location problems~\citep{GolowichEtAl18}, applying to two-sided
matching~\citep{sai21,feng22}, incorporating human
preferences~\citep{peri21}, balancing fairness and
revenue~\citep{kuo20}, providing certificates of
strategy-proofness~\citep{curry20a}, requiring complete
allocations~\citep{curry22}, developing permutation-equivariant
architectures~\citep{rahme21}, formulating the problem as a
two-player game between a designer and an adversary~\citep{rahme21a}, using context-integrated transformer-based neural network architecture for contextual auction design~\citep{pmlr-v162-duan22a}, and using attention mechanism through transformers for optimal auctions design~\citep{ivanov2022optimal}}. \hariadd{There has also been follow-up work on deriving sample complexity bounds for learning a Nash equilibrium \citep{zhijian+2021} using tools similar to the ones we use for our generalization bounds.} 

More recent work has adopted differentiable approaches for the design of taxation policies~\citep{zheng2022}, indirect auctions~\citep{shen2020reinforcement,brero2021reinforcement,brerolearning}, mitigations to price collusion~\cite{breromibuari22}, game design~\citep{balaguer22}, the study of platform economies~\citep{wangma22}, 
and for  multi-follower Stackelberg games~\citep{wang22}.
Deep learning has also been used to study other problems within the field of  
economics, for example using neural networks to predict the
behavior of human participants in strategic
scenarios~\citep{HartfordWL16,FudenbergLiang18,doi:10.1126/science.abe2629},
to provide an automated equilibrium analysis of
mechanisms~\citep{ThompsonEtAl17}, for causal
inference~\citep{DBLP:conf/icml/HartfordLLT17,DBLP:conf/nips/LouizosSMSZW17},
\dcpadd{and for solving for the equilibria of Stackelberg
games~\citep{wang22}, symmetric auction
games~\citep{DBLP:journals/natmi/BichlerFHKS21}}, and combinatorial
games~\citep{DBLP:journals/corr/abs-1711-02301}.
\dcpadd{The research described here also relates to the method of
{\em empirical mechanism
design}~\citep{DBLP:conf/uai/ViqueiraCMG19,DBLP:journals/aamas/VorobeychikRW12,DBLP:conf/sigecom/VorobeychikKW06,
  DBLP:conf/sigecom/BrinkmanW17}, which applies 
{\em empirical game theory} to  mechanism design, using
empirical game theory to search for the  equilibria of induced games by
building out a suitable set of candidate
strategies~\citep{DBLP:conf/atal/JordanSW10,DBLP:conf/atal/KiekintveldW08,wellman06a}; see also
more recent work on policy-space response oracles~\citep{lanctot17}.}

\subsection{Organization}

Section~\ref{sec:prelims} formulates the auction design problem as a
learning problem, introduces the charactization-based and 
characterization-free approaches, and gives the main generalization
bounds. Section~\ref{sec:regretnet} introduces the network
architectures of RochetNet and RegretNet, and instantiates the
specific generalization bound for these
networks. Section~\ref{sec:training} describes the training and
optimization procedures, and Section~\ref{sec:experiments} presents
extensive experimental results \dcpadd{including experiments that
  provide support for theoretical conjectures in regard to the design
  of optimal auctions along with the discovery of new,
  provably-optimal auction designs}. Section~\ref{sec:conclusion}
concludes.

\section{Auction Design as a Learning Problem}
\label{sec:prelims}

\subsection{Preliminaries} 

We consider a setting with a set of $n$ bidders $N= \{1,\ldots, n\}$ and $m$ 
items $M =\{1,\ldots,m\}$. Each bidder $i$ has a
{\em valuation function }
$v_i: 2^M \rightarrow \mathbb{R}_{\geq 0}$, 
where $v_i(S)$ denotes the bidder's value the subset of items $S
\subseteq M$.

In the simplest case, a bidder may have \emph{additive} valuations,
with  a value $v_i(\{j\})$ for each item $j \in M$, and a  value for a
subset of items $S \subseteq M$
that is 
	$v_i(S) = \sum_{j \in S} v_i(\{j\})$.
\new{Alternatively, if a bidder's value for a subset of items
  $S\subseteq M$ is $v_i(S) =  \max_{j\in S} v_i(\{j\})$,  the  bidder has a \emph{unit-demand} valuation.} \new{We also consider bidders with general combinatorial valuations, but defer the details to Appendix~\ref{sec:ca-architecture} and~\ref{sec:experCA}.}

Bidder $i$'s  valuation function is drawn independently from a distribution $F_i$ 
over possible valuation functions $V_i$. We write $v = (v_1, \dots, v_n)$ for a profile of 
valuations, and denote $V=\prod_{i = 1}^{n} V_i$. %
The auctioneer 
knows the distributions $F = (F_1, \dots, F_n)$, but does not know the bidders' realized 
valuation $v$. The bidders report their valuations (perhaps untruthfully), and an auction 
decides on an allocation of items to the bidders and charges a payment to them.

We denote an {\em auction} $(g,p)$ as a pair of
allocation rules $g_i: V \rightarrow 2^M$ and payment rules 
$p_i: V \rightarrow \mathbb{R}_{\geq 0}$ (these rules can be randomized).
Given bids $b = (b_1, \dots,b_n) \in V$, 
the auction computes an allocation $g(b)\in 2^M$,
and payments $p(b)\in
\mathbb{R}_{\geq 0}^n$.

A bidder with valuation $v_i$ receives 
{\em utility} \new{$u_i(v_i; b) = v_i(g_i(b)) - p_i(b)$}
at bid profile $b$.
Let $v_{-i}$ denote the valuation profile $v=(v_1,\ldots,v_n)$ without
element $v_i$, similarly for $b_{-i}$, and let
$V_{-i} = \prod_{j \neq i} V_j$ denote the possible valuation profiles
of bidders other than bidder $i$.  An auction is \emph{dominant
  strategy incentive compatible} (DSIC) if each bidder's utility is
maximized by reporting truthfully no matter what the other bidders
report. That is,
\new{$u_i(v_i; (v_i,b_{-i})) \geq u_i(v_i; (b_i,b_{-i}))$} for every
bidder $i$, every valuation $v_i \in V_i$, every bid $b_i \in V_i$,
and all bids $b_{-i} \in V_{-i}$ from others.  An auction is ex post
{\em individually rational} (IR) if each bidder receives a non-zero
utility when participating truthfully, i.e.
\new{$u_i(v_i; (v_i,b_{-i})) \geq 0$}~$\forall i \in N$,
$v_i \in V_i$, and $b_{-i} \in V_{-i}$ .

In a DSIC %
auction, it is in the best interest of each bidder to report 
truthfully, and so the equilibrium
revenue on valuation profile $v$ is simply
$\sum_i p_i(v)$. Optimal auction design
seeks to identify a DSIC auction %
that maximizes expected revenue.

\ZFadd{There is also a weaker notion
  of incentive compatibility, \emph{Bayesian Incentive Compatibility}
  (BIC). An auction is BIC if each bidder's utility is maximized by
  reporting truthfully when the other bidders also report truthfully,
  i.e. $\E_{v_{-i}}[u_i(v_i; (v_i,v_{-i}))] \geq \E_{v_{-i}}[u_i(v_i;
  (b_i,v_{-i}))]$ for every bidder $i$, every valuation $v_i \in V_i$,
  every bid $b_i \in V_i$. In this work, we focus on DSIC auctions
  rather than BIC auctions, since DSIC auctions are more preferable in
  practice because} \dcpadd{truthful bidding remains an equilibrium without
  common knowledge of the distributions on valuations or
  common knowledge on rationality.}

%
\subsection{Formulation as a Learning Problem}

We pose the problem of optimal auction design as a learning problem,
where in the place of a loss function that measures error against a
target label, we adopt  as the loss function
the negated, expected revenue on valuations
drawn from $F$.

We are given a parametric class of auctions,
$(g^w,p^w)\in \mathcal{M}$, for parameters $w\in \mathbb{R}^d$ for
some $d>0$, and a sample of $L$ bidder valuation profiles
$\mathcal{S}=\{v^{(1)},\ldots,v^{(L)}\}$ drawn i.i.d.~from $F$. There
is no need to compute equilibrium inputs; rather, we sample true
profiles, and seek to learn rules that are DSIC.  The goal is to find
an auction that minimizes the negated, expected revenue
\new{$-{\mathbf E}[\sum_{i\in N}p^w_i(v)]$}, among all auctions in
$\mathcal{M}$ that satisfy DSIC (or just IC).
\ZFadd{For a single-bidder setting, there is no difference between
  DSIC and BIC.}

We present two approaches for achieving IC. In the first, we
  leverage a characterization result to constrain the search space so
  that all mechanisms within this class are IC.
  In the second, we
  replace the IC constraints with a differentiable approximation, and
  move the constraints into the objective via the augmented Lagrangian
  method.
  The first approach affords a smaller search space and is exactly
  DSIC, but  only applies to single-bidder
  multi-item settings.  The second approach applies to multi-bidder,
  multi-item settings, but entails search through a
  larger parametric space and only achieves approximate IC.

  In Appendix~\ref{app:myerson}, we  also describe a construction based on~\citet{Myerson81}'s characterization result for multi-bidder single-item settings, which we refer to as \emph{MyersonNet}.

\subsubsection{Characterization-Based Approach}


We begin by describing our first approach, which we refer to as {\em RochetNet},  in which we exploit a
characterization of DSIC mechanisms to constrain the search space.

We describe the approach for additive valuations, but it can also be
extended to unit demand valuations.
For an additive valuation on $m$ items,
the utility function $u: \RB^m \rightarrow \R$ induced
for a single bidder by a mechanism
$(g,p)$ is,
\begin{equation}
u(v) = \sum_{j=1}^m g_j(v)\,v_j  \,-\, p(v),
\label{eq:induced-utility}
\end{equation}
where $g_j(v)\in \{0,1\}$ indicates whether or not the bidder is
assigned item $j$.

We can consider a {\em menu} of $J$ choices, for some $J\geq 1$,
where each choice
consists of a possibly randomized allocation, together with a price.
For choice
$j\in [J]$, let $\alpha_j\in [0,1]^m$ specify the randomized
allocation,
and parameter $\beta_j\in \R$ specify the negated price.
By choosing the menu item that maximizes the bidder's utility,
or the null (no allocation, no payment) outcome when this is better,
a menu
of
size $J$ induces the
following utility function:
\begin{equation}
u(v) \,=\, \max\left\{\max_{j \in [J]}\, \{ \alpha_j \cdot v \,+\, \beta_j\},\, 0\right\}.
\label{eq:induced-utility-max}
\end{equation}
%
%
%

The well known {\em taxation principle} from mechanism design theory
tells us that a mechanism that selects the menu choice that maximizes an agent's
reported utility, based on its bid $b\in \R^m$, is DSIC~\citep{hammond79,guesnerie95}. 
 To see this,
observe that the menu does not depend on the reports, and that the
agent will maximize its utility by reporting its true valuation
function so that the right choice is made on its behalf.  Moreover,
the taxation principle also tells us
that the use of a menu is without loss of generality
for DSIC mechanisms.  

Based on this, for a given $J\geq 0$, we seek to learn a mechanism 
with parameters $w = (\alpha, \beta)$, where $\alpha\in [0,1]^{mJ}$
and $\beta\in \R^J$, to maximize the expected revenue
$\E_{v \sim F}[\beta_{j^*(v)}]$,
where $j^*(v) \in \text{argmax}_{j \in [J]} \{ \alpha_j \cdot v \,+\,
\beta_j\}$,
and  denotes the best choice for the bidder, where choice
$0$ corresponds to the null outcome.
%
%
For a unit-demand bidder, the utility can also be
    represented via (\ref{eq:induced-utility}), with the additional
    constraint that $\sum_j g_j(v)\leq 1, \forall v$. We discuss this
    more in Section~\ref{sec:rochetnet-architecture}.


    We also have the following characterization of DSIC mechanisms for
    the single bidder case.
\begin{thm}[\citet{Rochet87}]\label{theorem:u-convex}
	The utility function $u: \RB^m \rightarrow \R$ that is induced by a
	DSIC mechanism for a single biddder 
        is 1-Lipschitz w.r.t.\ the $\ell_1$-norm, 
	non-decreasing, and convex. 
\end{thm}

 The convexity
    can be understood by recognizing that
the induced utility function~\eqref{eq:induced-utility-max} is the maximum over a set of hyperplanes,
each corresponding to a choice in the menu set. 
Figure~\ref{fig:monotonic_networkb} illustrates Rochet's theorem
for a single item
($m=1$) and a menu consisting of four choices ($J=4$). Here,
the induced
utility for choice $j$ given bid $b\in \R$ is
$h_j(b) \,=\, \alpha_j \cdot b \,+\, \beta_j$. 
\begin{figure}[t]
\centering
\scalebox{1}{
\centering
\begin{tikzpicture}[scale=0.7,shorten >=4pt]
\draw [dotted] (1,0) -- (6,1.25);
\draw [dotted] (0,0.15) -- (4,0.35);
\draw [dotted] (4.3,0.3) -- (7.3,3.3);
\draw [dotted] (6.1,1.3) -- (7,4);
\draw [dotted,line width=0.25mm] (0,0.75) -- (7,0.75);
\draw [dashed,line width=0.25mm	] (0,0.75) -- (4,0.75);
\draw (0,0.15) -- (2,0.25) -- (5,1) -- (6.5,2.5) -- (7,4);
\draw [->] (-0.25,0) -- (-0.25,4.5);
\draw [->] (-0.25,0) -- (7.25,0);
\node at (3.75,-0.25) {$b$};
\node at (-0.7,2.2) {$u(b)$};
\node at (6.5,0.45) {$u(b)=0$};
\node at (1.15,0.45) {$h_1$};
\node at (4.25,1.10) {$h_2$};
\node at (5.5,1.95) {$h_3$};
\node at (6.45,3.25) {$h_4$};
\end{tikzpicture}}
\vspace*{-10pt}
\caption{An induced  utility function represented by RochetNet for the
  case of a single item  ($m=1$) and menu with four
  choices $(J=4)$. \label{fig:monotonic_networkb}}
\vspace*{-8pt}
\end{figure}


Given this, to find the
optimal single-bidder auction we can search over a suitably sized menu
set and pick the one that maximizes expected revenue.
In Section~\ref{sec:rochetnet-architecture} we explain how to achieve
this by modeling the utility function as a neural network, and
formulating the above optimization as a
differentiable learning problem.

\subsubsection{Characterization-Free Approach}

Our second approach---which we refer to as \emph{RegretNet}---does not
require a characterizatio of IC. Instead, it replaces the IC
constraints with a differentiable approximation and brings the IC
constraints into the objective by augmenting the objective with a term
that accounts for the extent to which the IC constraints are violated.

We measure the extent to which an auction violates IC through the
following notion of ex post regret.  Fixing the bids of others, the
{\em ex post regret} for a bidder is the maximum increase in her
utility, considering all possible non-truthful bids. For a mechanism
($g^{w}, p^{w}$), with parameters $w$,
we will be interested in the \emph{expected ex post
  regret} for bidder $i$:
\begin{eqnarray*}
	\mathit{rgt}_i(w) =~
	\mathbf{E}\Big[\max_{v'_i \in V_i}\,u^w_i(v_i; (v'_i, v_{-i})) - u^w_i(v_i;(v_i, v_{-i}))\Big],
\end{eqnarray*}
where  the expectation is over $v \sim F$ and
$u^w_i(v_i;b) = v_i(g^w_i(b)) - p^w_i(b)$ for model parameters $w$.

We assume that $F$ has full support on the space of valuation profiles
$V$.
Given this, 
and recognizing that the regret is non-negative, an auction satisfies
DSIC if and only if $\mathit{rgt}_i(w) = 0, \forall i \in N$, except
for measure zero events.\footnote{In this work, we focus on DSIC, but the
RegretNet can also be adapted to handle BIC~\citep{FengEtAl18}.}

Given this, we re-formulate the learning problem as one of minimizing
the expected negated revenue subject to the expected ex post regret
being zero for each bidder:
\begin{eqnarray*}
	\min_{w \in \mathbb{R}^d}\, &\E_{v\sim F}\bigg[-\sum_{i\in N}p^w_i(v)\bigg]\\
	\text{s.t.\;}\, &{rgt}_i(w) \,=\, 0, ~\forall i\in N.
	\label{eq:ml-detailed1}
\end{eqnarray*}

Given a sample $\mathcal{S}$ of $L$ valuation profiles from $F$, 
we estimate the empirical {ex post }regret for bidder $i$ as: 
\begin{align}
\widehat{\mathit{rgt}}_i(w) =%
&
\frac{1}{L}\sum_{\ell=1}^L
\Big[\max_{v'_i \in V_i}\,u^w_i\big(v_i^{(\ell)}; \big(v'_i, v^{(\ell)}_{-i}\big)\big) - u^w_i(v_i^{(\ell)}; v^{(\ell)})\Big],
\label{eq:emp-rgt}
\end{align}
and  seek to minimize the
{empirical loss} (negated revenue) subject to the empirical regret
being zero for all bidders, and the following formulation:
\begin{eqnarray}
\min_{w \in \mathbb{R}^d}\, &-\frac{1}{L}\sum_{\ell=1}^L \sum_{i=1}^n p^w_i(v^{(\ell)})
\nonumber\\
\text{s.t.\;}\, &\widehat{rgt}_i(w) \,=\, 0, ~~\forall i\in N.
\label{eq:ml-detailed2}
\end{eqnarray}

\new{We additionally} require the auction to satisfy IR, which can be ensured by restricting the search space to a class of parametrized auctions that charge no bidder more than her
\new{valuation} for an allocation. 

In Section~\ref{sec:regretnet}, we model the allocation and
  payment rules through a neural network, and incorporate the IR
  requirement within the architecture. In Section~\ref{sec:training}
  we describe how the IC constraints can be incorporated into the
  objective using Lagrange multipliers, so that the resulting neural
  net can be trained with standard pipelines.

\subsection{Quantile-Based Regret}

The intent is that the characterization-free approach leads to mechanisms
with low expected ex post regret.  By seeking to minimize the expected
ex post regret, we can also obtain regret bounds of the form ``the
probability that the ex post regret is larger than $x$ is at most
$q$.''  For this, we define quantile-based ex post regret.
\begin{definition}[Quantile-based ex post regret]
	For each bidder $i$, and $q$ with \new{$0<q<1$}, the $q$-quantile-based ex post regret, $\mathit{rgt}^q_i(w)$, induced by the probability distribution $F$ on valuation profiles, is defined as the smallest $x$ such that
	\begin{align*}
	\PP\left(\max_{v'_i \in V_i}\,u^w_i(v_i; (v'_i, v_{-i})) - u^w_i(v_i;(v_i, v_{-i})) \new{\;\geq x} \right) \new{\leq q.}
	\end{align*} 
\end{definition}

We can bound the $q$-quantile based regret $\mathit{rgt}^q_i(w)$ by the expected ex post regret $\mathit{rgt}_i(w)$ as in the following {lemma}. The proof appears in Appendix~\ref{app:quantile-regret}.

\begin{lem}\label{lem:quantile-regret}
	For any fixed $q$, $0 < q < 1$, and bidder $i$, we can bound the $q$-quantile-based ex post regret by
	\begin{align*}
	\mathit{rgt}^q_i(w) \leq \frac{\mathit{rgt}_i(w)}{q}.
	\end{align*}
\end{lem}

\new{Using this Lemma~\ref{lem:quantile-regret}, we can show,  for example, that when the expected ex post regret is $0.001$, then the probability that the ex post regret exceeds $0.01$ is at most $10\%$.}

\subsection{Generalization Bounds}
\label{sec:generalization-bound}

We conclude this section with two generalization bounds.
  The first is a lower   bound on  the expected revenue 
  in terms of the empirical revenue during training, the complexity (or capacity) of
  the auction class that we optimize over, and the number of sampled
  valuation profiles.
  The  second is an  upper bound on 
  the expected ex post regret
  in terms of the 
  empirical regret during training, the complexity (or capacity) of
  the auction class that we optimize over, and the number of sampled
  valuation profiles.

  We measure the capacity of an auction class $\mathcal{M}$ using a
  definition of \emph{covering numbers} from the ranking literature
  \citep{Rudin09a}. %
  \new{For this, define the} $\ell_{\infty,1}$ distance between auctions
  $(g,p), (g',p') \in \M$ as
\[\new{\max_{v \in V} \sum_{i \in N, j \in M}|g_{ij}(v) - g'_{ij}(v)| + \sum_{i \in N} |p_i(v) - p'_i(v)|.}\] 
For any $\epsilon > 0$, let $\cN_\infty(\M, \epsilon)$ be the {minimum} number of balls of radius $\epsilon$ required to cover $\M$ under the $\ell_{\infty,1}$ distance.  
~\\[-10pt]
\begin{thm}
	\label{THM:GBOUND}
	For each bidder $i$, assume 
	that the valuation function $v_i$ satisfies $v_i(S) \leq 1,\, \forall S \subseteq M$. Let $\M$ be a class of auctions that satisfy individual rationality. Fix $\delta \in (0,1)$.  With probability at least $1-\delta$ over draw of sample $\mathcal{S}$ of $L$ profiles from $F$, for any $(g^w,  p^w) \in \M$, %
	~\\[-11pt]
	\begin{eqnarray*}
		&\hspace*{-7pt} \new{\E_{v\sim F}\bigg[\sum_{i\in N}p^w_i(v)\bigg]}%
		\new{ \,\geq\, \frac{1}{L}\sum_{\ell=1}^L \sum_{i=1}^n p^w_i(v^{(\ell)})}
		\new{\,-\, 2n{\Delta_{L}}}
		\new{\,-\, {Cn}\sqrt{\frac{\log(1/\delta)}{L}},} \\[-12pt]
	\end{eqnarray*}
	and
	\begin{eqnarray*}
		\frac{1}{n}\sum_{i=1}^n rgt_i(w) \,\leq\, \frac{1}{n}\sum_{i=1}^n \widehat{rgt}_i(w) \,+\,  
		2\Delta_{L}
		\,+\,  C'\sqrt{\frac{\log(1/\delta)}{L}},
	\end{eqnarray*}
	where $\Delta_{L}= \inf_{\epsilon>0}\Big\{\frac{\epsilon}{n} + 
	\,2\sqrt{\frac{2\log(\cN_\infty(\M, \,\frac{\epsilon}{2n}))}{L}}\Big\}$ and $C, C'$ are distribution-independent constants. 
\end{thm}
See \new{Appendix~\ref{app:gbound}} for the proof. If the term
$\Delta_{L}$ in the above bound goes to zero as the sample size $L$
increases then the above bounds go to zero as $L \rightarrow \infty$.
In Theorem~\ref{thm:cover_regretnet} in
Section~\ref{sec:regretnet}, we bound $\Delta_{L}$ for the neural
network architectures we present in this work.

\section{Neural Network Architectures}
\label{sec:regretnet}

We describe the RochetNet architecture for single-bidder multi-item
settings in Section~\ref{sec:rochetnet-architecture}, and the
RegretNet architecture for multi-bidder multi-item settings in
Section~\ref{sec:regretnet-architecture}. We focus on additive
valuations and unit-demand valuations, and discuss how to extend the
constructions to allow for combinatorial valuations in
Appendix~\ref{sec:ca-architecture}.

\subsection{The RochetNet Architecture}\label{sec:rochetnet-architecture}

   RochetNet  operationalizes the
idea of menu-based mechanisms through a suitable neural network architecture. 
\ZFadd{We first describe the construction for additive valuations
  and then explain how to extend it  to unit-demand valuations.}
The parameters correspond to  a menu of
$J$ choices, where each choice $j\in [J]$ is associated with
randomized allocation $\alpha_j\in [0,1]^m$ and negated price
$\beta_j\in \R$ ($\beta_j$s will be negative, and the smaller the
value of $\beta_j$, the larger the payment).  The network selects the
choice for the bidder that maximizes the bidder's reported utility
given its bid, or chooses the null outcome (no allocation, no payment)
when this is preferred. This ensures DSIC and
IR.

\begin{figure}[t]
\centering
\scalebox{1}{
\centering
\begin{tikzpicture}[scale=0.85, transform shape, shorten >=1pt,->,draw=black!100, node distance=\layersep, thick]
    \tikzstyle{input text}=[draw=white,minimum size=17pt,inner sep=0pt]
    \tikzstyle{input neuron}=[circle,draw=black!100,minimum size=17pt,inner sep=0pt,thick]
    \tikzstyle{hidden neuron}=[circle,draw=black!100,minimum size=17pt,inner sep=0pt,thick]
    \tikzstyle{hidden text}=[draw=white,minimum size=22pt,inner sep=0pt,thick]
    \tikzstyle{unit}=[draw=black!100,minimum size=22pt,inner sep=0pt,thick]

    \foreach \name / \y in {1,...,2}
    	\node[input neuron, pin={[pin edge={<-}]left:$b_\y$}] (I-\name) at (0,-\y) {};
    \node[input text] (I-0) at (0,-3) {$\vdots$};
    \node[input neuron, pin={[pin edge={<-}]left:$b_m$}] (I-3) at (0,-4) {};

    \path node[hidden neuron] (H-1) at (2,-0.5) {$h_1$};
    \path node[hidden neuron] (H-2) at (2,-2) {$h_2$};
    \path node[hidden text] (H-0) at (2,-3.25) {\vdots};
    \path node[hidden neuron] (H-3) at (2,-4.5) {$h_J$};
	\path node[hidden neuron] (H-4) at (4,-1) {$0$};

    \foreach \src in {1,...,3}
    	\foreach \dest in {1,...,3}
			\path (I-\src) edge (H-\dest); 

	\node[unit, pin={[pin edge={->}]right:$u(b)$}] (M) at (4,-2.5) {max};
	
	\foreach \src in {1,...,4}
		\path (H-\src) edge (M); 
\end{tikzpicture}}
\vspace*{-10pt}
\caption{RochetNet:  Neural network representation of a non-negative, monotone, convex induced utility function; here $h_j(b) \,=\, \alpha_j \cdot b \,+\, \beta_j$ for $b\in \R^m$ and $\alpha_j\in [0,1]^m$. 
\label{fig:monotonic_networka}}
\vspace*{-8pt}
\end{figure}

The utility function, represented as a 
single layer neural network, 
is illustrated in Figure~\ref{fig:monotonic_networka}, where each $h_j(b) \,=\, \alpha_j \cdot b \,+\, \beta_j$ for bid $b\in \R^m$.  The input layer takes a bid
$b\in \R^m$ and the output of the network is the induced utility.
%
%
%
For input $b$, 
$j^\ast(b) \in \text{argmax}_{j \in [J]\cup \{0\}} \{ \alpha_j \cdot b
\,+\, \beta_j\}$ denotes the best choice for the bidder, where choice
$0$ corresponds to $\alpha_0=0$ and $\beta_0=0$ and the null
outcome. This best choice defines the allocation and payment rule: for
bid $b$, the allocation is $g^w(b)=\alpha_{j^\ast(b)}$ and the payment
is $p^w(b)=-\beta_{j^\ast(b)}$.

\ZFadd{By using a large number of hyperplanes, one can use this
  neural network architecture to search over a sufficiently rich class
  of DSIC and IR auctions for the single-bidder, multi-item setting.
  Given the RochetNet construction, we seek to minimize the negated,
  expected revenue,
$\E_{v \sim F}[\beta_{j^*(v)}]$.}
%
%
%
%
%
%
%
%
%
%
To ensure that the objective is a continuous function of parameters
$\alpha$ and $\beta$, we adopt during training
a {\em softmax} operation in place of
the argmax, and the following loss function:
\begin{equation}
  \label{eq:loss-rochetnet}
\new{\mathcal{L}(\alpha,\beta) \,=\, -\E_{v \sim F}\big[\sum_{j\in [J]}\beta_j \widetilde{\nabla}_j(v)\big],}
\end{equation}
where
\begin{equation}
  \label{eq:rochetnet-gradient}
\new{\widetilde{\nabla}_j(v) \,=\, \mathrm{softmax}_j\big(\kappa\cdot(\alpha_1 \cdot v + \beta_1), \ldots, \kappa\cdot(\alpha_J \cdot v + \beta_J)\big)},
\end{equation}
and $\kappa > 0$ is a constant that controls the quality of the
approximation.  Here, the softmax function,
$\mathrm{softmax}_j(\kappa x_1,\ldots,\kappa x_J)=e^{\kappa
  x_j}/\sum_{j'}e^{\kappa x_{j'}}$, takes as input $J$ real numbers
and returns a probability distribution consisting of $J$
probabilities, proportional to the exponential of the inputs.
\ZFadd{We only do this approximation during training,
  and always use argmax during
   testing to guarantee the mechanism is DSIC.}

During training, we seek to optimize the parameters of the neural network, i.e.,
$\alpha \in [0,1]^{mJ}$, and $\beta \in \R^{J}$, to minimize loss~\eqref{eq:loss-rochetnet}.
%
%
For this,
given a sample $\mathcal{S} =\{v^{(1)}, \ldots, v^{(L)}\}$ drawn from
$F$,  we use stochastic gradient descent to
optimize an empirical version of the loss.
%
%

This approach easily extends to a
\new{single} bidder with a \new{unit-demand} valuation.
In this case, the new requirement
is that the sum of the allocation probabilities cannot exceed one. 
This can be  enforced by restricting the coefficients for each
hyperplane to sum up to at most \new{one}, i.e.\ $\sum_{k=1}^m \alpha_{jk} \leq 1,
\forall j\in [J]$, and $\alpha_{jk} \geq 0, \forall j\in J, k\in
[m]$.  To achieve this contraint, we can re-parameterize $\alpha_{jk}$
as $\mathrm{softmax}_k\big(\gamma_{j1}, \cdots, \gamma_{jm}\big)$,
where $\gamma_{jk}\in \R, \forall j\in J, k\in m$.
With this restriction, the resulting mechanism is DSIC for unit-demand
bidders since the selected menu choice
 corresponds a distribution over
single-item allocations.\footnote{In follow-up work,~\cite{STZ18} extend the RochetNet
  architecture to more general settings, including settings with
  non-linear utility functions.}

%
\subsection{The RegretNet Architecture} \label{sec:regretnet-architecture}

We next describe the architecture for the characterization-free,
RegretNet approach. In this case, we train a neural network that
explicitly encodes a multi-bidder allocation and payment rule. The
architecture consists of two logically distinct components that
comprise part of a single network: the allocation component and the
payment component. These are trained together as a single
network, and the outputs of these networks are used to compute the
regret and revenue, and thus quantities used by the loss function.

\subsubsection{Additive Valuations}

\begin{figure*}[t]
	\vspace{-2pt}
	\centering
\begin{tikzpicture}[scale=0.6,transform shape, shorten >=1pt,->,draw=black!100, node distance=\layersep, thick]
   \def\x{0}     
   \def\y{9}   
   
    \tikzstyle{input text}=[draw=white,minimum size=22pt,inner sep=0pt]
    \tikzstyle{input neuron}=[circle,draw=black!100,minimum size=17pt,inner sep=0pt,thick]
    \tikzstyle{hidden neuron1}=[draw=black!100,minimum size=25pt,inner sep=0pt,thick]
    \tikzstyle{hidden neuron}=[circle,draw=black!100,minimum size=25pt,inner sep=0pt,thick]
    \tikzstyle{hidden text}=[draw=white,minimum size=22pt,inner sep=0pt,thick]
    \tikzstyle{unit}=[circle,draw=black!100,minimum size=20pt,inner sep=0pt,thick]

    \node[input text] (I-0) at (0,-0.65) {$\vdots$};
    \node[input text] (I-0) at (0,-2.45) {$\vdots$};
    \node[input text] (I-0) at (0,-4.2) {$\vdots$};
    \node[input neuron, pin={[pin edge={<-}]left:$b_{11}$}] (I-1) at (\x,0) {};
    \node[input neuron, pin={[pin edge={<-}]left:$b_{1m}$}] (I-2) at (\x,-1.5) {};
    \node[input neuron, pin={[pin edge={<-}]left:$b_{n1}$}] (I-3) at (\x,-3.5) {};
	\node[input neuron, pin={[pin edge={<-}]left:$b_{nm}$}] (I-4) at (\x,-5) {};

    \path node[hidden neuron] (H-1) at (\x+1.5,-0.5) {$h^{(1)}_1$};
    \path node[hidden neuron] (H-2) at (\x+1.5,-2) {$h^{(1)}_2$};
    \path node[hidden text] (H-0) at (\x+1.5,-3.25) {\vdots};
    \path node[hidden neuron] (H-3) at (\x+1.5,-4.5) {$h^{(1)}_{J_1}$};

    \path node[hidden neuron] (J-1) at (\x+3,-0.5) {$h^{(R)}_1$};
    \path node[hidden neuron] (J-2) at (\x+3,-2) {$h^{(R)}_2$};
    \path node[hidden text] (J-0) at (\x+3,-3.25) {\vdots};
    \path node[hidden neuron] (J-3) at (\x+3,-4.5) {$h^{(R)}_{J_{R}}$};

    \node[input text] (S-0) at (\x+4.5,-0.65) {$\vdots$};
    \node[input text] (S-0) at (\x+4.5,-2.45) {$\vdots$};
    \node[input text] (S-0) at (\x+4.5,-4.2) {$\vdots$};
	\path node[unit, pin={[pin edge={->}]right:$z_{11}$}] (S-1) at (\x+4.5,-0) {};
    \path node[unit, pin={[pin edge={->}]right:$z_{n1}$}] (S-2) at (\x+4.5,-1.5) {};
    \path node[unit, pin={[pin edge={->}]right:$z_{1m}$}] (S-3) at (\x+4.5,-3.5) {};
    \path node[unit, pin={[pin edge={->}]right:$z_{nm}$}] (S-4) at (\x+4.5,-5) {};

    \foreach \src in {1,...,4}
    	\foreach \dest in {1,...,3}
			\path (I-\src) edge (H-\dest); 
			
	\foreach \src in {1,...,3}
		\foreach \dest in {1,...,4}
			\path (J-\src) edge (S-\dest); 
	
	\node at  ($(H-2.south east) + (0.45,0)$) {\ldots};

	\draw[thick,dashed] ($(S-1.north west)+(-0.2,0.3)$)  rectangle ($(S-2.south east)+(0.2,-0.3)$) ;
	\node (N-1) at  ($(S-1.north west) + (0.3,0.5)$) {$\mathit{softmax}$};

	\draw[thick,dashed] ($(S-3.north west)+(-0.2,0.3)$)  rectangle ($(S-4.south east)+(0.2,-0.3)$) ;
	\node at  ($(S-4.south west) + (0.3,-0.5)$) {$\mathit{softmax}$};
	
	\draw[thick,dotted] ($(S-1.north west)+(-0.2,0.3)$)  rectangle ($(S-4.south east)+(0.2,-0.3)$) ;

    \tikzstyle{input text}=[draw=white,minimum size=22pt,inner sep=0pt]
    \tikzstyle{input neuron}=[circle,draw=black!100,minimum size=17pt,inner sep=0pt,thick]
    \tikzstyle{hidden neuron1}=[draw=black!100,minimum size=25pt,inner sep=0pt,thick]
    \tikzstyle{hidden neuron}=[circle,draw=black!100,minimum size=25pt,inner sep=0pt,thick]
    \tikzstyle{hidden text}=[draw=white,minimum size=22pt,inner sep=0pt,thick]
    \tikzstyle{unit}=[draw=black!100,minimum size=20pt,inner sep=0pt,thick]

    \node[input text] (I-0) at (\y,-0.65) {$\vdots$};
    \node[input text] (I-0) at (\y,-2.45) {$\vdots$};
    \node[input text] (I-0) at (\y,-4.2) {$\vdots$};
    \node[input neuron, pin={[pin edge={<-}]left:$b_{11}$}] (I-1) at (\y,0) {};
    \node[input neuron, pin={[pin edge={<-}]left:$b_{1m}$}] (I-2) at (\y,-1.5) {};
    \node[input neuron, pin={[pin edge={<-}]left:$b_{n1}$}] (I-3) at (\y,-3.5) {};
	\node[input neuron, pin={[pin edge={<-}]left:$b_{nm}$}] (I-4) at (\y,-5) {};

    \path node[hidden neuron] (H-1) at (\y+1.5,-0.5) {$c^{(1)}_1$};
    \path node[hidden neuron] (H-2) at (\y+1.5,-2) {$c^{(1)}_2$};
    \path node[hidden text] (H-0) at (\y+1.5,-3.25) {\vdots};
    \path node[hidden neuron] (H-3) at (\y+1.5,-4.5) {$c^{(1)}_{J'_1}$};

    \path node[hidden neuron] (J-1) at (\y+3,-0.5) {$c^{(T)}_1$};
    \path node[hidden neuron] (J-2) at (\y+3,-2) {$c^{(T)}_2$};
    \path node[hidden text] (J-0) at (\y+3,-3.25) {\vdots};
    \path node[hidden neuron] (J-3) at (\y+3,-4.5) {$c^{(T)}_{J'_{T}}$};

    \path node[unit] (S-1) at (\y+4.5,-0.5) {$\sigma$};
    \path node[unit] (S-2) at (\y+4.5,-2) {$\sigma$};
    \path node[hidden text] (S-0) at (\y+4.5,-3.25) {\vdots};
    \path node[unit] (S-3) at (\y+4.5,-4.5) {$\sigma$};

    \path node[hidden text] (T-1) at (\y+6,-0.5) {$\tilde{p}_1$};
    \path node[hidden text] (T-2) at (\y+6,-2) {$\tilde{p}_2$};
    \path node[hidden text] (T-0) at (\y+6,-3.25) {\vdots};
    \path node[hidden text] (T-3) at (\y+6,-4.5) {$\tilde{p}_n$};

    \path node[hidden neuron, pin={[pin edge={->}]right:$\displaystyle {p}_1 = \tilde{p}_1\sum_{j=1}^m z_{1j}\,b_{1j}$}] (M-1) at (\y+7.5,-0.5) {$\times$};
    \path node[hidden neuron, pin={[pin edge={->}]right:$\displaystyle{p}_2 = \tilde{p}_2\sum_{j=1}^m z_{2j}\,b_{2j}$}] (M-2) at (\y+7.5,-2) {$\times$};
    \path node[hidden text] (M-0) at (\y+7.5,-3.25) {\vdots};
    \path node[hidden neuron, pin={[pin edge={->}]right:$\displaystyle {p}_n = \tilde{p}_n\sum_{j=1}^m z_{nj}\,b_{nj}$}] (M-3) at (\y+7.5,-4.5) {$\times$};

    \foreach \src in {1,...,4}
    	\foreach \dest in {1,...,3}
			\path (I-\src) edge (H-\dest); 
			
	\foreach \src in {1,...,3}
		\foreach \dest in {1,...,3}
			\path (J-\src) edge (S-\dest); 
			
	\foreach \src in {1,...,3}
		\path (S-\src) edge (T-\src); 
	
	\foreach \src in {1,...,3}
		\path (T-\src) edge (M-\src); 
	
	\node at  ($(H-2.south east) + (0.45,0)$) {\ldots};

   	 \node[input text] (temp-1) at (\x+5,-0.85) {}; 
   	 \draw (temp-1) edge[bend right=-25,solid,line width=0.6mm,dotted,color=gray] node[above,color=black]{$z_{11},\ldots,z_{1m}$} (M-1); 
	
	 \node[hidden text] (temp-2) at (\y+6.6,0.5) {$\mathbf b$};
	 \draw[<-,solid,line width=0.5mm,dotted,color=gray] (M-1)  -- (\y+6.7,0.43); 
	 
	 \node[input text] (temp-3) at (\x+5,-2.5) {}; 
   	 \path (temp-3) edge[bend right=15,solid,line width=0.6mm,dotted,color=gray]  (M-2);
   	 
   	 \node[hidden text] (temp-4) at (\y+6.6,-1.05) {$\mathbf b$};
	 \draw[<-,solid,line width=0.5mm,dotted,color=gray] (M-2)  -- (\y+6.7,-1.13); 
	 
	 \node[input text] (temp-5) at (\x+5,-4.2) {}; 
   	 \path (temp-5) edge[bend right=25,solid,line width=0.6mm,dotted,color=gray]  
node[auto,color=black]{$z_{n1},\ldots,z_{nm}$}    	 
   	 (M-3);
   	 
   	 \node[hidden text] (temp-6) at (\y+6.6,-3.60) {$\mathbf b$};
	\draw[<-,solid,line width=0.6mm,dotted,color=gray] (M-3)  -- (\y+6.7,-3.68);

    \draw [-,dashed] (6.5,1.5) -- (6.5,-6.5);
    
    \node at  (2.7,1.5) {\textbf{Allocation Network} $g$};	 
    \node at  (13.2,1.5) {\textbf{Payment Network} $p$};

    \node[unit] (metrics) at (6.5,-7.5) {~\textit{Metrics}: $rev$, $rgt_1,\ldots,rgt_n$~~~};
	\node (alloc) at (\x+2,-5.9) {};
	\path (alloc) edge[bend right=25,solid,line width=1.5mm]  node[below]{$w_g$~~~} (metrics); 	 
	\node (pay) at (\x+11,-5.9) {};
	\path (pay) edge[bend left=25,solid,line width=1.5mm] node[below]{~~~$w_p$}  (metrics); 	 
		 
\end{tikzpicture}
	\vspace{-3pt}
	\caption{The allocation component and payment
          component of the RegretNet neural network for a setting with
          $n$ additive bidders and $m$ items. The inputs are bids from
          each bidder for each item. The revenue $rev$ and expected ex
          post $rgt_i$ are defined as a function of the parameters of
          the allocation component and payment component $w = (w_g, w_p)$.}
	\vspace{-5pt}
	\label{fig:gen-net}
\end{figure*}  

\new{An overview of the RegretNet architecture for additive valuations is given in Figure~\ref{fig:gen-net}.}
The allocation component encodes a randomized allocation rule $g^w:  \R^{nm} \rightarrow [0,1]^{nm}$ and the payment component encodes a payment rule $p^w:  \R^{nm} \rightarrow \R_{\geq 0}^{n}$, both of which are modeled as feed-forward, fully-connected networks with
\new{a tanh activation function in each of the hidden nodes}. 
The input layer consists of bids \new{$b_{ij} \geq 0$} representing the valuation of bidder $i$ for item $j$. 

The allocation component outputs a vector of
allocation probabilities $z_{1j} = g_{1j}(b), \ldots, z_{nj} = g_{nj}(b)$, for each item
$j\in [m]$. To ensure feasibility, i.e., that the probability of an item being allocated is at most one, the allocations are computed using
a softmax activation function, so that for all items $j$, we have $\sum_{i=1}^n z_{ij} \leq 1$. To accommodate the possibility of an item not being assigned, we include a dummy node in the softmax computation to hold the residual allocation probability. The payment component outputs a payment for each bidder that denotes the amount the bidder should pay in expectation for a particular bid profile.  

To ensure that the auction satisfies \textit{ex post IR}, i.e., does
not charge a bidder more than her expected value for the allocation,
the network first computes a normalized payment
$\tilde{p}_i \in [0,1]$ for each bidder $i$ using a sigmoidal unit,
and then outputs a payment
$p_i = \tilde{p}_{i}(\sum_{j=1}^m z_{ij}\, b_{ij})$, where \new{the}
$z_{ij}$'s are \new{the} outputs from the allocation component.
\ZFadd{This guarantees ex post IR, since the payment can be
  represented as a distribution over payments for each allocation in
  the support of the randomized allocation, where each payment is at
  most the bidder's reported value for that allocation.}

\subsubsection{Unit-Demand Valuations}\label{sec:unit-demand-architecture}

The allocation component  for unit-demand bidders is the feed-forward network shown \new{in Figure \ref{fig:regretnet-ud}}. For revenue maximization in this setting, it is sufficient to
consider allocation rules that assign at most one item to each
bidder.\footnote{
	\color{black}This holds by a simple reduction argument: 
	for any IC auction that allocates multiple items, one can construct an IC auction with the same  revenue  by retaining only the most-preferred  item among those allocated to a bidder.}
      In the case of randomized allocation rules, this requires that
      the total allocation probability to each bidder is at most one,
      i.e., $\sum_{j}z_{ij} \leq 1, ~\forall i \in [n]$. We would also
      require that no item is over-allocated, i.e.,
      $\sum_{i}z_{ij}\leq 1, ~\forall j \in [m]$. Hence, we design the
      allocation component such that the matrix of output
      probabilities $[z_{ij}]_{i,j = 1}^n$ is doubly
      stochastic.\footnote{\new{This is a slightly
          more general definition for
          doubly stochastic than is typical. Doubly stochastic is more
          typically defined on a square matrix with the sum of rows
          and the sum of columns equal to 1.}}

\begin{figure*}[t]
	\begin{subfigure}[b]{0.99\textwidth}
		\centering
\begin{tikzpicture}[scale=0.7,transform shape, shorten >=1pt,->,draw=black!100, node distance=\layersep, thick]
    \tikzstyle{input text}=[draw=white,minimum size=22pt,inner sep=0pt]
    \tikzstyle{input neuron}=[circle,draw=black!100,minimum size=17pt,inner sep=0pt,thick]
    \tikzstyle{hidden neuron1}=[draw=black!100,minimum size=25pt,inner sep=0pt,thick]
    \tikzstyle{hidden neuron}=[circle,draw=black!100,minimum size=25pt,inner sep=0pt,thick]
    \tikzstyle{hidden text}=[draw=white,minimum size=22pt,inner sep=0pt,thick]
    \tikzstyle{unit}=[circle,draw=black!100,minimum size=20pt,inner sep=0pt,thick]

  \node[input text] (I-0) at (0,-0.65) {$\vdots$};
    \node[input text] (I-0) at (0,-2.45) {$\vdots$};
    \node[input text] (I-0) at (0,-4.2) {$\vdots$};
    \node[input neuron, pin={[pin edge={<-}]left:$b_{11}$}] (I-1) at (0,0) {};
    \node[input neuron, pin={[pin edge={<-}]left:$b_{1m}$}] (I-2) at (0,-1.5) {};
    \node[input neuron, pin={[pin edge={<-}]left:$b_{n1}$}] (I-3) at (0,-3.5) {};
	\node[input neuron, pin={[pin edge={<-}]left:$b_{nm}$}] (I-4) at (0,-5) {};    

    \path node[hidden neuron] (H-1) at (1.5,-0.5) {$h^{(1)}_1$};
    \path node[hidden neuron] (H-2) at (1.5,-2) {$h^{(1)}_2$};
    \path node[hidden text] (H-0) at (1.5,-3.25) {\vdots};
    \path node[hidden neuron] (H-3) at (1.5,-4.5) {$h^{(1)}_{J_1}$};

    \path node[hidden neuron] (J-1) at (3,-0.5) {$h^{(R)}_1$};
    \path node[hidden neuron] (J-2) at (3,-2) {$h^{(R)}_2$};
    \path node[hidden text] (J-0) at (3,-3.25) {\vdots};
    \path node[hidden neuron] (J-3) at (3,-4.5) {$h^{(R)}_{J_{R}}$};

    \node[input text] (S-0) at (4.5,-0.65) {$\vdots$};
    \node[input text] (S-0) at (5.5,-0.8) {$\ldots$};
    \node[input text] (S-0) at (5.5,-3.5) {$\ldots$};
 
    \node[input text] (S-0) at (6.5,-0.65) {$\vdots$};
    \node[input text] (S-0) at (5.5,-4.2) {$\vdots$};
    \node[input text] (S-0) at (5.5,-5) {$\ldots$};

	\path node[unit] (S-1) at (4.5,-0) {$s_{11}$};
    \path node[unit] (S-2) at (4.5,-1.5) {$s_{n1}$};
    \path node[unit] (S-3) at (4.5,-3.5) {$s'_{11}$};
    \path node[unit] (S-4) at (4.5,-5) {$s'_{n1}$};

    \path node[unit] (S-11) at (6.5,-0) {$s_{1m}$};
    \path node[unit] (S-22) at (6.5,-1.5) {$s_{nm}$};
    \path node[unit] (S-33) at (6.5,-3.5) {$s'_{1m}$};
    \path node[unit] (S-44) at (6.5,-5) {$s'_{nm}$};

	\path node[unit, pin={[pin edge={->}]right:$z_{11} = \min\{\bar{s}_{11},\bar{s}'_{11}\}$}] (S-111) at (8.5,-0) {};
    \path node[unit, pin={[pin edge={->}]right:$z_{n1} = \min\{\bar{s}_{n1},\bar{s}'_{n1}\}$}] (S-222) at (8.5,-1.5) {};
    \path node[unit, pin={[pin edge={->}]right:$z_{1m} = \min\{\bar{s}_{1m},\bar{s}'_{1m}\}$}] (S-333) at (8.5,-3.5) {};
    \path node[unit, pin={[pin edge={->}]right:$z_{nm} = \min\{\bar{s}_{nm},\bar{s}'_{nm}\}$}] (S-444) at (8.5,-5) {};

    \foreach \dest in {1,...,3}
    	\foreach \src in {1,...,4}
			\path (I-\src) edge (H-\dest); 			
			
	\foreach \src in {1,...,3}
		\foreach \dest in {1,...,4}
			\path (J-\src) edge (S-\dest); 
	
	\node at  ($(H-2.south east) + (0.45,0)$) {\ldots};
	
	\draw[thick,dashed] ($(S-1.north west)+(-0.2,0.3)$)  rectangle ($(S-2.south east)+(0.2,-0.3)$) ;
	\node at  ($(S-1.north west) + (0.3,0.5)$) {$softmax$};
	
	\draw[thick,dashed] ($(S-11.north west)+(-0.2,0.3)$)  rectangle ($(S-22.south east)+(0.2,-0.3)$) ;
	\node at  ($(S-11.north west) + (0.3,0.5)$) {$softmax$};
	
	\draw[thick,dotted] ($(S-1.north west)+(-0.3,0.4)$)  rectangle ($(S-22.south east)+(0.3,-0.4)$) ;

	\draw[thick,dashed] ($(S-3.north west)+(-0.2,0.3)$)  rectangle ($(S-33.south east)+(0.2,-0.3)$) ;
	\node at  ($(S-33.south west) + (0.3,-0.5)$) {$softmax$};
	
	\draw[thick,dashed] ($(S-4.north west)+(-0.2,0.3)$)  rectangle ($(S-44.south east)+(0.2,-0.3)$) ;
	\node at  ($(S-44.south west) + (0.3,-0.5)$) {$softmax$};
	
	\draw[thick,dotted] ($(S-3.north west)+(-0.3,0.4)$)  rectangle ($(S-44.south east)+(0.3,-0.4)$) ;
	\draw[thick,dashed] ($(S-111.north west)+(-0.2,0.3)$)  rectangle ($(S-444.south east)+(0.2,-0.3)$) ;	
	
    \node at (8.5,-0.65) {$\vdots$};
    \node at (8.5,-2.45) {$\vdots$};
    \node at (8.5,-4.2) {$\vdots$};

    \node[input text] (S-0) at (7.5,-1.1) {$\bar{s}$};
    \node[input text] (S-0) at (7.5,-3.95) {$\bar{s}'$};

    \path[draw=black,solid,line width=1mm,fill=black] (7,-1.25) -- (8,-2);
    \path[draw=black,solid,line width=1mm,fill=black] (7,-3.75) -- (8,-3);

\end{tikzpicture}\\
		\centering
	\end{subfigure}
	\caption{The allocation component of the RegretNet neural network for  settings with $n$ {\em unit-demand} bidders and $m$ items.
	}
	\label{fig:regretnet-ud}
\end{figure*}

In particular,the allocation component computes two sets of scores $s_{ij}$'s and $s'_{ij}$'s. Let $s$, $s' \in \R^{nm}$ denote the corresponding matrices. The first set of scores are normalized along the rows and the second set of scores normalized along the columns. Both normalizations can be performed by passing these scores through softmax functions.
The allocation for bidder $i$ and item $j$  is then computed as the minimum of the corresponding normalized scores:
\begin{eqnarray*}
	z_{ij} \,=\,\varphi^{DS}_{ij}(s, s') \,=\,  \min\bigg\{
	\frac{e^{s_{ij}}}{\sum_{k=1}^{n+1} e^{s_{kj}}},\,\frac{e^{s'_{ij}}}{\sum_{k=1}^{m+1} \new{e^{s'_{ik}}}}\bigg\},
\end{eqnarray*}
where indices $n+1$ and $m+1$ denote dummy inputs that correspond to an item not being allocated to any bidder and a bidder not being allocated any item, respectively. %

\new{We first show that $\varphi^{DS}(s,s')$ as constructed is doubly stochastic, and that we do not lose in generality by the constructive approach that we take. See Appendix~\ref{APP:UNIT_DEMAND_DS} for a proof.}

\begin{lem}\label{LEM:UNIT_DEMAND_DS}
	\new{The matrix} $\varphi^{DS}(s, s')$ is doubly stochastic $\forall\, s, s' \in \R^{nm}$. \new{For any doubly stochastic matrix} $z \in [0,1]^{nm}$, $\exists\, s, s' \in  \R^{nm}$, for which $z = \varphi^{DS}(s, s')$.
\end{lem}

It remains to show that doubly-stochastic matrices correspond to
lotteries over one-to-one assignments. %
\pdadd{This is an easy corollary of~\cite{Birkhoff46}} and also a special case of the \emph{bihierarchy} structure proposed
in~\cite{BudishAER13} (Theorem 1), which we state in the following
lemma for completeness.
\begin{lem}[\cite{Birkhoff46}]\label{LEM:BVN-DECOMPOSITION}
  Any doubly stochastic matrix $A \in \mathbb{R}^{n \times m}$ can be
  represented as a convex combination of matrices $B^1, \dots, B^k$
  where each $B^\ell \in \{0,1\}^{n \times m}$ and
  $\sum_{j \in [m]} B_{ij} \leq 1$, $\forall i \in [n]$ and
  $\sum_{i \in [n]} B_{ij} \leq 1$, $\forall j \in [m]$.
      \end{lem}

      \citet{BudishAER13} also propose a polynomial algorithm to
      decompose the doubly stochastic matrix.
      The payment component for unit-demand valuations is the same as
      for the case of additive valuations (see Figure
      \ref{fig:gen-net}).

\subsection{Covering Number Bounds} 

\new{We conclude this section by instantiating our generalization
  bound from Section~\ref{sec:generalization-bound} to RegretNet,
  where we have both a regret and revenue term. Analogous results can
  also   be stated for RochetNet, where we only have a revenue term.}
Here,  $\Vert \cdot \Vert_1$ is the induced matrix norm, i.e. $\Vert
w\Vert_1 = \max_{j}\sum_{i} |w_{ij}|$.
\begin{thm}\label{thm:cover_regretnet}
	For RegretNet with $R$ hidden layers, $K$ nodes per hidden layer, $d_g$ parameters in the allocation component, $d_p$ parameters in the payment component, $m$ items, $n$ bidders, a sample size of $L$, and  the vector of all model parameters w \new{satisfying} %
	$\|w\|_{1} \leq W$,
	the following are valid bounds
	for \new{the $\Delta_{L}$ term defined in Theorem~\ref{THM:GBOUND},}
	for different bidder valuation types:

	(a) additive valuations:
	
	$%
	\Delta_{L} \leq O\big(\sqrt{R(d_g+d_p)
		\log(LW\max\{K, mn\})
		/
		{L}
	}
	\big)$,
	
	(b) unit-demand valuations:
	
	$\displaystyle
	\Delta_{L} \leq O\big(
	\sqrt{R(d_g+d_p)
		\log(LW\max\{K, mn\})
		/
		{L}
	}\big)$,
\end{thm}
The proof is given in Appendix~\ref{APP:THM_COVER_REGRETNET}. As the sample size $L \rightarrow \infty$, the term $\Delta_{L} \rightarrow 0$. 
The dependence of the above result on the number of layers, nodes, and parameters in the network is similar to standard covering number bounds for neural networks \citep{AnthonyP09}. 

\section{Training the Networks} \label{sec:training}

We next describe how we train the neural network architectures
presented in the previous sections.

The approach that we take for RochetNet is the standard
  (projected) stochastic gradient descent
  (SGD) for loss function
  $\mathcal{L}(\alpha, \beta)$ in Equation~\ref{eq:loss-rochetnet}.
  For additive valuations, we 
  project each weight $\alpha_{jk}$ during training
  into $[0,1]$ to guarantee
  feasibility.

In the case of RegretNet, we need to take care
  of the need for incentive alignment  directly. We use the augmented Lagrangian method to solve the constrained
  training problem in~\eqref{eq:ml-detailed2} over the space of
  \new{neural network} parameters $w$.
  The Lagrangian function for the optimization problem,
augmented with a quadratic penalty term for violating the constraints, is
\begin{eqnarray}
	\C_\rho(w; \lambda) ~=~
	-\frac{1}{L}\sum_{\ell=1}^L \sum_{i \in N} p^w_i(v^{(\ell)})
	\,+\, \sum_{i\in N}\lambda_{i}\,\widehat{\mathit{rgt}}_i(w)
	\,+\, \frac{\rho}{2} \sum_{i\in N}\Big(\widehat{\mathit{rgt}}_i(w)\Big)^2,
\end{eqnarray}
where $\lambda\in\mathbb{R}^n$ is a vector of Lagrange multipliers, 
and
$\rho > 0$ is a fixed
parameter that controls the weight on the quadratic penalty. 

\begin{algorithm}[tb]
	\caption{RegretNet Training}
	\label{alg:training}
	\begin{algorithmic}[1]
		\STATE {\bfseries Input:} Minibatches $\mathcal{S}_1, \ldots, \mathcal{S}_T$ of size $B$%
		\STATE {\bfseries Parameters:} 
		$\forall t, \rho_t>0$, $\gamma>0$, $\eta>0$, $\new{\Gamma} \in \N$, $K \in \N$
		\STATE {\bfseries Initialize:} $w^0 \in \R^d$, $\lambda^0 \in \R^n$ %
		\FOR{$t~=~0$ \textbf{to} $T$}
		\STATE Receive minibatch \new{$\mathcal{S}_t \,=\, \{v^{(1)}, \ldots, v^{(B)}\}$}
		\STATE Initialize misreport ${v'}_{i}^{(\ell)}\in V_i, \forall \ell \in [B],~ i \in N$
		\FOR{$r~=~0$ \textbf{to} \new{$\Gamma$}}
		\STATE ~~$\forall \ell \in [B],~ i \in N:$
		\STATE ~~~~~~~~${v'}_{i}^{(\ell)} \leftarrow {v'}_{i}^{(\ell)} +\gamma\nabla_{v'_i}\!\big[u^w_i\big(v^{(\ell)}_i; \big(v'_i, v^{(\ell)}_{-i}\big)\big)\big]\,\Big\vert_{v'_i={v'}_{i}^{(\ell)}}$
		\ENDFOR
		\STATE Compute regret gradient: ~$\forall \ell \in [B], i \in N$:
		\STATE ~~~~$g^t_{\ell, i} ~=~$
		\STATE ~~~~~~$\nabla_w\big[u^{w}_i\big(v_i^{(\ell)}; \big({v'}_{i}^{(\ell)}, v^{(\ell)}_{-i}\big)\big) - u^{w}_i(v_i^{(\ell)}; v^{(\ell)}) \big]\,\Big\vert_{w=w^t}$
		\STATE Compute Lagrangian gradient using \eqref{eq:C-grad} and update $w^t$:
		\STATE ~~~~$w^{t+1} \leftarrow w^t \,-\, \eta\nabla_w\, \C_{\rho_t}	(w, \lambda^{t})\,\Big\vert_{w=w^t}$
		\STATE Update Lagrange multipliers once in $Q$ iterations:
		\STATE ~~~~\textbf{if} {$t$ is a multiple of $Q$}
		\STATE ~~~~~~~~$\lambda^{t+1}_i \leftarrow \lambda_i^{t} + \rho_t\,\widetilde{\mathit{rgt}}_i(w^{t+1}), ~~\forall i \in N$
		\STATE ~~~~\textbf{else}  
		\STATE ~~~~~~~~$\lambda^{t+1} \leftarrow \lambda^t$
		\ENDFOR
	\end{algorithmic}
\end{algorithm}

The solver is described in Algorithm \ref{alg:training} and
alternates  between  the
following updates on the model parameters and the Lagrange multipliers: (a) $w^{new}\,\in\,  \text{argmin}_{w}\,\, \C_\rho(w^{old};\, \lambda^{old})$, and (b) $\lambda^{new}_{i} \,=\, \lambda_i^{old}\,+\, \rho\,\widehat{\mathit{rgt}}_i(w^{new}),~
\forall i\in N.$

We divide the training sample $\mathcal{S}$ into minibatches of size
$B$, and perform several passes over the training samples (with random
shuffling of the data after each pass).  We denote the minibatch
received at iteration $t$ by
\new{$\mathcal{S}_t \,=\, \{v^{(1)}, \ldots, v^{(B)}\}$.}  The update
(a) on model parameters involves an unconstrained optimization of
$\C_\rho$ over $w$ and is performed using a gradient-based
optimizer.

Let $\widetilde{\mathit{rgt}}_i(w)$ denote the empirical
regret in \eqref{eq:emp-rgt} computed on minibatch $\mathcal{S}_t$.
The gradient of $\C_\rho$ w.r.t.\ $w$ for fixed $\lambda^t$ is given
by:
\begin{align}
\nabla_w \, \C_\rho(w;\, \lambda^{t}) ~=~ &-\frac{1}{B}\sum_{\ell=1}^B \sum_{i\in N} \nabla_w\, p^w_i(v^{(\ell)})
+\, \sum_{i\in N}\, \sum_{\ell = 1}^B \lambda^t_{i}\, g_{\ell, i}
\,+\,\rho \sum_{i\in N} \, \sum_{\ell = 1}^B\, \widetilde{\mathit{rgt}}_i(w)\, g_{\ell, i},
\label{eq:C-grad}
\end{align}
where
\begin{align*}
g_{\ell, i} ~=~ \nabla_w\Big[ \max_{v'_i \in V_i}\,u^w_i\big(v_i^{(\ell)}; \big(v'_i, v^{(\ell)}_{-i}\big)\big) - u^w_i(v_i^{(\ell)}; v^{(\ell)})\Big].
\end{align*}

The terms $\widetilde{rgt}_i$ and $g_{\ell, i} $ in turn involve a ``max'' over misreports for each bidder $i$ and valuation profile $\ell$. We solve \new{this} inner maximization over misreports using another gradient based optimizer. 
In particular, we maintain a misreport ${v'}_{i}^{(\ell)}$ for each $i$ and valuation profile $\ell$. For each minibatch, we compute the optimal misreport, for each agent $i$ and each valuation profile $\ell$, by taking \new{$\Gamma$} gradient updates \dcpadd{from a randomly initialized valuation}, each update of the form
\begin{align}
    {v'}_{i}^{(\ell)} &= {v'}_{i}^{(\ell)} +\gamma\nabla_{v'_i}\!\big[u^w_i\big(v^{(\ell)}_i; \big(v'_i, v^{(\ell)}_{-i}\big)\big)\big]\,\Big\vert_{v'_i={v'}_{i}^{(\ell)}},\label{eq:gradasc}
%
\end{align}
for some $\gamma > 0$. 
\dcpadd{This is in the spirit of adversarial  machine learning, where these gradient steps on the input are taken to try to find a misreport for the agent that ``defeats" the incentive alignment of the mechanism.}

%
%
%
%

%
%

%
%
%
%
%
%
%
%
%
%
%
%
%
%
%
%
%
%
%
%
%
%
%
%
%
%
%

%
%
%
%
%
%
%
%

\begin{figure}[t]
\centering
\includegraphics[page=1,scale=0.3,trim={2.5cm 0 2.5cm 0},clip]{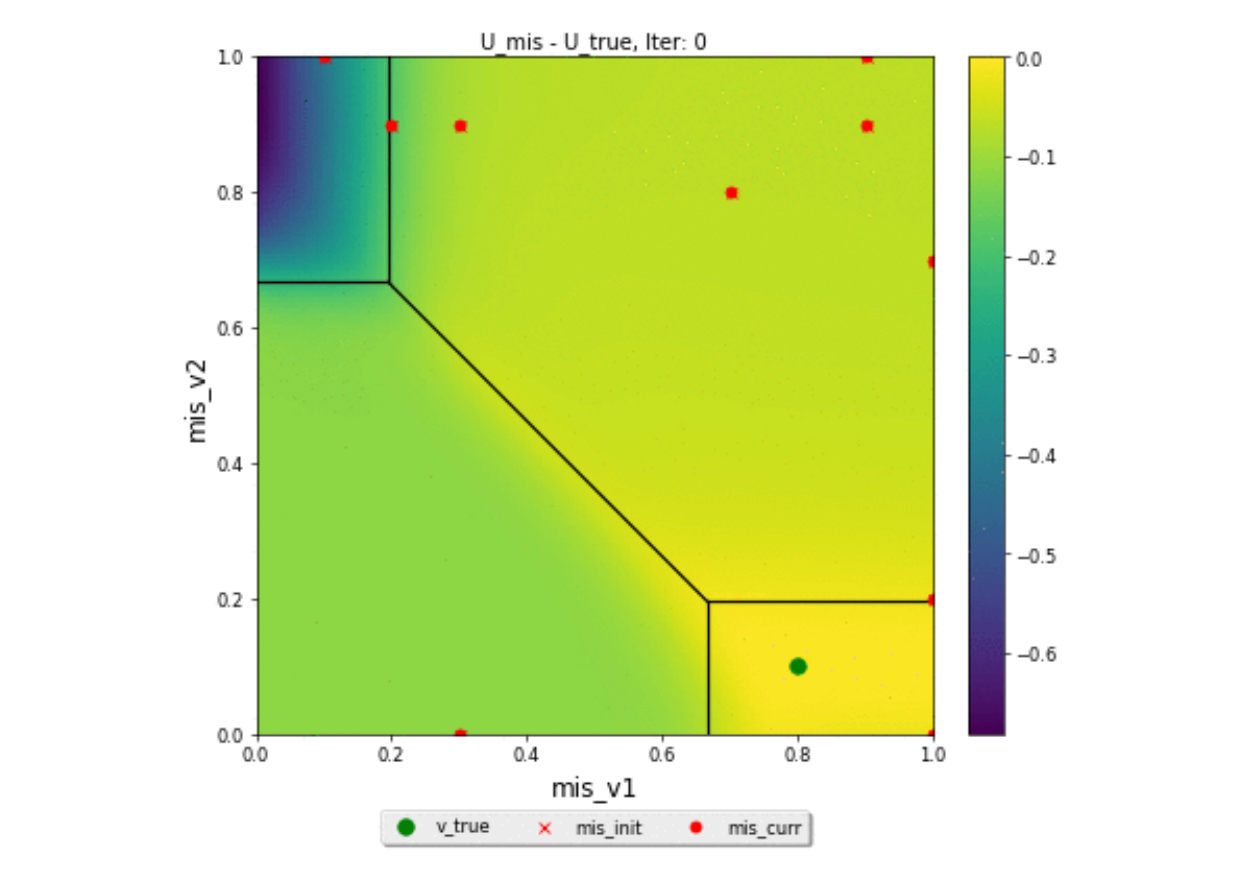}
\includegraphics[page=3,scale=0.3,trim={2.5cm 0 2.5cm 0},clip]{./plots_paul/vis0801}%
\includegraphics[page=5,scale=0.3,trim={2.5cm 0 2.5cm 0},clip]{./plots_paul/vis0801}\\
\includegraphics[page=7,scale=0.3,trim={2.5cm 0 2.5cm 0},clip]{./plots_paul/vis0801}
\includegraphics[page=9,scale=0.3,trim={2.5cm 0 2.5cm 0},clip]{./plots_paul/vis0801}%
\includegraphics[page=11,scale=0.3,trim={2.5cm 0 2.5cm 0},clip]{./plots_paul/vis0801}\\
\caption{The gradient-based approach to regret approximation, shown for a well-trained auction for Setting~\ref{SI}. The top left plot shows the true valuation (green dot) and ten random initial misreports (red dots). %
	The remaining plots give snapshots of the progress of gradient ascent on the input, showing this  every four steps.
\label{fig:gradient}}
\vspace{-10pt}
\end{figure}


Figure~\ref{fig:gradient} gives a  
visualization of this search for defeating misreports
when learning an optimal auction for a problem with 
a single bidder  with an  
additive valuation over two items, where the 
bidder's value for each item is an  independent draw
from $U[0,1]$ (see Section~\ref{sec:manelliandvincent}, Setting~\ref{SI}).
In the visualization, the bidder has true valuation $(v_1, v_2) = (0.1, 0.8)$, with this input represented as a green dot. The red crosses represent possible misreports.
The heat map shows the {\em utility gain}, $u_1((v_1, v_2); (b_1, b_2))  - u_1((v_1, v_2); (v_1, v_2))$, for this bidder when bidding 
some  amount $(b_1,b_2) \in [0,1]^2$ rather than truthfully.
This mechanism is already approximately DSIC and
the utility gain 
is  negative everywhere (and truthful bidding has zero regret), with shades of yellow corresponding to a misreport that is almost as good as a true report and shades of green towards blue corresponding to a harmful misreport. 
We illustrate the use of input gradients
by initializing each of $10$ possible misreports \pdadd{(we are using 10 misreports for illustration, in our experiments we initialize only a single misreport)},
and  performing $\Gamma=20$ \dcpadd{gradient-ascent} steps~\eqref{eq:gradasc} for each misreport.  
Figure~\ref{fig:gradient} shows the  initial misreports along with a new snapshot of the location of each misreport 
every four \dcpadd{gradient-ascent} steps. 

We use the Adam optimizer \citep{KingmaB15} for
updates on \new{model parameters} $w$ and \new{misreports}
${v'}_{i}^{(\ell)}$.\footnote{\new{Adam is a variant of SGD that makes
    use of a momentum term to update weights. Lines 9 and 15 in the
    pseudo-code of Algorithm~\ref{alg:training} are stated for a
    standard SGD algorithm.}}
Since the optimization problem is non-convex, 
{the solver is not guaranteed to reach a globally optimal solution.} 
However, this training algorithm proves very effective in our experiments. 
The learned auctions incur very low regret and closely match the structure of optimal auctions in settings where this structure is known from existing theory.

\section{Experimental Results}
\label{sec:experiments}

In this section, we demonstrate that our approach can recover
near-optimal auctions for essentially all settings for which an
analytical solution is known, \new{that it is an effective tool for
  confirming or refuting hypotheses about optimal designs,} and that
it %
can find new auctions for settings where there is no known analytical
solution.  \new{We present %
  a representative subset of the results here, and provide additional
  experimental results in Appendix~\ref{APP:ADDITIONAL-EXPERIMENTS}.}

\subsection{Setup}
We implement our framework using the TensorFlow deep learning
library. 
%
\new{For RochetNet we initialized parameters $\alpha$ and $\beta$
  in~\eqref{eq:loss-rochetnet}
  using a random uniform
  initializer over the interval [0,1] and a zero initializer,
  respectively.} \new{For RegretNet} we used the tanh activation
function at the hidden nodes, and Glorot uniform initialization
\new{\citep{GlorotY10}}.  \new{We perform
  cross validation to decide
  on the number of hidden layers and the number of nodes in each
  hidden layer. We include exemplary numbers that illustrate the
  tradeoffs in Section~\ref{sec:scaling-up}.} %
\new{We trained RochetNet on $2^{15}$ valuation profiles sampled every iteration in an online manner. We used the Adam optimizer with a learning rate of 0.1 for 20,000 iterations for making the updates. The parameter $\kappa$ in Equation~(\ref{eq:rochetnet-gradient}) was set to 1,000. Unless specified otherwise we used a max network over 1,000 linear functions to model the induced utility functions, and report our results on a sample of 10,000 profiles.} 

\new{For RegretNet we used a sample of 640,000 valuation profiles for training and a sample of 10,000 profiles for testing.}
The augmented Lagrangian solver was run 
for a maximum of 80 epochs \new{(full passes over the training set)} %
with a minibatch size of 128. 
The value of $\rho$ in \new{the} augmented Lagrangian was set to {\color{black}1.0} and 
incremented every \new{two} epochs.

	An  update on $w^t$ was performed for every minibatch   %
	using the Adam optimizer 
	with learning rate
	0.001. For each update on $w^t$, we ran \new{$\Gamma=25$}
        misreport updates steps with learning rate 0.1.
At the end of 25 updates, the optimized misreports for the current minibatch were cached and used to initialize the misreports for the same minibatch in the next epoch. 
An update on $\lambda^t$ was performed \new{once every} 100 minibatches (i.e., \new{$Q=100$}).

We ran all experiments on a compute cluster with NVDIA \new{Graphics Processing Unit (GPU)} cores.%

\subsection{Evaluation}
In addition to the revenue of the learned auction on a test set,  we also
evaluate the regret \new{achieved by RegretNet}, averaged across all bidders and test valuation
profiles\new{, i.e., $rgt \,=\, \frac{1}{n}\sum_{i=1}^n
	\widehat{\mathit{rgt}}_i(g^{w},p^{w})$.}
Each  $\widehat{\mathit{rgt}_i}$ has \new{an inner ``max''} of the utility function over bidder valuations $v_i' \in V_i$ (see \eqref{eq:emp-rgt}). We evaluate these terms by running gradient ascent on $v_i'$ with a step-size of 0.1 for 2,000 iterations (we test 1,000 different random  initial $v'_i$ and  report the one \new{that} achieves the largest regret). %
%
%
%
%
%
%
%

\new{For some of the experiments we also report the total time
  required to train the network. This time is incurred during offline
  training, while the allocation and payments can be computed in a few
  milliseconds once the network is trained.}

\subsection{The Manelli-Vincent and Pavlov Auctions}\label{sec:manelliandvincent}

\new{As a representative example of the exhaustive set of analytical results that we can recover with our approach we discuss the Manelli-Vincent and Pavlov auctions \citep{ManelliVincent06,Pavlov11}. We specifically consider the following single-bidder, two-item settings:}

\begin{enumerate}[label=\Alph*.,ref=\Alph*]
	\item Single bidder with additive valuations over \new{two} items, where the item values \new{are independent draws} from $U[0,1]$. 
	\label{SI}
	\item Single bidder with unit-demand valuations over \new{two} items, where the item values \new{are independent draws} from $U[2,3]$. 
	\label{SV}
\end{enumerate}

\new{The optimal design for the first setting is given by \citet{ManelliVincent06}, who show that the optimal mechanism is deterministic and offers the bidder three options:  receive both items and pay $(4-\sqrt{2})/3$, receive item $1$ and pay $2/3$, or receive item 2 and pay $2/3$.}
\new{For the second setting \citet{Pavlov11} shows that it is optimal to offer a fair lottery $(\frac{1}{2},\frac{1}{2})$ over the items (at a discount), or to purchase any item at a fixed price. For the parameters here the price for the lottery is $\frac{1}{6} (8 + \sqrt{22}) \approx 2.115$ and the price for an individual item is $\frac{1}{6}+\frac{1}{6} (8 + \sqrt{22}) \approx 2.282$.}

\begin{figure}[t]
	\centering
	\centering
	\begin{subfigure}{0.49\textwidth}
		\centering
		\hspace*{-10pt}
		\includegraphics[scale=0.37]{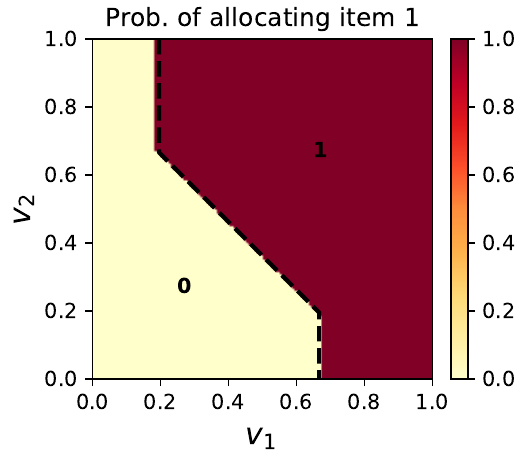}
		\hspace*{-10pt}{\scriptsize (a)}\hspace*{-2pt}
		\includegraphics[scale=0.37]{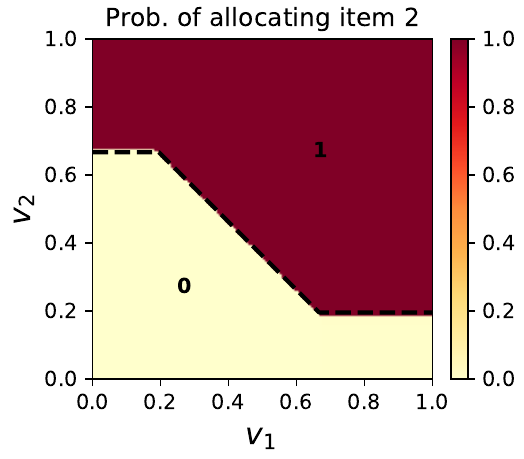}
	\end{subfigure}
	\begin{subfigure}{0.49\textwidth}
		\centering
		\hspace*{-10pt}
		\includegraphics[scale=0.37]{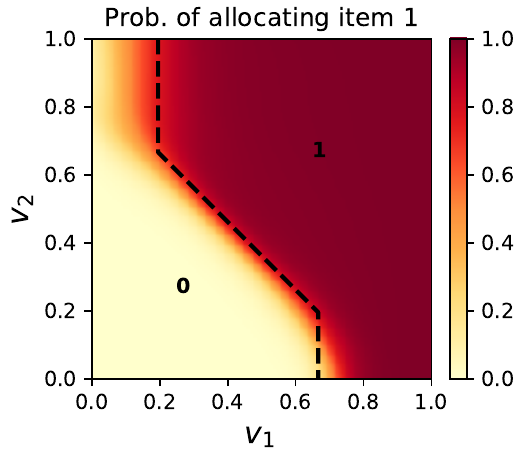}
		\hspace*{-10pt}{\footnotesize (b)}\hspace*{-4pt}
		\includegraphics[scale=0.37]{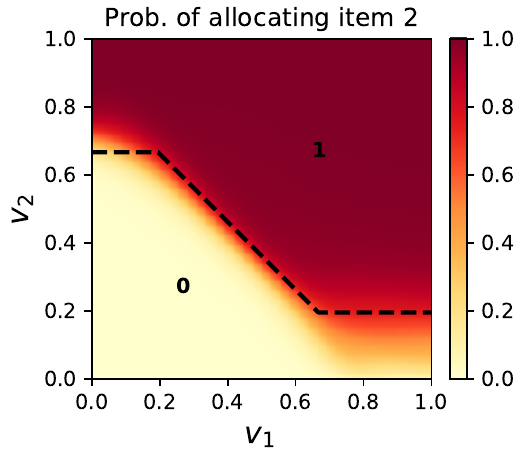}
	\end{subfigure}
	
	\begin{subfigure}{0.49\textwidth}
		\centering
		\hspace*{-16pt}
		\includegraphics[scale=0.145]{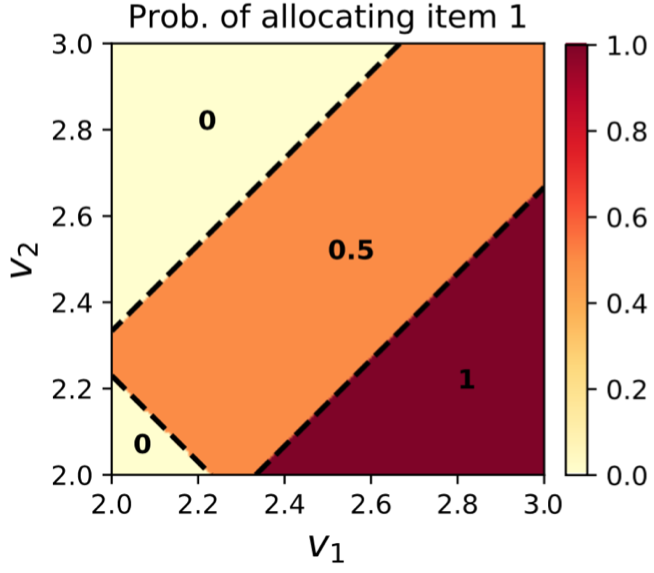}
		\hspace*{-10pt}{\scriptsize (c)}\hspace*{-2pt}
		\includegraphics[scale=0.37]{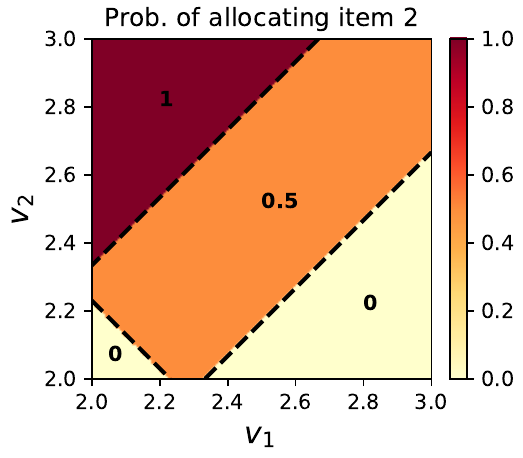}
	\end{subfigure}
	\begin{subfigure}{0.49\textwidth}
		\centering
		\hspace*{-10pt}
		\includegraphics[scale=0.37]{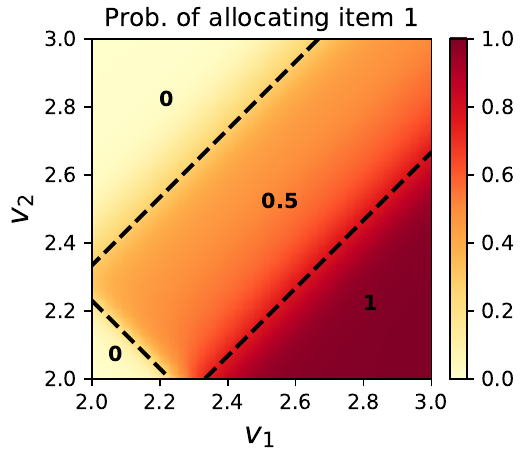}
		\hspace*{-10pt}{\footnotesize (d)}\hspace*{2pt}
		\includegraphics[scale=0.37]{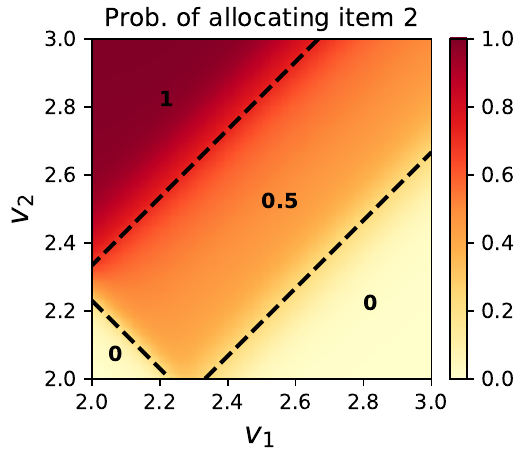}
	\end{subfigure}
	\caption{%
		Side-by-side comparison of allocation rules learned by RochetNet and RegretNet for single bidder, two items settings. Panels (a) and (b) are for Setting~\ref{SI} and Panels (c) and (d) are for Setting~\ref{SV}. 
		The panels describe the probability that the bidder is allocated item 1 (left) and item 2 (right) for different valuation inputs.  
		The optimal auctions are described by the regions separated by the dashed black lines, with the numbers in black the optimal probability of allocation in the region.
		\label{fig:alloc-basic-settings}}
\end{figure}

\new{We used two hidden layers with 100 hidden nodes in RegretNet for these settings.}
\new{A visualization of the optimal allocation rule and those learned by RochetNet and RegretNet is given in Figure~\ref{fig:alloc-basic-settings}. Figure~\ref{fig:small-settings}(a) gives the optimal revenue, the revenue and regret obtained by RegretNet, and the revenue obtained by RochetNet. Figure~\ref{fig:small-settings}(b) shows how these terms evolve over time during training in RegretNet.}

\new{We find that both approaches essentially recover the optimal design, not only in terms of revenue, but also in terms of the allocation rule and transfers. The auctions learned by RochetNet are exactly DSIC and match the optimal revenue precisely, with sharp decision boundaries in the allocation and payment rule. The decision boundaries for RegretNet are smoother, but still remarkably accurate. The revenue achieved by RegretNet matches the optimal revenue up to a $< 1\%$ error term and the regret it incurs is $< 0.001$.}
\new{The plots of the test revenue and regret show that the augmented Lagrangian method is effective in driving the test revenue and the test regret towards optimal levels.}

\new{The additional domain knowledge incorporated into the RochetNet architecture leads to exactly DSIC mechanisms that match the optimal design more accurately, and speeds up computation (the training took about 10 minutes compared to 11 hours). On the other hand, we find it surprising how well RegretNet performs given that it starts with no domain knowledge at all.}

\new{We present and discuss a host of additional experiments with single-bidder, two-item settings in Appendix~\ref{APP:ADDITIONAL-EXPERIMENTS}.}

\begin{figure}[t]
	\begin{subfigure}[b]{0.99\textwidth}
		\begin{center}
			\begin{tabular}{|l|c|c|c|c|c|c|c|}
				\hline
				\multirow{2}{*}{Distribution} & Opt & \multicolumn{5}{|c|}{RegretNet} & RochetNet \\
				\cline{2-8} & $\mathit{rev}$ & $\mathit{rev}$ & $\mathit{rgt (mean)}$& $\mathit{rgt (90\%)}$ &$\mathit{rgt(95\%)}$& $\mathit{rgt(99\%)}$ & $\mathit{rev}$ \\
				\hline
				Setting \ref{SI} & 0.550 & 0.554 & $< 0.001$ & $< 0.001$ & $ 0.001$ & $ 0.002$ &0.550 \\
				Setting \ref{SV} & 2.137 & 2.137 & $< 0.001$ & $< 0.001$ & $< 0.001$ & $ 0.002$ &2.136 \\ 
				\hline
			\end{tabular}
			~\\[16pt]\footnotesize{(a)}
		\end{center}
	\end{subfigure}

	\begin{subfigure}[b]{0.99\textwidth}
		\vspace*{0.25cm}
		\begin{center}
			\includegraphics[scale=0.57]{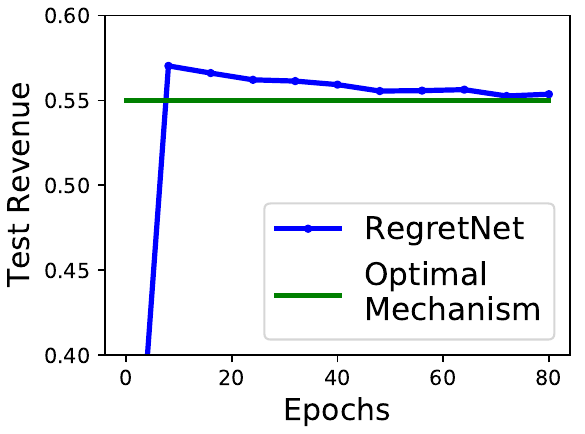}
			~~
			\includegraphics[scale=0.57]{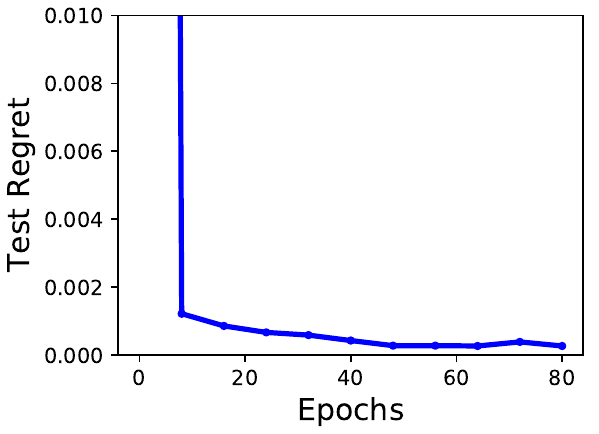}
			~\\[-3pt]\footnotesize{(b)}
		\end{center}
	\end{subfigure}
	\caption{ (a): Test revenue and test regret for RegretNet and test revenue for RochetNet  for Settings \ref{SI} and \ref{SV}. 
		(b): Plot of test revenue and test regret as a function of training epochs for Setting \ref{SI} with RegretNet.}
	\label{fig:small-settings}
	\vspace{-7pt}
\end{figure}

\subsection{The Straight-Jacket Auction}

\begin{figure}[t]
	\centering
	\begin{tabular}{|r|c|c|}
		\hline
		Items & SJA ($\mathit{rev}$) & RochetNet ($\mathit{rev}$)\\
		\hline
		2& 0.549 & 0.549\\
		3& 0.875 & 0.875\\
		4& 1.219 & 1.219\\
		5 & 1.576 & 1.576\\
		6 & 1.943 & 1.943\\
		7 & 2.318 & 2.318\\
		8 & 2.699 & 2.699\\
		9 & 3.086 & 3.086\\
		10 & 3.478 & 3.478\\
		\hline
	\end{tabular}
	\caption{Revenue of the Straight-Jacket Auction (SJA) computed via the recursive formula in~\citep{GK18}, and the test revenue of the auction learned by RochetNet, for various numbers of items $m$. The SJA is known to be optimal for up to six items and conjectured to be optimal for any number of items.\label{fig:sja}}
	\vspace{-10pt}
\end{figure}

Extending the analytical result of \cite{ManelliVincent06} to a single bidder, and an arbitrary number of items (even with additive preferences, all uniform on $[0,1]$) has proven elusive. It is not even clear whether the optimal mechanism is deterministic or requires randomization.

\new{A breakthrough came with~\cite{GK18}, who were able to find a pattern in the results for two items and three items. The proposed mechanism---the Straight-Jacket Auction (SJA)---offers bundles of items at fixed prices. The key to finding these prices is to view the best-response regions as a subdivision of the $m$-dimensional cube, and observe that there is an intrinsic relationship between the price of a bundle of items and the volume of the respective best-response region.}

\new{\citet{GK18} give a recursive algorithm for finding the
	subdivision and the prices, and used LP duality to prove that the
	SJA is optimal for $m \leq 6$ items.\footnote{
		\new{The duality argument developed by~\citeauthor{GK18}
			is similar but incomparable to the duality approach of~\cite{DaskalakisDT13}. We will return to the latter in Section~\ref{sec:dualityddt}.}} They also conjecture that the SJA remains optimal for general $m$, but were unable to prove it.}

\new{Figure~\ref{fig:sja} gives the revenue of the SJA, and that found by RochetNet for $m \leq 10$ items. We used a test sample of $2^{30}$ valuation profiles (instead of 10,000) to compute these numbers for higher precision. It shows that RochetNet finds the optimal revenue for $m \leq 6$ items, and that it finds DSIC auctions whose revenue matches that of the SJA for $m = 7, 8, 9,$ and $10$ items. Closer inspection reveals that the allocation and payment rules learned by RochetNet essentially match those predicted by \citeauthor{GK18} for all $m \leq 10$.}
\new{We take this as strong additional evidence that their conjecture
  is correct.}

For these experiments, we used a max network over 10,000 linear
functions (instead of 1,000) to increase the representation and
flexibility of the neural network. \ZFadd{This
  \emph{overparameterization trick} is commonly used in deep learning
  and has proven to be very effective in practice~\citep{KSH12, pmlr-v97-allen-zhu19a}. We illustrate this effect in Appendix~\ref{app:rochet-overparam}.} We followed up on the usual training
phase with an additional 20 iterations of training using Adam
optimizer with learning rate 0.001 and a minibatch size of
$2^{30}$.

We also found it useful to impose {\em item-symmetry}
on the learned auction, especially for $m = 9$ and $10$ items, as this helped with
accuracy and reduced training time. Imposing symmetry comes without
loss of generality for auctions with an item-symmetric
distribution~\citep{Daskalakis12}. To impose item symmetry, we first permute the inputs to be in ascending order, compute the allocation and payment on this permuted input, and then invert the permutation of allocation to compute the mechanism for the original inputs. With these modifications it took about 13 hours to train the networks. 

\subsection{Discovering New Analytical Results}
\label{sec:dualityddt}

We next demonstrate the potential of RochetNet to help to
  discover new analytical results for optimal auctions.
  For this, we consider a single bidder with additive but correlated
  valuations for two items as follows:
	\begin{enumerate}[label=\Alph*.,start=3,ref=\Alph*]
		\item One additive bidder and two items, where the bidder's valuation is drawn uniformly from the triangle $T=\{(v_1, v_2)|\frac{v_1}{c}+v_2 \leq 2, v_1\geq 0, v_2\geq 1\}$ where $c>0$ is a free parameter. \label{exp:triangle-1}
                \end{enumerate}
                
	There is no analytical result for the optimal auction design
        for this setting. We ran RochetNet for different values of $c$
        to discover the optimal auction. The mechanisms learned by
        RochetNet for $c=0.5, 1, 3,$ and $5$ are shown in
        Figure~\ref{fig:alloc-triangle-setting}.

        Based on this, we
        conjectured that the optimal mechanism contains two menu items
        for $c \leq 1$, namely $\{(0,0), 0\}$ and $\{(1,1),
        \frac{2+\sqrt{1+3c}}{3}\}$, and three menu items for $c > 1$,
        namely $\{(0,0), 0\}$, $\{(1/c, 1), 4/3\}$, and $\{(1,1),
        1+c/3\}$, giving the optimal allocation and payment in each
        region.
        In particular, as $c$ transitions from values less than or equal to 1 to values larger than 1, the optimal mechanism transitions from being deterministic to being randomized. Figure~\ref{tab:triangle-via-dual} gives the revenue achieved by RochetNet and the conjectured optimal format for a range of parameters $c$ computed on $2^{30}$ valuation profiles.

\new{We can  validate the optimality of this conjectured design
  through duality theory~\citep{DaskalakisDT13}.
  The proof is given in Appendix~\ref{app:duality-framework}.}
\begin{figure}[t]
	\centering
	\begin{tabular}{|r|c|c|}
		\hline
		\multicolumn{1}{|c|}{c}& \multicolumn{1}{|c|}{Opt ($\mathit{rev}$)} & \multicolumn{1}{|c|}{RochetNet ($\mathit{rev}$)}\\
		\hline
		0.125 & 1.029 & 1.029\\
		0.200   & 1.046 & 1.046\\
		0.250  & 1.056 & 1.056\\
		0.500   & 1.104 & 1.104\\
		1.000     & 1.185 & 1.185\\
		3.000     & 1.481 & 1.481\\
		5.000     & 1.778 & 1.778\\
		8.000     & 2.222 & 2.222\\
		10.000    & 2.518 & 2.518\\
		20.000    & 4.000 & 4.000\\
		\hline
	\end{tabular}
	\caption{Test revenue of the newly discovered optimal mechanism and that of RochetNet, for Setting~\ref{exp:triangle-1} with varying parameter $c$.}\label{tab:triangle-via-dual}
\end{figure}

\begin{figure}[t]
	\centering
	\centering
	\begin{minipage}{0.49\textwidth}
		\centering
		\hspace*{-10pt}
		\includegraphics[scale=0.38]{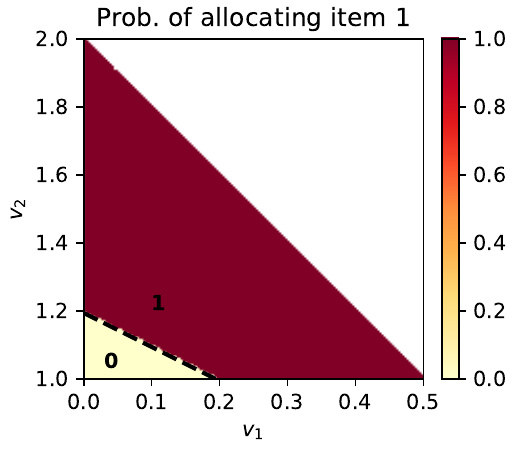}
		\hspace*{-10pt}{\scriptsize (a)}\hspace*{-2pt}
		\includegraphics[scale=0.38]{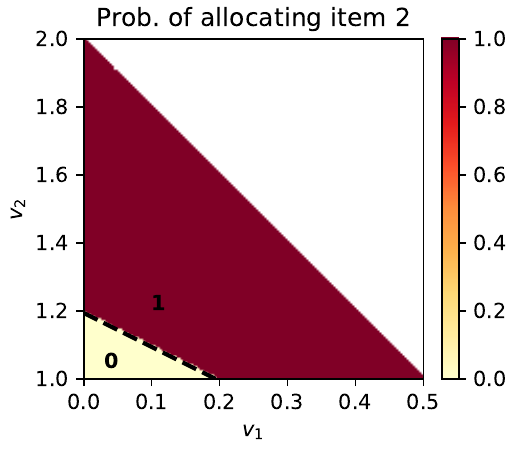}
	\end{minipage}
	\begin{subfigure}{0.49\textwidth}
		\centering
		\hspace*{-10pt}
		\includegraphics[scale=0.38]{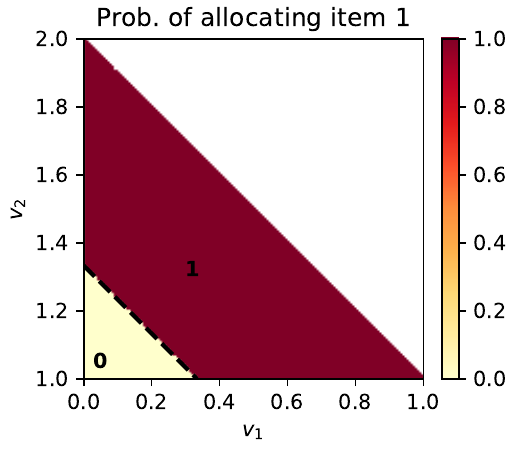}
		\hspace*{-10pt}{\footnotesize (b)}\hspace*{-4pt}
		\includegraphics[scale=0.38]{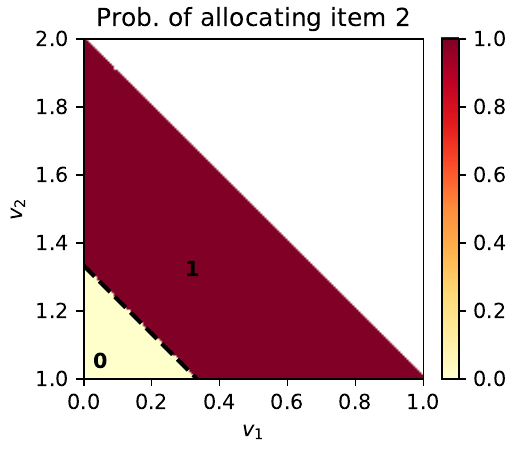}
	\end{subfigure}
	
	\begin{minipage}{0.49\textwidth}
		\centering
		\hspace*{-16pt}
		\includegraphics[scale=0.38]{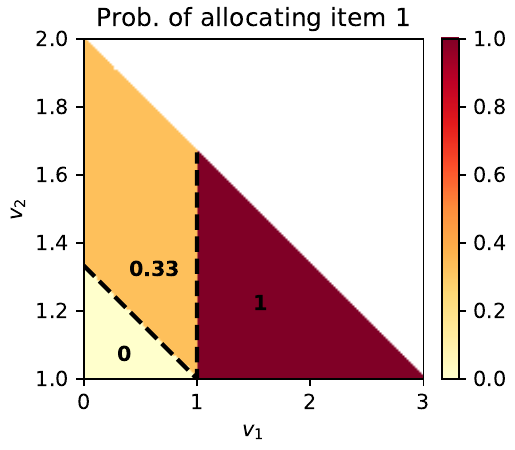}
		\hspace*{-10pt}{\scriptsize (c)}\hspace*{-2pt}
		\includegraphics[scale=0.38]{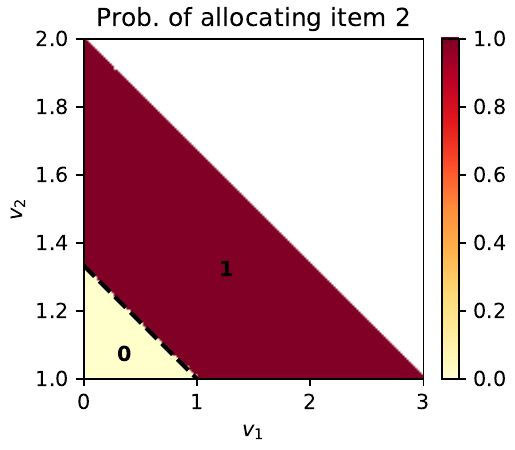}
	\end{minipage}
	\begin{subfigure}{0.49\textwidth}
		\centering
		\hspace*{-10pt}
		\includegraphics[scale=0.38]{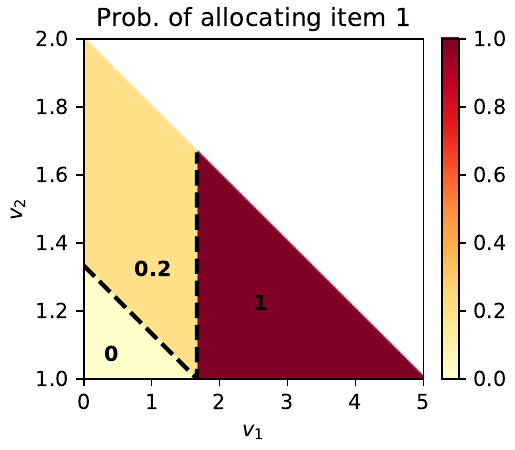}
		\hspace*{-10pt}{\footnotesize (d)}\hspace*{2pt}
		\includegraphics[scale=0.38]{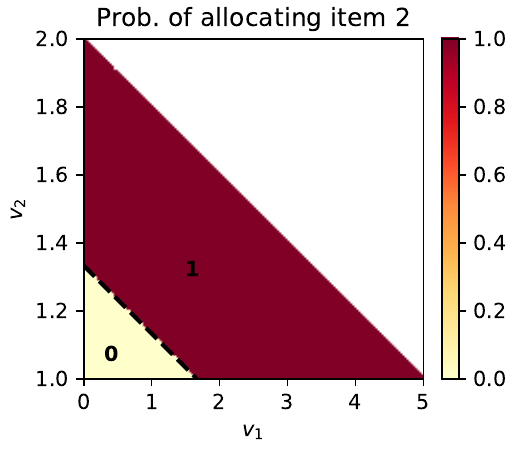}
	\end{subfigure}
	\caption{%
		Allocation rules learned by RochetNet for Setting~\ref{exp:triangle-1}. 
		The panels describe the probability that the bidder is allocated item 1 (left) and item 2 (right) for $c = 0.5, 1, 3,$ and $5$. The auctions proposed in Theorem~\ref{THM:NEW_UNIFORM_TRIANGLE} are described by the regions separated by the dashed black lines, with the numbers in black the optimal probability of allocation in the region.
		\label{fig:alloc-triangle-setting}}
	\vspace{-10pt}
\end{figure}

\begin{thm}\label{THM:NEW_UNIFORM_TRIANGLE}
	\new{For any $c>0$, suppose the bidder's valuation is uniformly distributed over set $T=\{(v_1, v_2)|\frac{v_1}{c}+v_2 \leq 2, v_1\geq 0, v_2\geq 1\}$. Then the optimal auction contains}
	\new{two menu items $\{(0,0), 0\}$ and $\{(1,1), \frac{2+\sqrt{1+3c}}{3}\}$ when $c \leq 1$, and}
	\new{three menu items $\{(0,0), 0\}$, $\{(1/c, 1), 4/3\}$, and $\{(1,1), 1+c/3\}$ otherwise.}
\end{thm}

  \ssredit{In  Appendix~\ref{app:additional-tri}, we also give the mechanisms learned by RochetNet for two additional settings. Taken together, these results demonstrate that RochetNet is a powerful tool to help in the discovery of new analytical results.}
\ZFadd{In follow-up work,~\citet{STZ18} also use a neural network framework,  closely related to RochetNet, to discover an
  optimal analytical result for a similar setting: a single additive
  bidder and two items, where the bidder's valuation is drawn
  uniformly from the triangle
  $\{(v_1, v_2)|\frac{v_1}{c}+v_2 \leq 1, v_1\geq 0, v_2\geq 0\}$.}

\subsection{Experiments with Optimal Mechanisms that Require an Infinitely-sized Menu}

\ssredit{We now demonstrate how RochetNet performs for settings where the optimal mechanism is known to require an  infinite number of menu choices. For this, we consider the following setting from~\cite{DaskalakisEtAl17}}:
\begin{enumerate}[label=\Alph*.,start=4,ref=\Alph*]
\item %
\new{One} additive bidders and \new{two} items, where bidders draw their value for each item \new{independently} from $\mathit{Beta}(\alpha=1,\beta=2)$.\footnote{A $Beta$ distribution with $\alpha=1, \beta=2$  has the density function $f(x) = 2(1 -x)$}\label{exp:beta}
\end{enumerate}

\ssredit{This setting and the corresponding optimal mechanism, with its infinite menu size, is described in detail in Example~3 of~\cite{DaskalakisEtAl17}. 
We seek to evaluate the performance of RochetNet for different-sized menus. In Figure~\ref{fig:beta}, we report the revenue, the number of menu choices represented in RochetNet, and the number of menu choices that are active for one or more samples in the test set.} 
\ssredit{As we increase the number of initialized menu choices, the number of active menu items increases as well.} 
\ssredit{Comparing the optimal infinite-sized menu with the menu learned by RochetNet, we find that the difference in revenue comes from a large number of menu items that each only contribute marginally to the net revenue ($< 10^{-5})$. RochetNet fails to learn some of these menus  due to the fixed size of mini-batches and the numerical tolerance error of the optimization routine. Regardless, the overall gap in revenue is negligible.} \pdadd{Already with two active menu items, RegretNet achieves a revenue of $\sim 0.3309$ ($99.93\%$ of optimal), while with three or more active menu items the revenue is at least $\sim0.3310$ ($99.96\%$ of optimal). } 



%
\begin{figure}[t]
	\centering
	\begin{tabular}{|c|c|c|c|}
		\hline
		\vtop{\hbox{\strut Number of initialized menu}\hbox{\strut choices in RochetNet}} & RochetNet ($\mathit{rev}$)& \vtop{\hbox{\strut Number of active menu}\hbox{\strut items in RochetNet}} \\
		\hline
		2 & 0.3309 & 2 \\
		5 & 0.3309 & 2 \\
		10 & 0.3310 & 3 \\
		20 & 0.3310 & 4 \\
		50 & 0.3310 & 6 \\
		100 & 0.3310 & 7 \\
		200 & 0.3310 & 11 \\
		500 & 0.3310 & 11 \\
		1,000 & 0.3310 & 12 \\
		2,000 & 0.3310 & 17 \\
		5,000 & 0.3310 & 34 \\
		10,000 & 0.3310 & 59 \\
		20,000 & 0.3310 & 89 \\
		\hline
	\end{tabular}
	\caption{Test revenue of the auction learned by RochetNet for different menu sizes in Setting~\ref{exp:beta}. The number of active menus increases as the number of initialized menu choices increases. The optimal mechanism requires an infinitely-sized menu and achieves a revenue of $0.3311$.\label{fig:beta}}
	\vspace{-10pt}
\end{figure}

\subsection{Scaling Up}\label{sec:scaling-up}

In this section, we consider settings with up to five bidders and up
to ten items. %
\new{This is} several orders of magnitude more complex than existing analytical or computational \new{results}. \new{It is also a natural playground for RegretNet, as no tractable characterizations of IC mechanisms are known for these settings.} 
\new{We specifically consider the following two settings, that
  generalize the basic setting considered in~\citet{ManelliVincent06} and~\citet{GK18} to more than one bidder:}
\begin{enumerate}[label=\Alph*.,start=5,ref=\Alph*]
	\item \new{Three} additive bidders and \new{ten} items, where  bidders draw their value for each item \new{independently} from $U[0,1]$. \label{exp:scale-up-1}
	\item \new{Five} additive bidders and \new{ten} items, where  bidders draw their value for each item \new{independently} from $U[0,1]$. \label{exp:scale-up-2}
\end{enumerate}

\new{An analytical description of the optimal auction for these settings is not known. However, running a separate Myerson auction for each item is optimal in the limit of the number of bidders \citep{Palfrey83}. For a regime with a small number of bidders, this provides a strong benchmark. We also compare to selling the grand bundle via a Myerson auction.} 

For Setting \ref{exp:scale-up-1}, we show in Figure \ref{fig:large-settings}(a) the revenue and regret of the learned auction on a validation sample of 10,000 profiles, obtained %
with different architectures. Here \new{$(R, K)$} denotes an architecture with $R$ hidden layers and \new{$K$} nodes per layer. %
The (5,\,100) architecture %
has the lowest regret among all the 100-node networks for \new{both Setting \ref{exp:scale-up-1}
	and Setting \ref{exp:scale-up-2}.} %
Figure \ref{fig:large-settings}(b) shows that the learned auctions yield higher revenue compared to the baselines, and do so with tiny regret.

\begin{figure}
	\begin{subfigure}{0.99\textwidth}
		\centering
		\includegraphics[scale=0.25]{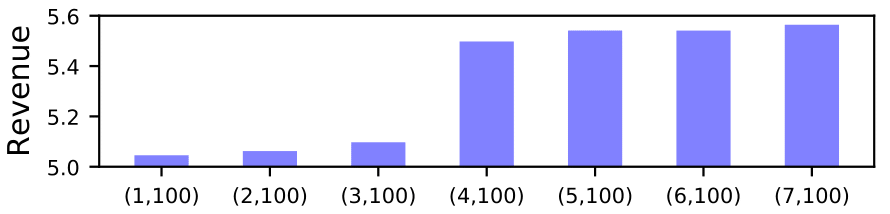}
		\includegraphics[scale=0.25]{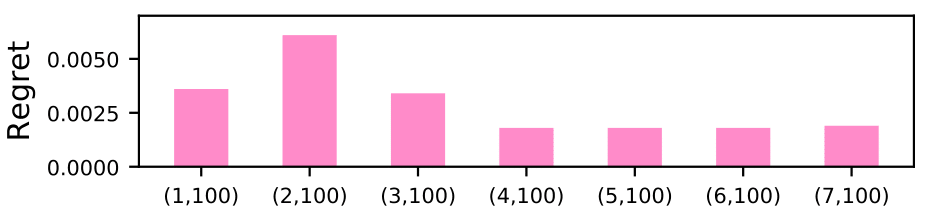}
		~\\[-5pt]
		{\scriptsize (a)}%
	\end{subfigure}
	~\\[5pt]

	\begin{subfigure}[b]{0.99\textwidth}
		\begin{center}
			\begin{tabular}{|l|c|c|c|c|c|c|c|}
				\hline
				\multirow{2}{*}{Setting} & \multicolumn{5}{|c|}{RegretNet} & Item-wise & Bundled \\
				\cline{2-6} & $\mathit{rev}$ & $\mathit{rgt (mean)}$& $\mathit{rgt (90\%)}$ &$\mathit{rgt(95\%)}$& $\mathit{rgt(99\%)}$ & Myerson & Myerson \\
				
				\hline
				\ref{exp:scale-up-1}: $3\times 10$
				& 5.541 &  0.002 & 0.003 & 0.004 & 0.006 & 5.310 & 5.009 
				\\
				\hline
				\ref{exp:scale-up-2}: $5\times 10$
				& 6.778 &  0.005 & 0.008 & 0.009 & 0.014  & 6.716 & 5.453 
				\\
				\hline
			\end{tabular}
			~\\[6pt]\scriptsize{(b)}\\[5pt]
		\end{center}
	\end{subfigure}
	~\\
	\caption{
		(a) Revenue and  regret of RegretNet on the validation set for auctions learned for Setting \ref{exp:scale-up-1} using different architectures, where $(R, K)$ denotes $R$ hidden layers and $K$ nodes per layer.
		(b) Test revenue and test regret for Settings \ref{exp:scale-up-1} and \ref{exp:scale-up-2} for the (5, 100) architecture.
		%
	}
	\label{fig:large-settings}
	\vspace{-14pt}
\end{figure}

\subsection{Comparison to Linear Programming}
\label{sec:lpdcp} 
In this section, we compare the train time  and solution quality of \new{RegretNet} with
the solve time and solution quality of the LP-based approach proposed \new{in}~\citet{ConitzerS02,ConitzerS04a}.
To be able to run the LP, we consider the small setting of two additive bidders and two items,
and with  bidders that draw their value for each item independently from $U[0,1]$. 

\new{For RegretNet, we used two hidden layers with 100 nodes per hidden layer.}  The LP \new{was} solved with the commercial solver Gurobi and on an Amazon AWS EC-2 instance with 48 cores and 96GB memory. For the LP-based approach, we \new{handle} continuous valuations by discretizing the values into \new{11} bins per item  (resulting in $\approx 9\times10^5$
decision variables and $\approx 3.6\times10^6$ constraints),
\dcpadd{and adopt two  different rounding strategies, one that 
rounds a continuous input valuation profile to the nearest
discrete profile for evaluation ({\tt -nearest}), and one that rounds  
a continuous input valuation profile
down to the nearest
discrete profile for evaluation ({\tt -down}).
Whereas the LP-based mechanism with {\tt -nearest} rounding fails IR, the use of {\tt -down} rounding
ensures the LP-based approach is IR.}

\begin{figure}[t]
\vspace*{10pt}
\centering
\begin{tabular}{|c|c|c|c|c|c|c|}
\hline

Setting & \multicolumn{2}{c|}{Method} & $\mathit{rev}$ & $\mathit{rgt (mean)}$ & \textit{IR viol.} & Run-time (in hours) \\
\hline
\multirow{4}{*}{$2 \times 2$} & \multicolumn{2}{c|}{RegretNet} & 0.878 & \textless{} 0.0005 & 0 & $\sim$ 3.7 \\
\cline{2-7}
 & \multirow{2}{*}{LP (11 bins/ value)} & {\tt -nearest} & 0.927 & 0.002 & 0.003 & \multirow{2}{*}{61.9} \\
 \cline{3-6}
 &  & {\tt -down} & 0.837 & 0.001 & 0 &  \\
 \cline{2-7}
 & \multicolumn{1}{c|}{ LP (12 bins/ value)} & - & - & - & - & \textgreater{}216\\
 \hline
\end{tabular}
	\caption{Test revenue, test regret, test IR violation, and training time or solve time for RegretNet and an LP-based approach, for a two bidder, two items setting with additive uniform valuations.\label{tab:regret_vs_lpa}}
\end{figure}

The results for this set-up are shown in Figure~\ref{tab:regret_vs_lpa}.  We also
report the violations in IR constraints incurred by the LP on the test
set; for $L$ valuation profiles, this is measured by
$\frac{1}{nL}\sum_{\ell = 1}^L\sum_{i \in N}\, \max\{$  $-u_i(v^{(\ell)}),
0\}$. \sloppy Due to the coarse discretization, the LP approach with nearest-point rounding suffers substantial
IR violations. \dcpadd{As a result of this, as well as its relatively high regret compared to RegretNet,  the relatively high revenue achieved by the LP together with nearest-point rounding, compared with RegretNet, is misleading.} \ssredit{For this reason, we also include the performance of the LP-based mechanism when the continuous input valuation profiles are rounded down to their respective discrete profiles. There we see zero IR violation but substantially lower revenue than RegretNet (and still with  higher regret)}
We \new{were} not able to run \new{an} LP for this setting for a finer discretization than 11 bins per item value
in more than nine days (216 hours) of compute time.\footnote{We used an AWS EC-2 instance with 48 cores and 96GB of memory} In contrast, \new{RegretNet}
yields very low regret along with zero IR violations (as the neural
\new{network} \new{satisfies} IR by design), and does so in 
around four hours. In fact, even for the larger Settings
\ref{exp:scale-up-1}--\ref{exp:scale-up-2}, the training time of
RegretNet was less than $13$ hours. 
\begin{figure}[t]
	\begin{subfigure}{0.99\textwidth}
		\centering
		\includegraphics[scale = 0.75]{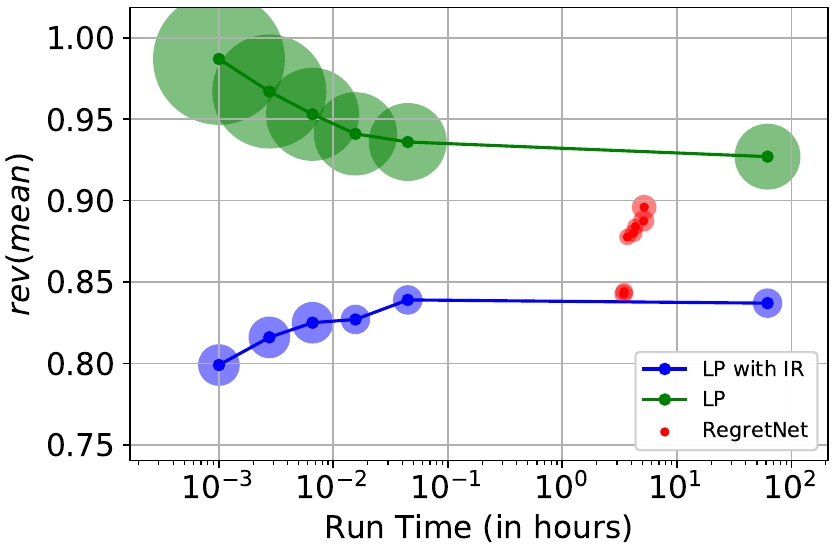}
	\end{subfigure}
	\caption{Test revenue vs.~training or solve time (in hours) in the two bidders, two items, additive $U[0,1]$ value setting. Comparing  the LP-based approach (with  {\tt -nearest} and {\tt -down} rounding, corresponding to ``LP" and ``LP with IR" respectively)  and RegretNet, and varying the number of variables in the LP by modifying the level of discretization, and   the number of parameters in RegretNet by modifying the network structure. 
	The size of a marker corresponds to the sum of regret and IR violations of each method.  For RegretNet, we only plot the results that lie on the efficient frontier.
 \label{fig:rgtnet_vs_lp_plot}}
\end{figure}

\ssredit{In Figure~\ref{fig:rgtnet_vs_lp_plot}, we plot the test revenue, test regret and the run-time of the LP-based 
  and RegretNet methods, while varying the number of variables in the LP  and  number of parameters in RegretNet.}  \dcpadd{For the LP, this is done by varying the discretization and for RegretNet, this is done by varying the network structure.} In  Appendix~\ref{APP:lp_comparison}, we include the complete set of results for varying the  discretization in the LP-based method and varying the  number of hidden layers and hidden units configurations in RegretNet. \dcpadd{Introducing an increasingly fine discretization into the LP-based method provides an initial increase in revenue in return for a modest increase in run time, but this gives way to a huge increase in run time with no effect on revenue. For RegretNet, the training time is relatively stable as the number of hidden layers and units per layer is varied, while larger networks bring a substantive increase in revenue. We only plot the results for RegretNet that lie on the efficient frontier, and refer to  Figure~\ref{fig:2x3_comp} for the full details.}
  \pdadd{Taken together these results show that RegretNet's performance substantially extends the  revenue-time Pareto frontier available from the LP method, obtaining higher revenue for a relatively modest training time.}
  


%
\section{Conclusion}\label{sec:conclusion}


In this paper, we have introduced a  new framework of differentiable economics for using neural networks
for economic discovery, specifically for the discovery of
revenue-optimal, multi-bidder and multi-item auctions.
%
%
We have demonstrated that standard machine learning pipelines can be
used to essentially re-discover all known, optimal auction designs,
and  to discover the design of auctions for settings out of reach
of theory and settings that are orders of magnitude larger than those
that can be solved through other computational approaches.
  \new{We also see promise for the framework in advancing
  economic theory, for example in supporting or refuting conjectures
  and as an assistant in guiding new economic discovery.}

  This framework has already inspired a great deal of follow-on work,
  in taking differentiable economics to additional domains and
  in scaling-up the methods to support  networks that 
 simultaneously handle multiple sizes of markets (number of bidders and number of items). 
\dcpadd{Looking ahead, there remain a number of interesting challenges. Beyond expanding the domains that are studied by differentiable economics, the methodological challenges include the interpretability of learned mechanisms, integrating additional structural regularities from economic theory, scaling up to larger economic systems, and providing robustness guarantees in the form of certificates for economic properties.  
Combinatorial auctions  (CAs) presents an especially important domain,  and one whose study we have only initiated here (see Appendix~\ref{sec:ca-architecture} and~\ref{sec:experCA}, for theory and experiment results for the case of CAs with two items). CAs are important to practice~\citep{huerta22}, and yet 
concerns around low revenue and their vulnerability to collusion~\citep{ausubel06x,ausubel06,day08,levin16,goeree16} 
mean that we lack a complete understanding even for the design of efficient auctions, never mind finding revenue-optimizing designs.}
%
%

      %

      \bibliographystyle{apalike}
       \bibliography{deep}

\appendix

\section{Additional Architectures}\label{APP:ADDITIONAL-ARCHITECTURES}

\new{In this appendix we present additional  network architectures, for 
  a multi-bidder single-item setting,
  and for a general multi-bidder
  multi-item setting with combinatorial valuations.}

\subsection{The MyersonNet Approach}
\label{app:myerson}

\new{We start by describing an architecture that yields optimal DSIC auction for selling a single item to multiple buyers.}

In the single-item setting,
each bidder holds a private value $v_i \in \R_{\geq 0}$ for the item.
We consider a randomized auction $(g,p)$ that maps a reported 
bid profile $b \in \R_{\geq 0}^n$ to a
vector of allocation probabilities $g(b) \in \R_{\geq 0}^n$, where
$g_i(b) \in \R_{\geq 0}$ denotes the probability that bidder $i$ is
allocated the item and $\sum_{i=1}^n g_i(b) \leq 1$. We shall represent 
the payment rule $p_i$ via a price conditioned on the item being
allocated to bidder $i$, i.e. $p_i(b) = g_i(b)\,t_i(b)$ for
some conditional payment function 
$t_i: \R_{\geq 0}^n \rightarrow \R_{\geq 0}$. 
The expected revenue of the auction, when bidders
are truthful, is  given by:
\begin{equation}
rev(g, p) \,=\, \E_{v \sim F}\bigg[\sum_{i=1}^n g_i(v)\,t_i(v)\bigg].
\label{eq:revenue-single-item}
\end{equation}  

The structure of the revenue-optimal auction is well understood for this setting.

\begin{thm}[\citet{Myerson81}]
There exist a collection of monotonically \new{non-decreasing} functions, $\bar{\phi_i}: \mathbb{R}_{\geq 0} \rightarrow \mathbb{R}$ called the ironed virtual valuation functions such that the optimal BIC auction for selling a single item is the DSIC auction that assigns the item to the buyer with the highest ironed virtual value $\bar{\phi}_i(v_i)$ %
provided that this value is non-negative, \new{with ties broken in an arbitrary value-independent manner, and charges the bidders according to $p_i(v_i) = v_i g_i(v_i) - \int_{0}^{v_i} g_i(t) \;dt$.} %
\end{thm}

\new{For distribution $F_i$ with density $f_i$ the virtual valuation function is $\psi_i(v_i) = v_i - (1-F(v_i))/f(v_i)$. A distribution $F_i$ with density $f_i$ is regular if $\psi_i$ is monotonically non-decreasing. For regular distributions $F_1, \dots, F_n$ no ironing is required and $\bar{\phi}_i = \psi_i$ for all $i$.}

\new{If the virtual valuation functions $\psi_1, \dots, \psi_n$ are furthermore monotonically increasing and not only monotonically non-decreasing,} the optimal auction 
can be viewed as applying the monotone transformations 
to the input bids
 $\bar{b}_i = \bar{\phi}_i(b_i)$, feeding the computed virtual values
 to a second price auction (SPA) with zero reserve price, denoted $(g^0, p^0)$, 
 making an allocation according to $g^0(\bar{b})$, and
 charging a payment $\bar{\phi}^{-1}_i(p^0_i(\bar{b}))$ for winning bidder $i$.
In fact, this auction is DSIC for any choice of strictly monotone \new{transformations of the values:} %
\begin{thm}
For any set of strictly monotonically increasing functions 
$\bar{\phi}_1, \ldots, \bar{\phi}_n$, an auction defined by
outcome rule $g_i = g^0_i \, \circ \, \bar{\phi}$ and payment rule $p_i = \bar{\phi}_i^{-1}\, \circ\, p^0_i \, \circ\, \bar{\phi}$ is DSIC and IR, where $(g^0, p^0)$ is the allocation and payment rule of a second price auction with zero reserve.
\end{thm}

 \new{For regular distributions with monotonically increasing virtual value functions designing an optimal DSIC auction thus reduces to finding the right strictly monotone transformations and corresponding inverses, and modeling a second price auction with zero reserve.} 
 
 \new{We present a high-level overview of a neural network architecture that achieves this in Figure~\ref{fig:advanced-auction-network}(a), and describe the components of this network in more detail in Section~\ref{sec:mon-trafo} and Section~\ref{sec:spa} below.}
 
 \new{The MyersonNet is tailored to monotonically increasing virtual
   value functions. For regular distributions with virtual value
   functions that are not strictly increasing and for irregular
   distributions this approach only yields approximately optimal
   auctions.}

 \begin{figure}[t]
\begin{center}
\vspace{-5pt}
\hspace{-15pt}
\begin{subfigure}[b]{0.45\textwidth}
\centering
\begin{tikzpicture}[scale=0.8, transform shape, shorten >=1pt,->,draw=black!100, node distance=\layersep, thick]
    \tikzstyle{neuron}=[circle,draw=black!100,minimum size=17pt,inner sep=0pt,thick]
    \tikzstyle{itext}=[draw=white,minimum size=17pt,inner sep=0pt]
    \tikzstyle{unit}=[draw=black!100,minimum size=22pt,inner sep=0pt,thick]
    \tikzstyle{spa}=[draw=black!100,minimum size=32pt,inner sep=0pt,thick]
    \tikzstyle{annot} = [text width=4em, text centered]

	\node[unit, pin={[pin edge={<-}]left:$b_1$}] (P-1) at (0,-1) {$\bar{\phi}_1$};
    \node[itext] (P-2) at (0,-1.9) {\vdots};
	\node[unit, pin={[pin edge={<-}]left:$b_n$}] (P-3) at (0,-3) {$\bar{\phi}_n$};

	\node[spa, pin={[pin edge={->}]right:$(z_1,\ldots,z_n)$}] (S-1) at (2,-1.3) {$g^0$};
	\node[spa] (S-2) at (2,-2.7) {$p^0$};
	
	\draw[thick,dotted] ($(S-1.north west)+(-0.1,0.1)$)  rectangle ($(S-2.south east)+(0.1,-0.1)$) ;	
	\node at  ($(S-1.north west) + (0.6,0.35)$) {SPA-0};
	
	\node[itext] (I-2) at (4,-2.6) {\vdots};
	\node[unit, pin={[pin edge={->}]right:$t_1$}] (I-1) at (4,-2.05) {$\bar{\phi}^{-1}_1$};
	\node[unit, pin={[pin edge={->}]right:$t_n$}] (I-3) at (4,-3.3) {$\bar{\phi}^{-1}_n$};
	
	\path (P-1) edge (S-1); 
	\path (P-1) edge (S-2); 
	
	\path (P-3) edge (S-1); 
	\path (P-3) edge (S-2);
	
	\path (S-2) edge (I-1); 
	\path (S-2) edge (I-3);
	
\end{tikzpicture}
\caption{}
\end{subfigure}
\begin{subfigure}[b]{0.45\textwidth}
\centering
\begin{tikzpicture}[scale=0.7, transform shape, shorten >=1pt,->,draw=black!100, node distance=\layersep, thick]
    \tikzstyle{input text}=[draw=white,minimum size=17pt,inner sep=0pt]
    \tikzstyle{input neuron}=[circle,draw=black!100,minimum size=17pt,inner sep=0pt,thick]
    \tikzstyle{hidden neuron}=[circle,draw=black!100,minimum size=17pt,inner sep=0pt,thick]
    \tikzstyle{hidden text}=[draw=white,minimum size=22pt,inner sep=0pt,thick]
    \tikzstyle{unit}=[draw=black!100,minimum size=22pt,inner sep=0pt,thick]

    \node[input neuron, pin={[pin edge={<-}]left:$b_i$}] (I-1) at (0,-2.5) {};

    \path node[hidden text] (H-0) at (1.5,-1.05) {\vdots};
    \path node[hidden neuron] (H-1) at (1.5,-0.5) {$h_{1,1}$};
    \path node[hidden neuron] (H-2) at (1.5,-1.7) {$h_{1,J}$};

    \path node[hidden text] (G-0) at (1.5,-4.05) {\vdots};
	\path node[hidden neuron] (G-1) at (1.5,-3.5) {$h_{K,1}$};
    \path node[hidden neuron] (G-2) at (1.5,-4.8) {$h_{K,J}$};

    \foreach \dest in {1,...,2}
		\path (I-1) edge (H-\dest); 

	\node[unit] (M-1) at (2.75,-1.1) {max};
	
	\foreach \src in {1,...,2}
		\path (H-\src) edge (M-1); 

    \foreach \dest in {1,...,2}
		\path (I-1) edge (G-\dest); 

	\node[unit] (M-2) at (2.75,-4.1) {max};
	
	\foreach \src in {1,...,2}
		\path (G-\src) edge (M-2); 
	
	\path node[hidden text] (T) at (1.5,-2.55) {\vdots};
	\path node[hidden text] (TT) at (2.75,-2.65) {\vdots};
	
	\node[unit, pin={[pin edge={->}]right:$\bar{b}_i$}] (M-3) at (4,-2.75) {min};
	
	\path (M-1) edge (M-3); 
	\path (M-2) edge (M-3); 

\end{tikzpicture}
\caption{}
\end{subfigure}
\vspace{-5pt}
\caption{(a) MyersonNet: The network applies monotone transformations $\bar{\phi}_{1}, \ldots, \bar{\phi}_{n}$ to the input bids, passes the virtual values to the SPA-0 network in Figure \ref{fig:basic-auction-network}, and applies the inverse transformations $\bar{\phi}^{-1}_{1}, \ldots, \bar{\phi}^{-1}_{n}$ to the payment outputs. (b) Monotone virtual
value function $\bar{\phi}_i$, where 
$h_{kj}(b_i) = e^{\alpha^i_{kj}}b_i + \beta^i_{kj}$.
}
\vspace{-15pt}
\label{fig:advanced-auction-network}
\end{center}
\end{figure}

\subsubsection{\new{Modeling Monotone Transforms}} \label{sec:mon-trafo}
We model each virtual value function $\bar{\phi}_i$ 
as a two-layer feed-forward network with min and 
max operations over linear functions.
For $K$ groups of $J$ linear functions, 
with strictly positive slopes $w^i_{kj} \in \R_{>0},~k = 1,\ldots,K,~ j = 1,\ldots, J$ and intercepts $\beta^i_{kj} \in \R,~k = 1,\ldots,K,~ j = 1,\ldots, J$,
we define:
\[
\bar{\phi}_{i}(b_i) \,=\, \min_{k \in [K]} \max_{j \in [J]}\, w^i_{kj}\,b_i + \beta^i_{kj}.
\]
Since each of the above linear function is strictly non-decreasing, so is
 $\bar{\phi}_{i}$. In practice, we can set each $w^i_{kj} = e^{\alpha^i_{kj}}$ for parameters
 $\alpha^i_{kj} \in [-B,B]$ in a bounded range. 
A graphical representation of the 
neural network used for this transform is shown in Figure \ref{fig:advanced-auction-network}(b). 
For sufficiently large $K$ and $J$, 
this neural network 
can be used to approximate any 
continuous, bounded monotone function (that
satisfies a mild regularity condition)
to an arbitrary degree of accuracy \citep{Sill98}.
A particular advantage of this representation is that the inverse transform 
$\bar{\phi}^{-1}$ can be directly obtained from the parameters for the forward transform:
\[
\bar{\phi}^{-1}_{i}(y) \,=\, \max_{k \in [K]} \min_{j \in [J]}\, e^{-\alpha^i_{kj}}(y - \beta^i_{kj}).
\]
\subsubsection{\new{Modeling SPA with Zero Reserve}} \label{sec:spa}
We also need to model a SPA with zero reserve (SPA-0) within the neural network structure. 
For the purpose of training, 
we  employ a smooth approximation to the allocation rule 
using a neural network. Once we learn value functions using this approximate allocation rule, 
we  use them together with an exact SPA with zero reserve to construct the final auction.

\begin{figure}[t]
\begin{center}
\def\layersep{2.5cm}

\begin{subfigure}[b]{0.5\textwidth}
\centering
\begin{tikzpicture}[scale=0.7, transform shape, shorten >=1pt,->,draw=black!100, node distance=\layersep, thick]
    \tikzstyle{input text}=[draw=white,minimum size=17pt,inner sep=0pt]
    \tikzstyle{input neuron}=[circle,draw=black!100,minimum size=17pt,inner sep=0pt,thick]
    \tikzstyle{hidden neuron}=[circle,draw=black!100,minimum size=17pt,inner sep=0pt,thick]
    \tikzstyle{hidden text}=[draw=white,minimum size=22pt,inner sep=0pt,thick]

    \foreach \name / \y in {1,...,2}
    	\node[input neuron, pin={[pin edge={<-}]left:$\bar{b}_\y$}] (I-\name) at (0,-\y) {};
    \node[input text] (I-0) at (0,-3) {$\vdots$};
    \node[input neuron, pin={[pin edge={<-}]left:$\bar{b}_n$}] (I-3) at (0,-4) {};
    \node[input neuron, pin={[pin edge={<-}]left:$0$}] (I-4) at (0,-5) {};

    \foreach \name / \y in {1,...,2}
        \path
            node[hidden neuron,pin={[pin edge={->}]right:$z_\y$}] (H-\name) at (\layersep,-\y) {};
    \path
            node[hidden neuron,pin={[pin edge={->}]right:$z_{n}$}] (H-4) at (\layersep,-5) {};
    \path
            node[hidden neuron,pin={[pin edge={->}]right:$z_{n-1}$}] (H-3) at (\layersep,-4) {};
    \path
            node[hidden text] (H-5) at (\layersep,-3 cm) {\vdots};

    \foreach \dest in {1,...,4}
		\path (I-1) edge (H-\dest); 
    \foreach \dest in {1,...,4}
		\path (I-2) edge (H-\dest);  
	\foreach \dest in {1,...,4}
		\path (I-3) edge (H-\dest);  
	\foreach \dest in {1,...,4}
		\path (I-4) edge (H-\dest);

	 \draw[thick,dotted] ($(H-1.north west)+(-0.2,0.3)$)  rectangle ($(H-4.south east)+(0.2,-0.3)$) ;
	 \node at  ($(H-1.north west) + (0.2,0.5)$) {$softmax$};
\end{tikzpicture}
\caption{Allocation rule $g^0$}
\end{subfigure}
\hspace{-1cm}
\begin{subfigure}[b]{0.5\textwidth}
\hspace{0.5cm}
\centering
\begin{tikzpicture}[scale=0.7, transform shape, shorten >=1pt,->,draw=black!100, node distance=\layersep, thick]
    \tikzstyle{input text}=[draw=white,minimum size=17pt,inner sep=0pt]
    \tikzstyle{input neuron}=[circle,draw=black!100,minimum size=17pt,inner sep=0pt,thick]
    \tikzstyle{hidden neuron}=[draw=black!100,minimum size=22pt,inner sep=0pt,thick]
    \tikzstyle{hidden text}=[draw=white,minimum size=22pt,inner sep=0pt,thick]

    \foreach \name / \y in {1,...,2}
    	\node[input neuron, pin={[pin edge={<-}]left:$\bar{b}_\y$}] (I-\name) at (0,-\y) {};
    \node[input text] (I-0) at (0,-3) {$\vdots$};
    \node[input neuron, pin={[pin edge={<-}]left:$\bar{b}_n$}] (I-3) at (0,-4) {};
    \node[input neuron, pin={[pin edge={<-}]left:$0$}] (I-4) at (0,-5) {};

    \foreach \name / \y in {1,...,2}
        \path
            node[hidden neuron,pin={[pin edge={->}]right:$t^0_\y$}] (H-\name) at (\layersep,-\y) {$max$};
    \path
            node[hidden neuron,pin={[pin edge={->}]right:$t^0_{n-1}$}] (H-3) at (\layersep,-4) {$max$};
    \path
            node[hidden neuron,pin={[pin edge={->}]right:$t^0_n$}] (H-4) at (\layersep,-5) {$max$};
    \path
            node[hidden text] (H-5) at (\layersep,-3) {\vdots};

    \foreach \dest in {2,3,4}
		\path (I-1) edge (H-\dest); 
    \foreach \dest in {1,3,4}
		\path (I-2) edge (H-\dest);  
	\foreach \dest in {1,2,4}
		\path (I-3) edge (H-\dest);  
	\foreach \dest in {1,2,3,4}
		\path (I-4) edge (H-\dest);  

\end{tikzpicture}
\caption{Payment rule $t^0$}
\end{subfigure}
\vspace{-5pt}
\caption{ 
MyersonNet: SPA-0 network for (approximately) modeling a second price auction with zero reserve price.
The inputs are (virtual) bids $\bar{b}_1,\ldots,\bar{b}_n$
and the output is a vector of assignment probabilities $z_1,\ldots,z_n$ and
prices (conditioned on allocation) $t^0_1,\ldots,t^0_n$.
\label{fig:basic-auction-network}}
\end{center}
\vspace{-5pt}
\end{figure}

The  SPA-0 allocation rule $g^0$ can be approximated using a
`softmax'  function on the virtual values $\bar{b}_1, \ldots, \bar{b}_n$ 
and an additional dummy input $\bar{b}_{n+1} = 0$: 
\begin{align}
g^0_i(\bar{b}) \,=\, \frac{e^{\kappa \bar{b}_i}}{\sum_{j=1}^{n+1} e^{\kappa \bar{b}_j}},
~ i \in N,
\end{align}
where $\kappa > 0$ is a constant fixed a priori, and determines the quality 
of the approximation. The higher the value of $\kappa$, the better the approximation
but the less smooth the resulting allocation function.

The SPA-0 payment to bidder $i$, conditioned on being allocated, is the maximum of the
 virtual values from the other bidders and zero:
\begin{align}
t^0_i(\bar{b})\,=\, \max\big\{\max_{j \ne i} \bar{b}_j, \, 0\big\},~ i \in N.
\end{align}

Let $g^{\alpha,\beta}$ and $t^{\alpha,\beta}$ denote the
 allocation and conditional payment rules for the overall auction in Figure \ref{fig:advanced-auction-network}(a), where
$(\alpha,\beta)$ are the parameters of the forward monotone transform. 
Given a sample of valuation profiles $\mathcal{S} =\{v^{(1)}, \ldots, v^{(L)}\}$ drawn i.i.d.\ from $F$, 
we optimize the parameters using the negated revenue on $\mathcal{S}$ as the error function, where the revenue is approximated as:
\begin{equation}
\widehat{rev}(g,t) \,=\, 
\frac{1}{L}\sum_{\ell=1}^L\sum_{i=1}^n g^{\alpha,\beta}_i(v^{(\ell)})\,t^{\alpha,\beta}_i(v^{(\ell)}).
\end{equation}

We solve this training problem using a minibatch stochastic gradient descent solver. 

\subsection{RegretNet for Combinatorial Valuations}\label{sec:ca-architecture}

\new{We next show how to adjust the RegretNet architecture to handle}
bidders with general, combinatorial valuations. \footnote{In the present work,
we develop this architecture only for small number of
items.
With more items, combinatorial valuations can be
  succinctly represented using appropriate bidding languages; see,
  e.g.\ \citep{BH01}.}

In this case, each bidder $i$ reports a bid $b_{i, S}$ for
every bundle of items $S \subseteq M$ (except the  empty bundle, for
which her valuation is taken as zero).  The allocation component of
the network has an output
$z_{i, S} \in [0,1]$ for each bidder $i$ and bundle $S$,
denoting the probability that the bidder is allocated the
bundle.

{To prevent the items from being over-allocated, we require that the probability that an item appears 
in a bundle allocated to some bidder is at most one. We also require that the total allocation to a bidder is at most one:}
\begin{align}
&\sum_{i \in N}\sum_{S \subseteq M: j \in S } z_{i,S} \,\leq\, 1,\, \forall j \in M;
\label{eq:combinatorial-constraints-1}
\\
&\sum_{S \subseteq M} z_{i,S} \,\leq\, 1,\, \forall i \in N.
\label{eq:combinatorial-constraints-2}
\end{align}

We refer to an allocation that satisfies constraints \eqref{eq:combinatorial-constraints-1}--\eqref{eq:combinatorial-constraints-2} as being \textit{combinatorial feasible}.
To enforce these constraints, the allocation component of the network computes a set of scores for each bidder and a set of scores for each item. Specifically, there is a group of bidder-wise scores $s_{i,S}, \forall S\subseteq M$ for each bidder $i \in N$, and a group of item-wise scores $\smash{s^{(j)}_{i,S}, \forall i \in N,\, S\subseteq M}$ for each item $j \in M$. 

\new{Let $s, s^{(1)}, \dots, s^{(m)} \in \mathbb{R}^{n \times 2^m}$ denote these bidder scores and item scores.}   
Each group of scores is normalized using a softmax function:
$
\bar{s}_{i, S} = {\exp({s_{i,S}})}/{\sum_{S'} \exp({s_{i,S'}})}$ and $\bar{s}^{(j)}_{i, S} = {\exp({s^{(j)}_{i,S}})}/{\sum_{i',S'} \exp({s^{(j)}_{i',S'}})}.
$
The allocation for bidder $i$ and bundle $S \subseteq M$ is defined as the minimum of the normalized bidder-wise score $\bar{s}_{i,S}$ and the normalized item-wise scores $\smash{\bar{s}^{(j)}_{i, S}}$ for each $j \in S$: %
\begin{align}
{z_{i,S} \,=\, \varphi^{CF}_{i,S}({s}, {s}^{(1)},\ldots,{s}^{(m)}) \,=\, \min\big\{\bar{s}_{i,S}, \, \bar{s}^{(j)}_{i,S}:\, j \in S\big\}}.
\end{align}

\new{Similar to the unit-demand setting, we first show that $\varphi^{CF}({s}, {s}^{(1)},\ldots,{s}^{(m)})$ is combinatorial feasible and that our constructive approach is without loss of generality. See Appendix~\ref{APP:CA_DS} for a proof.}

\begin{lem}\label{LEM:CA_DS}
\new{The matrix} $\varphi^{CF}({s}, {s}^{(1)},\ldots,{s}^{(m)})$ is combinatorial feasible $\forall\, {s},$ ${s}^{(1)}, \ldots,{s}^{(m)} \in \R^{n \times 2^m}$.  For any combinatorial feasible \new{matrix} $z \in [0,1]^{n \times 2^m}$, $\exists\, {s}, {s}^{(1)},\ldots,{s}^{(m)} \in \R^{n \times 2^m}$, for which $z = \varphi^{CF}({s}, {s}^{(1)},\ldots,{s}^{(m)})$.
\end{lem}

In addition, we want to understand whether a \emph{combinatorial feasible} allocation $z$ can be \emph{implementable}, defined in the following way.

\begin{definition} 
A fractional combinatorial allocation $z$ is implementable if and only if $z$ can be represented as a convex combination of combinatorial feasible,  deterministic allocations.
\end{definition}

Unfortunately, Example~\ref{ex:1} shows that a combinatorial feasible allocation may not have an integer
decomposition, even for the case of two bidders and two items.
\begin{example}
	\label{ex:1}
	Consider a setting with two bidders and two items, and the following fractional, combinatorial feasible allocation:
	
	$$z = \begin{bmatrix}
	z_{1,\{1\}} & z_{1,\{2\}} & z_{1, \{1,2\}}\\
	z_{2,\{1\}} & z_{2,\{2\}} & z_{2, \{1,2\}}
	\end{bmatrix}=\begin{bmatrix}
	3/8 & 3/8 & 1/4\\
	1/8 & 1/8 & 1/4
	\end{bmatrix}
	$$
	
	Any integer decomposition of this allocation $z$ would need to have the following structure:
	\begin{eqnarray*}
	z &=& a\begin{bmatrix}
	0 & 0 & 1\\
	0& 0 & 0
	\end{bmatrix} + b\begin{bmatrix}
	0 & 0 & 0\\
	0& 0 & 1
	\end{bmatrix} + c\begin{bmatrix}
	1 & 0 & 0\\
	0& 1 & 0
	\end{bmatrix} + d\begin{bmatrix}
	1 & 0 & 0\\
	0& 0 & 0
	\end{bmatrix} + e\begin{bmatrix}
	0 & 0 & 0\\
	0 & 1 & 0
	\end{bmatrix}\\
    && + f\begin{bmatrix}
	0 & 1 & 0\\
	1 & 0 & 0
	\end{bmatrix} + g\begin{bmatrix}
	0 & 1 & 0\\
	0 & 0 & 0
	\end{bmatrix} + h\begin{bmatrix}
	0 & 0 & 0\\
	1 & 0 & 0
	\end{bmatrix}
	\end{eqnarray*}
	
	where the coefficients sum to at most 1.
	Firstly, it is straightforward to see that $a = b = 1/4$. Given the construction, we must have $c + d = 3/8, e\geq 0$ and $f + g = 3/8, h \geq 0$. Thus, $a+b+c+d+e+f+g+h \geq 1/2 + 3/4 = 5/4$ for any decomposition. Hence, $z$ is not implementable.
\end{example}

To ensure that a combinatorial feasible allocation has an integer
decomposition
we need to introduce additional constraints. For the two items case,
we introduce the following constraint:
\begin{eqnarray}\label{eq:additional-constraints-2-items}
\forall i, z_{i, \{1\}} + z_{i, \{2\}} \leq 1 -\sum_{i'=1}^{n} z_{i', \{1, 2\}}.
\end{eqnarray}

\begin{thm}
	For $m=2$, any combinatorial feasible allocation $z$ with additional
	constraints~(\ref{eq:additional-constraints-2-items}) can be
	represented as a convex combination of matrices $B^1, \dots, B^k$
	where each $B^\ell$ is a combinatorial feasible, 0-1 allocation.
\end{thm}

\begin{proof}
	Firstly, we observe in any deterministic allocation $B^\ell$, if there exists an $i$, s.t. $B^\ell_{i, \{1, 2\}} = 1$, then $\forall j \neq i, S: B^\ell_{j, S} = 0$. Therefore, we first decompose $z$ into the following components,
	\begin{eqnarray*}
		z = \sum_{i=1}^n z_{i, \{1, 2\}}\cdot B^{i} + C,
	\end{eqnarray*}
	and
	$$
	B^i_{j, S} = \left\{
	\begin{array}{ll}
	1  & \text{if  $j=i, S =\{1, 2\}$, and} \\
	0 & \text{otherwise.}
	\end{array}
	\right.
	$$
	
	Then we want to argue that $C$ can be represented as $\sum_{\ell = i+1}^{k} p_\ell \cdot B^\ell$, where $\sum_{\ell = i+1}^{k} p_\ell \leq 1 - \sum_{i=1}^n z_{i, \{1, 2\}}$ and each $B^\ell$ is a feasible 0-1 allocation. Matrix $C$ has all zeros in the last (items $\{1, 2\}$) column, $\sum_{i} C_{i, \{1\}} \leq 1 - \sum_{i=1}^n z_{i, \{1, 2\}}$, and $\sum_{i} C_{i, \{2\}} \leq 1 - \sum_{i=1}^n z_{i, \{1, 2\}}$. 
	
	In addition, based on constraint~(\ref{eq:additional-constraints-2-items}), for each bidder $i$,
	\begin{eqnarray*}
		C_{i, \{1\}} + C_{i, \{2\}} =  z_{i, \{1\}} + z_{i, \{2\}} \leq 1 - \sum_{i'=1}^n z_{i', \{1, 2\}}.
	\end{eqnarray*}
	
	Thus $C$ is a doubly stochastic matrix with scaling factor $1 - \sum_{i'=1}^n z_{i', \{1, 2\}}$. Therefore, we can always decompose $C$ into a linear combination $\sum_{\ell = i+1}^{k} p_\ell \cdot B^\ell$, where $\sum_{\ell = i+1}^{k} p_\ell \leq 1 - \sum_{i'=1}^n z_{i', \{1, 2\}}$ and each $B^\ell$ is a feasible 0-1 allocation.
\end{proof}

We leave to future work to characterize the additional constraints
needed for the multi-item ($m>2$) case.

\subsubsection{RegretNet for Two-item Auctions with Implementable Allocations}
To accommodate the additional
constraint~(\ref{eq:additional-constraints-2-items}) for the two items
case we add an additional softmax layer for each bidder. In addition
to the original (unnormalized) bidder-wise scores $s_{i, S}, \forall i
\in N, S\subseteq M$ and item-wise scores $s^{(j)}_{i, S}, \forall i
\in N, S\subseteq M, j\in M$ and their normalized counterparts
$\bar{s}_{i, S}, \forall i \in N, S\subseteq M$ and $\bar{s}^{(j)}_{i,
  S}, \forall i \in N, S\subseteq M, j\in M$, the allocation component
of the network computes an additional set of scores for each bidder $i$, ${s'}^{(i)}_{i, \{1\}}, {s'}^{(i)}_{i, \{2\}}, {s'}^{(i)}_{1, \{1, 2\}}, \cdots, {s'}^{(i)}_{n, \{1, 2\}}$. These additional scores are then normalized using a softmax function as follows,
\begin{eqnarray*}
	\forall i, k \in N, S\subseteq M, &
	\bar{s'}^{(i)}_{k, S} = \frac{\exp\left({s'}^{(i)}_{k, S} \right)}{\exp\left({s'}^{(i)}_{i, \{1\}}\right) + \exp\left({s'}^{(i)}_{i, \{2\}}\right) + \sum_k \exp\left({s'}^{(i)}_{k, \{1, 2\}} \right)}.
\end{eqnarray*}

To satisfy constraint~(\ref{eq:additional-constraints-2-items}) for each bidder $i$, we compute the normalized score $\bar{s'}_{i, S}$ for each $i, S$ as,
$$
\bar{s'}_{i, S} = \left\{
\begin{array}{ll}
\bar{s'}^{(i)}_{i, S} & \text{ if } S=\{1\} \text{ or } \{2\}, \text{and}\\
\min\left\{\bar{s'}^{(k)}_{i, S}: k\in N\right\} & \text{ if } S=\{1, 2\}.
\end{array}
\right.
$$

Then the final allocation for each bidder $i$ is:
\begin{eqnarray*}
z_{i, S} =  \min\left\{\bar{s}_{i, S}, \bar{s'}_{i, S}, \bar{s}^{(j)}_{i, S}: j\in S\right\}.
\end{eqnarray*}

The payment component of the network for combinatorial bidders has the same structure as the one in Figure \ref{fig:gen-net}, 
computing a fractional payment $\tilde{p}_i \in [0,1]$ for each bidder $i$ using a sigmoidal unit, and outputting  a payment $p_i = \tilde{p}_{i}\, \sum_{S\subseteq M} z_{i,S}\, b_{i,S}$. %

\section{Additional Experiments}\label{APP:ADDITIONAL-EXPERIMENTS}

\new{We present a broad range of additional experiments for the two main architectures used in the body of the paper, and additional ones for the
  architectures presented in Appendix~\ref{APP:ADDITIONAL-ARCHITECTURES}}
\subsection{Experiments with MyersonNet}\label{app:expt-myersonnet}

We \new{first} evaluate the MyersonNet \new{architecture introduced in Appendix~\ref{app:myerson}} for designing single-item auctions. \new{We focus on settings with a small number of bidders because this is where
revenue-optimal auctions are meaningfully
different from efficient auctions.}  
\new{We present experimental results for the following four settings:}
\begin{enumerate}[label=\Alph*.,ref=\Alph*,start=7]
\item \new{Three bidders with independent, regular, and symmetrically distributed valuations $v_i \sim U[0,1]$. \label{exp:myer-1}}
\item \new{Five bidders with independent, regular, and asymmetrically distributed valuations $v_i \sim U[0,i]$. \label{exp:myer-2}}
\item \new{Three bidders with independent, regular, and symmetrically distributed valuations $v_i \sim Exp(3)$.\label{exp:myer-3}}
\item \new{Three bidders with independent irregular distributions $F_\text{irregular}$, where each $v_i$ is drawn from $U[0,3]$ with probability 3/4 and from $U[3,8]$ with probability 1/4. \label{exp:myer-4}}
\end{enumerate}

\new{We note that the} optimal auctions for the first three distributions involve virtual value
functions $\bar{\phi}_i$ that are strictly monotone. For the fourth and final distribution
the optimal auction uses ironed virtual value functions that are not strictly monotone.

For the training set and test set \new{we used} 1,000 valuation
profiles sampled i.i.d.~from the respective valuation distribution. 
We \new{modeled} each
transform $\bar{\phi}_i$ in the MyersonNet architecture
using 5 sets of 10 linear functions, and we used $\kappa = 10^3$.

\begin{figure}[t]
\centering
\begin{tabular}{|l|c|c|c|c|}
\hline
\multirow{2}{*}{Distribution} & 
\multirow{2}{*}{$n$} & 
Opt & 
SPA 
& MyersonNet
\\
\cline{3-5}
&
& $\mathit{rev}$  
& $\mathit{rev}$
& $\mathit{rev}$\\
\hline
Setting \ref{exp:myer-1} & 3
 & 0.531 & 0.500 & 0.531\\
Setting \ref{exp:myer-2}  & 5
& 2.314 & 2.025  & 2.305\\
Setting \ref{exp:myer-3}  & 3
& 2.749 & 2.500 & 2.747\\
Setting \ref{exp:myer-4}  & 3
& 2.368 & 2.210 & 2.355 \\
\hline
\end{tabular}
\caption{The test revenue of the single-item auctions obtained with MyersonNet.
\label{tab:single-item}}
\vspace{-15pt}
\end{figure}

The results are summarized in \new{Figure~\ref{tab:single-item}.} For
comparison, we also report the revenue obtained by the optimal Myerson
auction and the second price auction (SPA) without reserve. The auctions learned by the neural network yield revenue close to the optimal.

\subsection{\new{Additional Experiments with RochetNet and RegretNet}}

\new{In addition to the experiments with RochetNet and RegretNet on the single bidder, multi-item settings in Section~\ref{sec:manelliandvincent} we also considered the following settings:}

\begin{enumerate}[label=\Alph*.,ref=\Alph*, start=11]
\item Single additive bidder with independent preferences over two non-identically distributed items, where
$v_1 \sim U[4,16]$ and $v_2 \sim U[4,7]$. The optimal mechanism is given by \citet{DaskalakisEtAl17}. \label{SII}
\item Single additive bidder with preferences over two items, where
$(v_1, v_2)$ are drawn jointly and uniformly from a unit triangle with vertices $(0,0), (0,1)$ and $(1,0)$. The optimal mechanism is due to \citet{HaghpanahH15}. \label{SIII}
\item Single unit-demand bidder with independent preferences over two items, where the item values $v_1, v_2 \sim U[0,1]$. See~\cite{Pavlov11} for the optimal mechanism. \label{SIV}
\end{enumerate}

\new{We used RegretNet architectures with two hidden layers with 100 nodes each. The optimal allocation rules as well as a side-by-side comparison of those found by RochetNet and RegretNet are given in Figure~\ref{fig:alloc-extended}. Figure~\ref{fig:regretnet-numbers} gives the revenue and regret achieved by RegretNet and the revenue achieved by RochetNet.}

\new{We find that in all three settings RochetNet recovers the optimal mechanism basically exactly, while RegretNet finds an auction that matches the optimal design to surprising accuracy.}

\begin{figure}[t]
\centering
\centering
	\begin{subfigure}{0.49\textwidth}
		\centering
		\hspace*{-10pt}
		\includegraphics[scale=0.4]{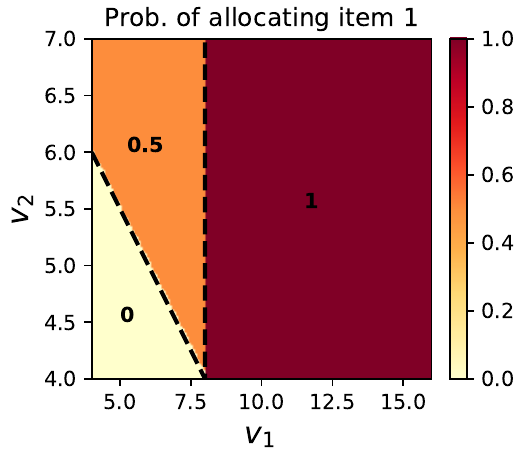}
		\hspace*{-10pt}{\scriptsize (a)}\hspace*{-2pt}
		\includegraphics[scale=0.4]{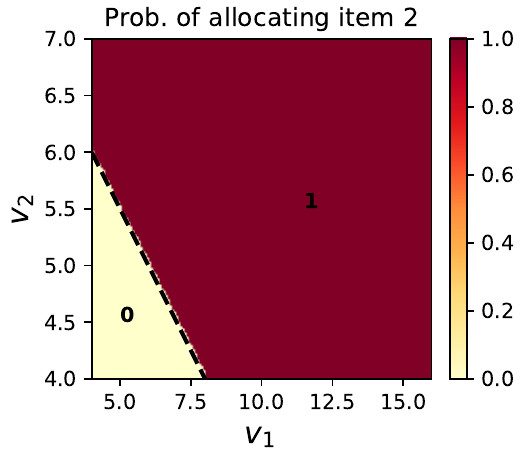}
	\end{subfigure}
	\begin{subfigure}{0.49\textwidth}
	\centering
	\hspace*{-10pt}
	\includegraphics[scale=0.4]{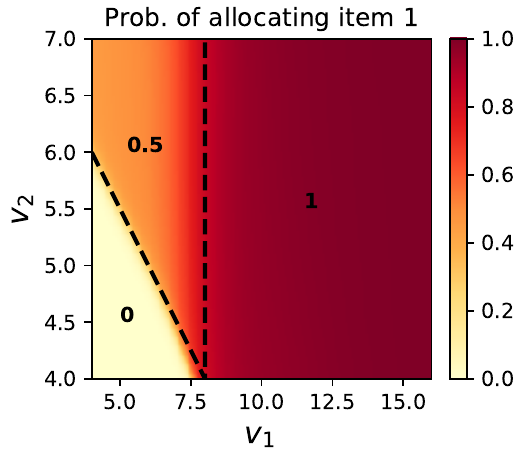}
	\hspace*{-8pt}{\footnotesize (b)}%
	\includegraphics[scale=0.4]{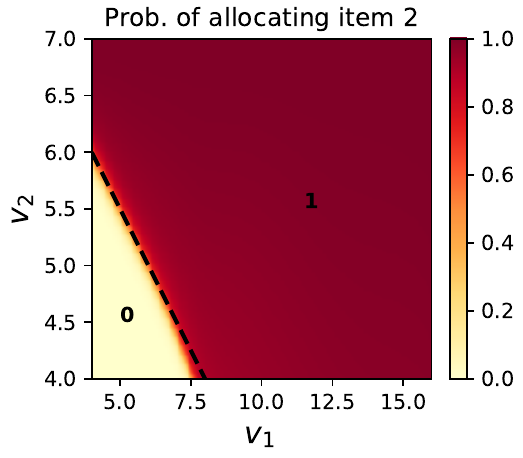}
	\end{subfigure}
	\begin{minipage}{0.49\textwidth}
		\centering
		\hspace*{-10pt}
		\includegraphics[scale=0.4]{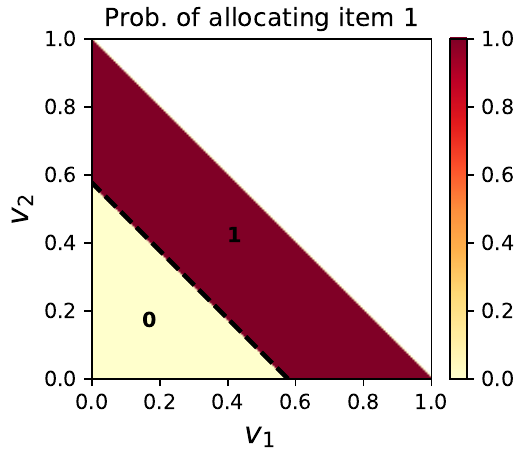}
		\hspace*{-10pt}{\scriptsize (c)}\hspace*{-2pt}
		\includegraphics[scale=0.4]{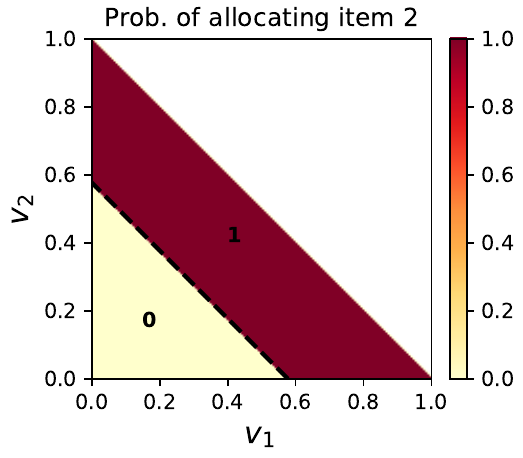}
	\end{minipage}
	\begin{subfigure}{0.49\textwidth}
	\centering
	\hspace*{-10pt}
	\includegraphics[scale=0.4]{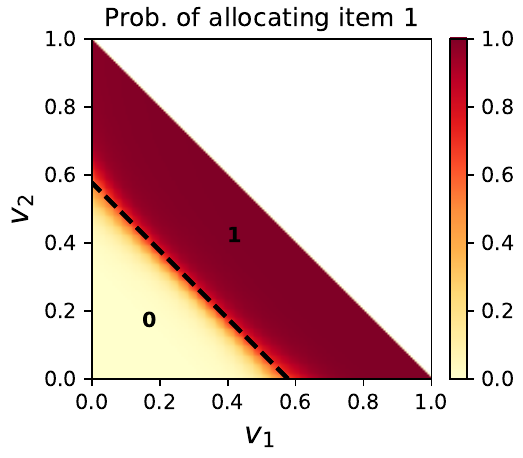}
	\hspace*{-6pt}{\footnotesize (d)}%
	\includegraphics[scale=0.4]{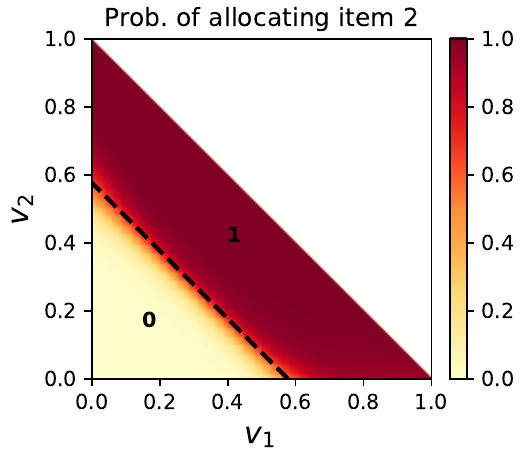}
	\end{subfigure}
	\begin{subfigure}{0.49\textwidth}
		\centering
		\hspace*{-10pt}
		\includegraphics[scale=0.4]{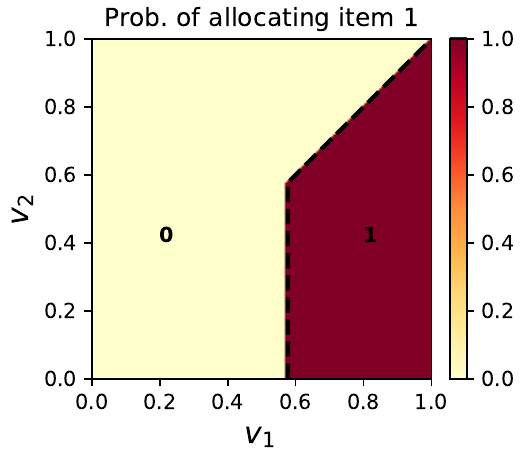}
		\hspace*{-10pt}{\scriptsize (e)}\hspace*{-2pt}
		\includegraphics[scale=0.4]{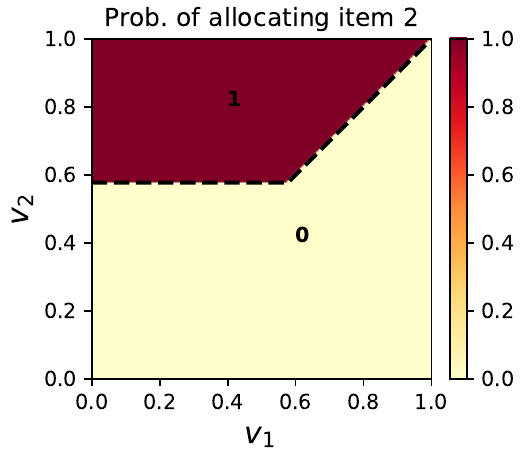}
	\end{subfigure}
	\begin{subfigure}{0.49\textwidth}
	\centering
	\hspace*{-10pt}
	\includegraphics[scale=0.4]{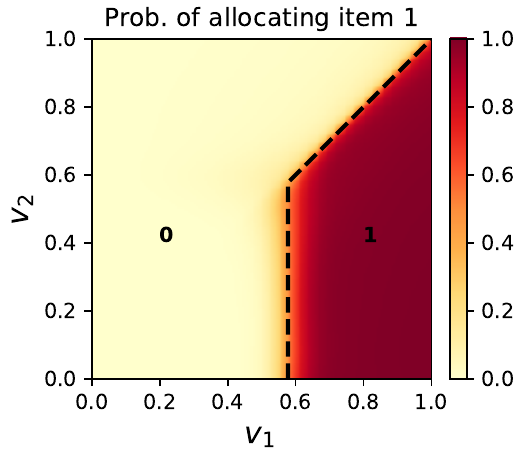}
	\hspace*{-6pt}{\footnotesize (f)}%
	\includegraphics[scale=0.4]{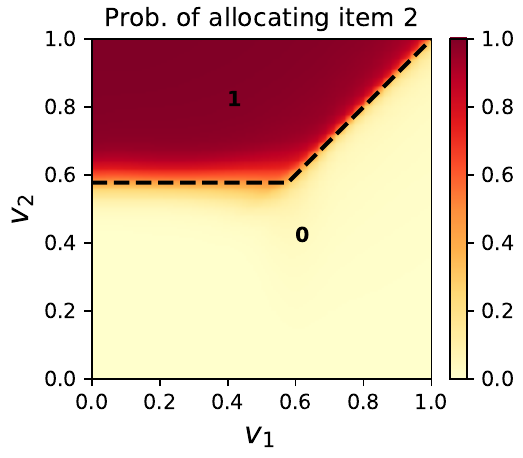}
	\end{subfigure}
\caption{%
\new{Side-by-side comparison of the allocation rules learned by RochetNet and RegretNet for single bidder, two items settings. Panels (a) and (b) are for Setting~\ref{SII}, Panels (c) and (d) are for Setting~\ref{SIII}, and Panels (e) and (f) for Setting~\ref{SIV}. Panels describe the learned allocations for the two items (item 1 on the left, item 2 on the right). Optimal mechanisms are indicated via dashed lines and allocation probabilities in each region.}
\label{fig:alloc-extended}}
\vspace{-17pt}
\end{figure}

\begin{figure}[t]
\begin{center}
\begin{tabular}{|l|c|c|c|c|c|c|c|}
\hline
\multirow{2}{*}{Distribution} & Opt & \multicolumn{5}{|c|}{RegretNet} & RochetNet \\
\cline{2-8} & $\mathit{rev}$ & $\mathit{rev}$ & $\mathit{rgt (mean)}$& $\mathit{rgt (90\%)}$ &$\mathit{rgt(95\%)}$& $\mathit{rgt(99\%)}$ & $\mathit{rev}$ \\
\hline
Setting \ref{SII}
& 9.781  & 9.734 & $< 0.001$ & $< 0.001$ & $< 0.001$ &  0.001 & 9.779 \\
Setting \ref{SIII}
& 0.388  & 0.392  & $< 0.001$ & $< 0.001$ & $< 0.001$ & 0.001 & 0.388 \\
Setting \ref{SIV}
& 0.384 & 0.384 & $< 0.001$ & $< 0.001$ & $< 0.001$ &  0.001 & 0.384 \\
\hline
\end{tabular}
\end{center}
\caption{Test revenue and test regret achieved by RegretNet  and test revenue achieved by RochetNet for \new{Settings \ref{SII}--\ref{SIV}.}}\label{fig:regretnet-numbers}
\vspace{-5pt}
\end{figure}

\subsection{\new{Experiments with RegretNet with Combinatorial Valuations}}
\label{sec:experCA}
We next compare \new{our RegretNet architecture for combinatorial valuations described in Section~\ref{sec:ca-architecture} to} the computational results of \citet{SandholmL15}
for the following settings for which the optimal auction is not known:
\begin{enumerate}[label=\Alph*.,start=14,ref=\Alph*]
\item %
\new{Two} additive bidders and \new{two} items, where  bidders draw their value for each item \new{independently} from $U[0,1]$.\footnote{\new{This setting can be handled by the non-combinatorial RegretNet architecture and is included here for comparison to \cite{SandholmL15}.}}\label{exp:comb-1}%
\item %
\new{Two} bidders and \new{two} items, with item valuations $v_{1,1}, v_{1,2}, v_{2,1}, v_{2,2}$ \new{drawn independently from} $U[1,2]$ and set valuations $v_{1,\{1,2\}} = v_{1,1}+v_{1,2} + C_1$ and $v_{2,\{1,2\}} = v_{2,1}+ v_{2,2} + C_2$, where $C_1, C_2$ are drawn independently from $U[-1,1]$. \label{exp:comb-2}%
 \item \new{Two} bidders and \new{two} items, with item valuations
 $v_{1,1}, v_{1,2}$ drawn independently from $U[1,2]$, item valuations $v_{2,1}, v_{2,2}$ drawn independently from $U[1,5]$, and set valuations $v_{1,\{1,2\}} = v_{1,1} + v_{1,2} + C_1$ and $v_{2,\{1,2\}} = v_{2,1}+ v_{2,2} + C_2$, where $C_1, C_2$ are drawn independently from $U[-1,1]$. \label{exp:comb-3}
\end{enumerate}

\new{These settings correspond to Settings I.-III.~described in Section 3.4 of \cite{SandholmL15}. These authors conducted extensive experiments with several different classes of incentive compatible mechanisms, and different heuristics for setting the parameters of these auctions. They observed the highest revenue for two classes of mechanisms that generalize mixed bundling auctions and $\lambda$-auctions \citep{JehielMM07}.}

\new{These two classes of mechanisms are the \emph{Virtual Value Combinatorial Auctions} ($\text{VVCA}$) and \emph{Affine Maximizer Auctions} ($\text{AMA}$). They also considered a restriction of $\text{AMA}$ to bidder-symmetric auction ($\text{AMA}_{\text{bsym}}$). We use   $\text{VVCA}^*$, $\text{AMA}^*$, and~$\text{AMA}^*_{\text{bsym}}$ to denote the best mechanism in the respective class, as reported by~\citeauthor{SandholmL15} and found using a heuristic grid search technique.}

\new{For Setting~\ref{exp:comb-1} and \ref{exp:comb-2}, \citeauthor{SandholmL15} observed the highest revenue for $\text{AMA}^*_{\text{bsym}}$, and for Setting~\ref{exp:comb-3} the best performing mechanism was $\text{VVCA}^*$.}
\new{Figure \ref{fig:comb} compares the performance of RegretNet to that of these best performing, benchmark mechanisms. To compute the revenue of the benchmark mechanisms we used the parameters reported in~\citet{SandholmL15} (Table 2, p.\ 1011), and evaluated the respective mechanisms on the same test set %
used for RegretNet.}
\new{Note that RegretNet is able to learn new auctions with improved revenue and tiny regret.}

\begin{figure}
	\begin{subfigure}{0.99\textwidth}
		\centering
		\includegraphics[scale=0.5]{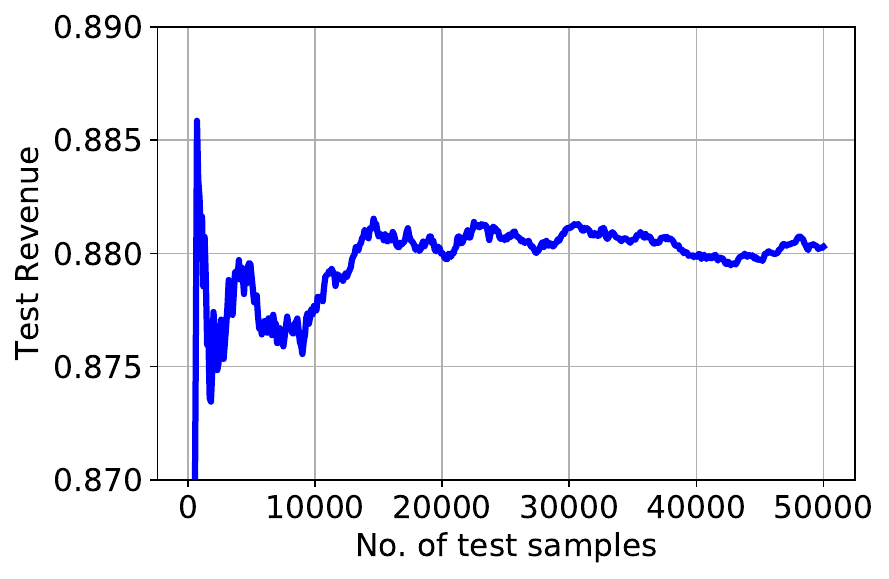}
		\includegraphics[scale=0.5]{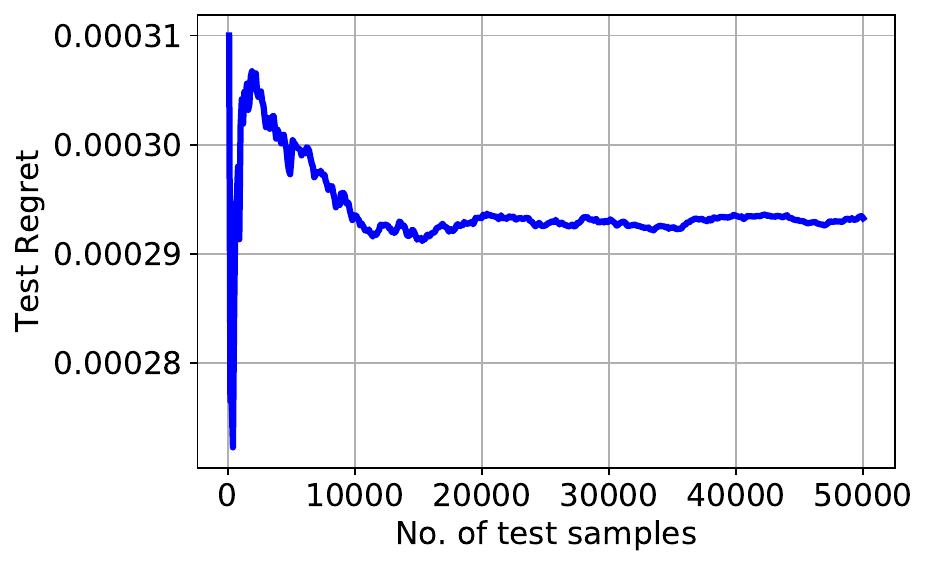}
		~\\[-5pt]
	\end{subfigure}
	\caption{Test revenue and test regret achieved by RegretNet for Setting~\ref{exp:comb-1} as we increase the size of the test set. \label{fig:large-test}}
\end{figure}

\ssredit{In order to make sure we are using sufficient data to report our results, we re-ran our evaluation for setting~\ref{exp:comb-1} on a bigger test set with up to 50,000 samples and computed the regret using 5,000 gradient ascent steps. The estimated revenue and regret remained approximately the same as that observed on our regular test set with 10,000 samples with regret computed using 2,000 gradient ascent steps. Figure~\ref{fig:large-test} shows how the revenue and regret varies as we increase the size of the test set.} 
\begin{figure}[t]
\centering
\begin{tabular}{|c|c|c|c|c|c|c|c|}
\hline
\multirow{2}{*}{Distribution} & 
\multicolumn{5}{|c|}{RegretNet} & $\text{VVCA}^*$ & $\text{AMA}^*_\text{bsym}$ \\

\cline{2-6} & $\mathit{rev}$ & $\mathit{rgt (mean)}$& $\mathit{rgt (90\%)}$ & $\mathit{rgt(95\%)}$ & $\mathit{rgt(99\%)}$ & $\mathit{rev}$ & $\mathit{rev}$ \\
\hline

Setting~\ref{exp:comb-1} & 0.878 & $<0.001$ & $<0.001$ & $<0.001$ &  0.001 & --- & 0.862 \\
Setting~\ref{exp:comb-2} & 2.860 & $<0.001$ & $<0.001$ & $<0.001$ &  $<0.001$ & --- & 2.765 \\
Setting~\ref{exp:comb-3} & 4.269 & $<0.001$ & $<0.001$ & $<0.001$ &  $<0.001$ & 4.209 & --- \\

\hline
\end{tabular}
\vspace{-2pt}
\caption{\new{Test revenue and test regret for RegretNet for Settings~\ref{exp:comb-1}--\ref{exp:comb-3} and a comparison with the best performing VVCA and $\text{AMA}_\text{bsym}$ auctions as reported by~\citet{SandholmL15}.} \label{fig:comb}}
\vspace{-4pt}
\end{figure}

\subsection{Experiments with RochetNet with varying linear units}\label{app:rochet-overparam}

\ssredit{In Figure~\ref{fig:6items}, we show how the performance of RochetNet varies as we increase the number of initialized menu choices (i.e., number of units in the network). We consider here a single bidder and six items, where the bidder's valuation is sampled independently $U[0,1]$ for each item. The optimal mechanism is given by the Straight-Jacket Auction (SJA). We observe that RochetNet recovers the optimal design with increasing accuracy as we increase the number of menu choices (units in the network) and even while only a small fraction of the menu choices are active ( $<3\%$ active when the number of initialized menu choices are over a 1000).} \ssredit{When we also impose item-symmetry, we observed that the performance of RochetNet is relatively invariant to increasing the number of initialized menu choices.} 
\begin{figure}
	\begin{subfigure}{0.99\textwidth}
		\centering
		\includegraphics[scale=0.5]{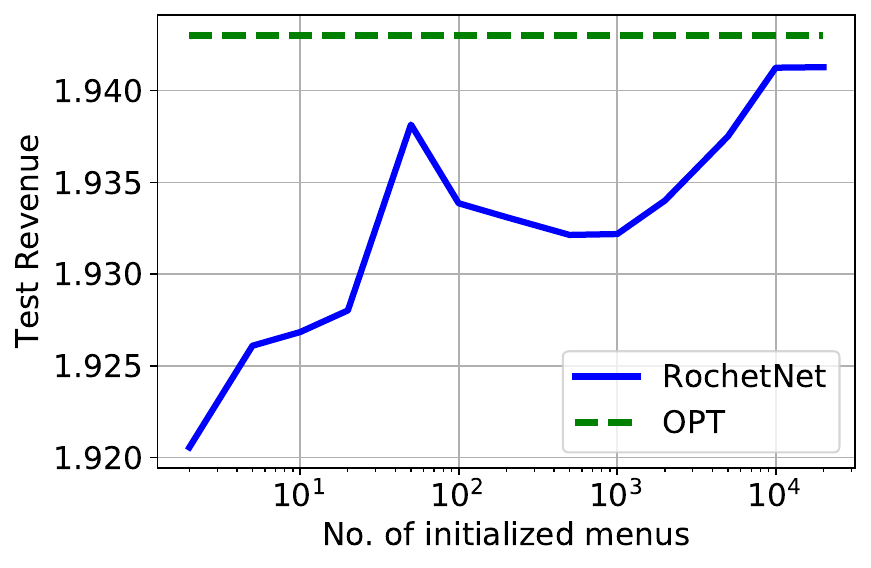}
		~\\[-5pt]
	\end{subfigure}
	\caption{Test revenue achieved by RochetNet for different numbers of initialized menu choices (units in the neural network). \label{fig:6items}}
\end{figure}

\subsection{Additional Experiments with discovering new analytical results }\label{app:additional-tri}

\ssredit{In Section~\ref{sec:experiments}, we described how RochetNet can be used to discover new analytical results for optimal auctions. In this section, we give analogous computational results, again suggestive of the structure of theoretically-optimal auction designs, for two such additional settings:}

\begin{enumerate}[label=\Alph*.,start=17,ref=\Alph*]
		\item One additive bidder and two items, where the bidder's valuation is drawn uniformly from the triangle $T=\{(v_1, v_2)|v_1 +v_2(c-1) \leq 2c - 1, v_1\geq 1, v_2\geq 1\}$ where $c\geq1$ is a free parameter \label{exp:triangle-2}
		\item One additive bidder and two items, where the bidder's valuation is drawn uniformly from the triangle $T=\{(v_1, v_2)|v_1 + v_2 \leq c + 1, v_1\geq 1, v_2\geq 1\}$ where $c\geq1$ is a free parameter. \label{exp:triangle-3}
\end{enumerate}

\ssredit{The mechanisms learned by RochetNet for setting~\ref{exp:triangle-2} and setting~\ref{exp:triangle-3} for various values of $c$ are shown in Figure~\ref{fig:alloc-triangle-setting-2} and Figure~\ref{fig:alloc-triangle-setting-3} respectively}

\begin{figure}[t]
	\centering
	\begin{minipage}{0.49\textwidth}
		\centering
		\hspace*{-10pt}
		\includegraphics[scale=0.38]{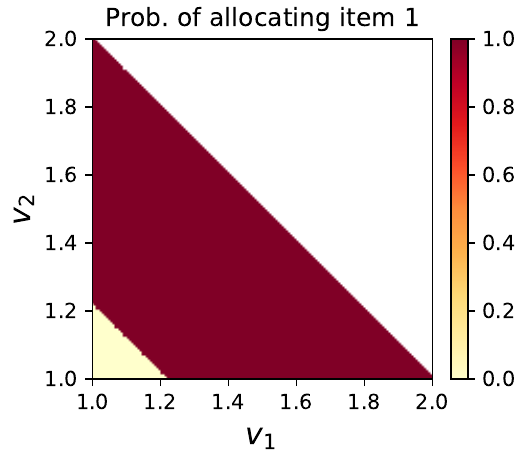}
		\hspace*{-10pt}{\scriptsize (a)}\hspace*{-2pt}
		\includegraphics[scale=0.38]{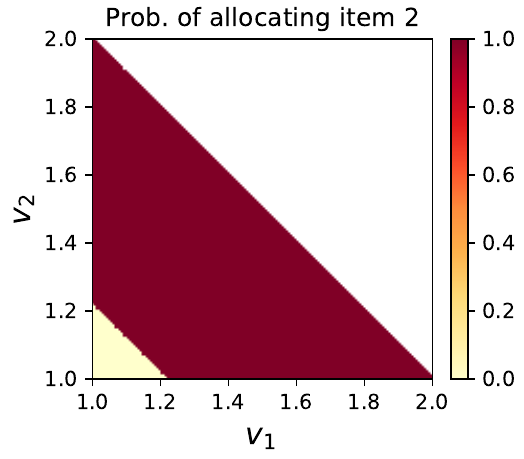}
	\end{minipage}
	\begin{subfigure}{0.49\textwidth}
		\centering
		\hspace*{-10pt}
		\includegraphics[scale=0.38]{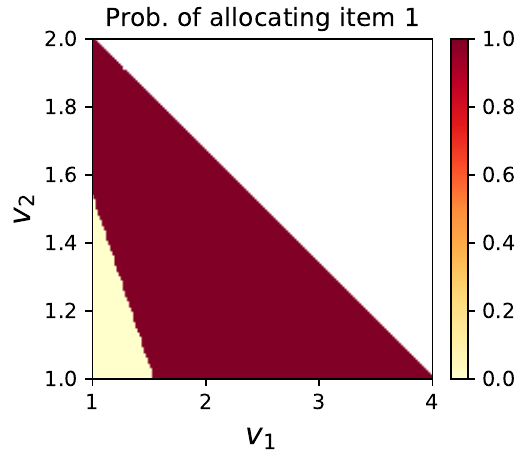}
		\hspace*{-10pt}{\scriptsize (b)}\hspace*{-2pt}
		\includegraphics[scale=0.38]{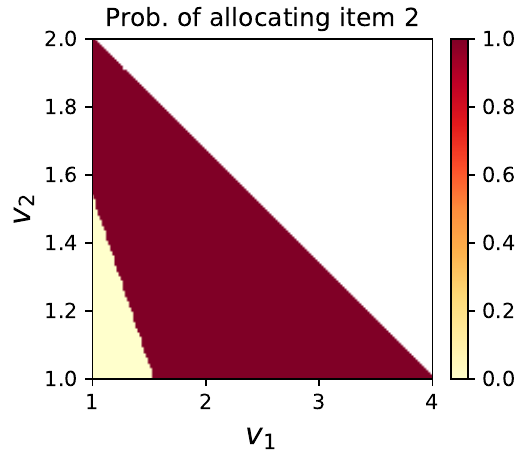}
	\end{subfigure}
	
	\centering
	\begin{minipage}{0.49\textwidth}
		\centering
		\hspace*{-10pt}
		\includegraphics[scale=0.38]{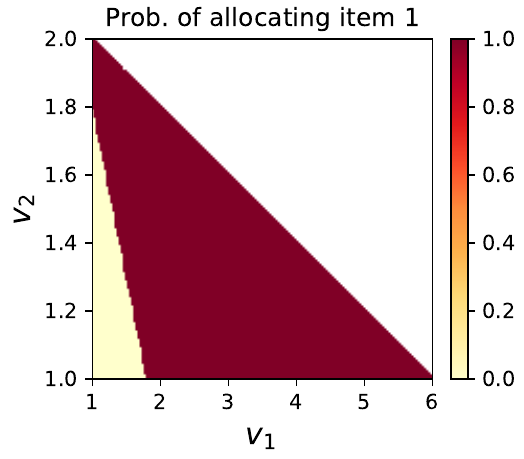}
		\hspace*{-10pt}{\scriptsize (c)}\hspace*{-2pt}
		\includegraphics[scale=0.38]{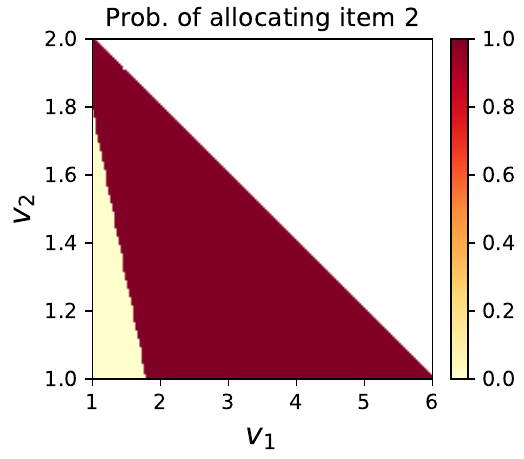}
	\end{minipage}
	\begin{subfigure}{0.49\textwidth}
		\centering
		\hspace*{-10pt}
		\includegraphics[scale=0.38]{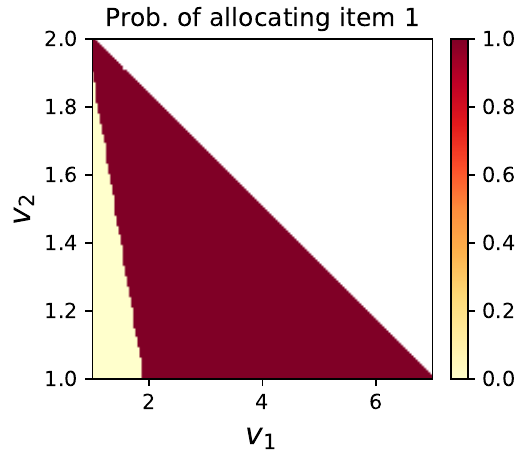}
		\hspace*{-10pt}{\scriptsize (d)}\hspace*{-2pt}
		\includegraphics[scale=0.38]{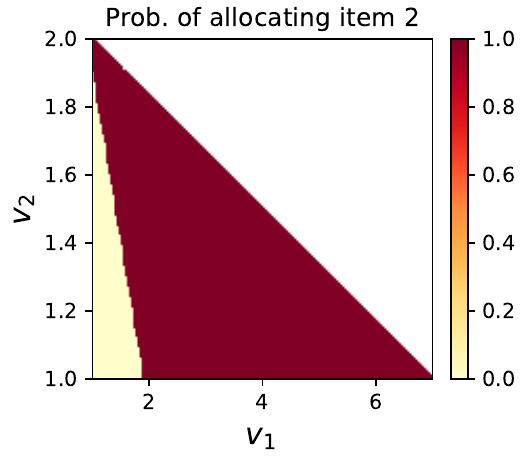}
	\end{subfigure}
	
	\centering
	\begin{minipage}{0.49\textwidth}
		\centering
		\hspace*{-10pt}
		\includegraphics[scale=0.38]{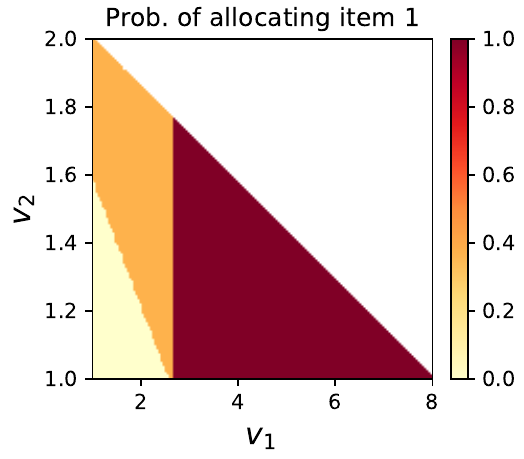}
		\hspace*{-10pt}{\scriptsize (e)}\hspace*{-2pt}
		\includegraphics[scale=0.38]{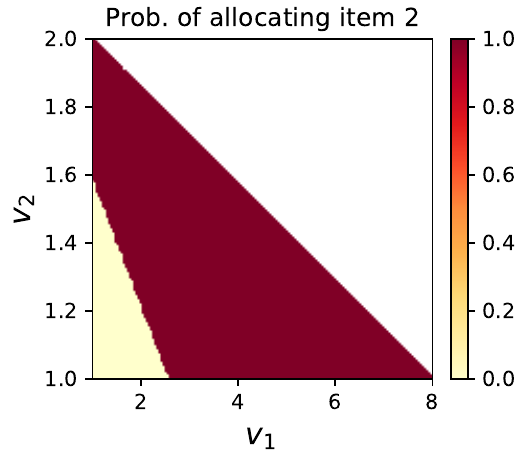}
	\end{minipage}
	\begin{subfigure}{0.49\textwidth}
		\centering
		\hspace*{-10pt}
		\includegraphics[scale=0.38]{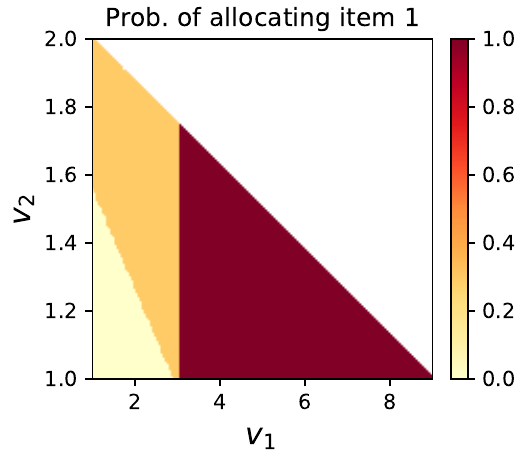}
		\hspace*{-10pt}{\scriptsize (f)}\hspace*{-2pt}
		\includegraphics[scale=0.38]{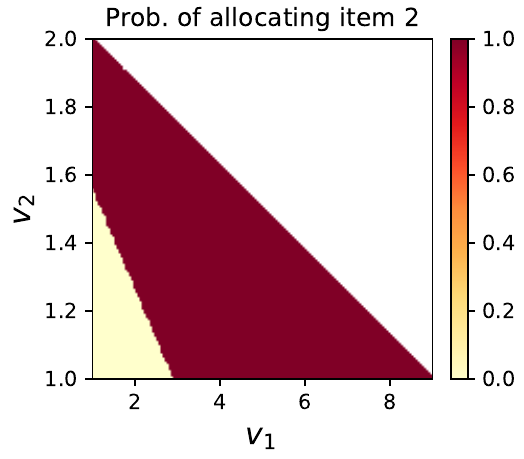}
	\end{subfigure}
	
	\centering
	\begin{minipage}{0.49\textwidth}
		\centering
		\hspace*{-10pt}
		\includegraphics[scale=0.38]{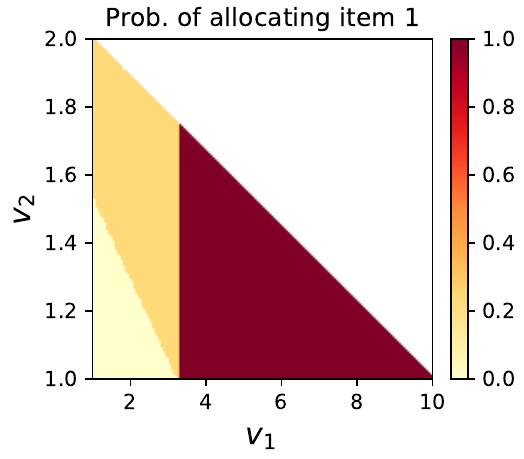}
		\hspace*{-10pt}{\scriptsize (g)}\hspace*{-2pt}
		\includegraphics[scale=0.38]{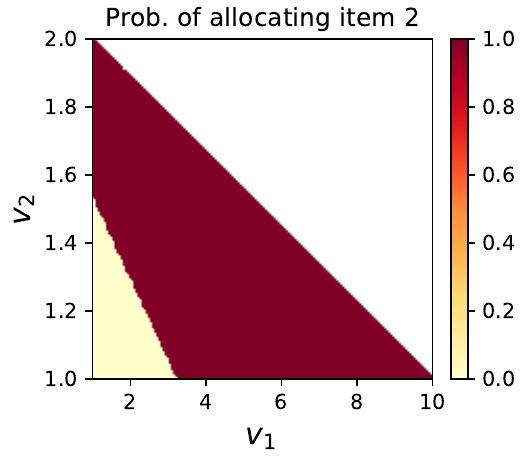}
	\end{minipage}
	\begin{subfigure}{0.49\textwidth}
		\centering
		\hspace*{-10pt}
		\includegraphics[scale=0.38]{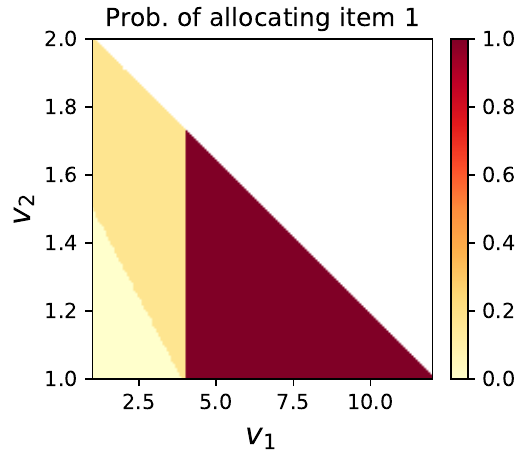}
		\hspace*{-10pt}{\scriptsize (h)}\hspace*{-2pt}
		\includegraphics[scale=0.38]{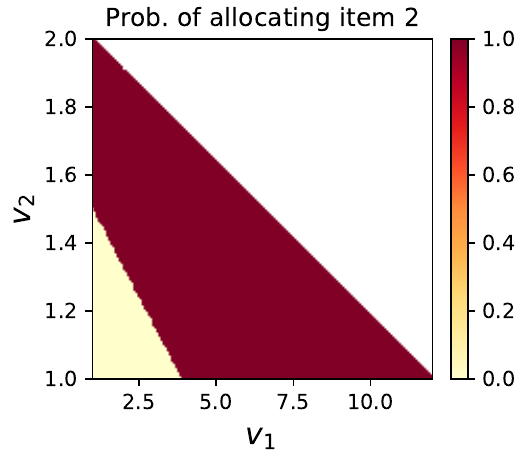}
	\end{subfigure}

	\caption{%
		Allocation rules learned by RochetNet for Setting~\ref{exp:triangle-2}. 
		The panels describe the probability that the bidder is allocated item 1 (left) and item 2 (right) for different values $c = 2.0, 4.0, 6.0, 7.0, 8.0, 9.0, 10.0$ and $12$.
		\label{fig:alloc-triangle-setting-2}}
	\vspace{-10pt}
\end{figure}

\begin{figure}[t]
	\centering
	\begin{minipage}{0.49\textwidth}
		\centering
		\hspace*{-10pt}
		\includegraphics[scale=0.38]{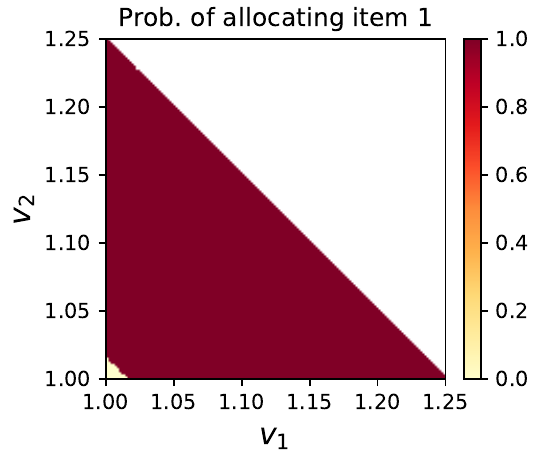}
		\hspace*{-10pt}{\scriptsize (a)}\hspace*{-2pt}
		\includegraphics[scale=0.38]{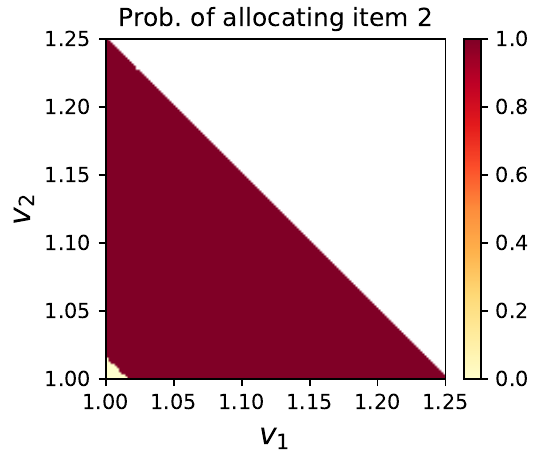}
	\end{minipage}
	\begin{subfigure}{0.49\textwidth}
		\centering
		\hspace*{-10pt}
		\includegraphics[scale=0.38]{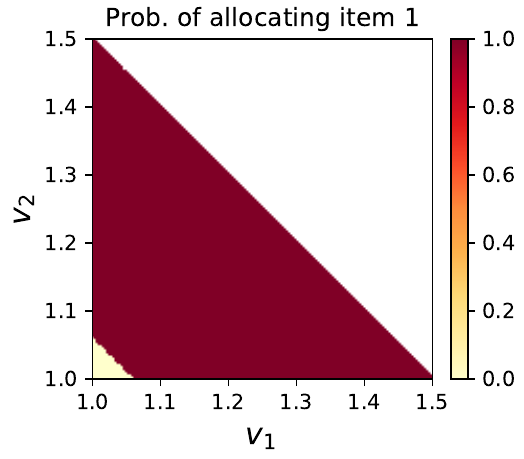}
		\hspace*{-10pt}{\scriptsize (b)}\hspace*{-2pt}
		\includegraphics[scale=0.38]{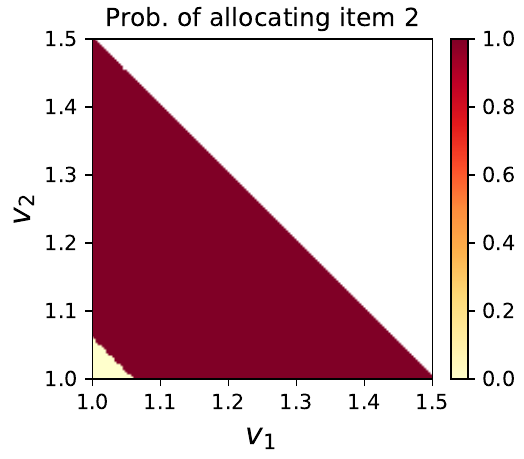}
	\end{subfigure}
	
	\centering
	\begin{minipage}{0.49\textwidth}
		\centering
		\hspace*{-10pt}
		\includegraphics[scale=0.38]{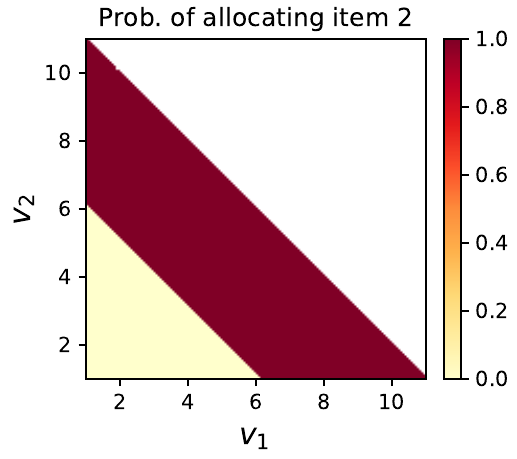}
		\hspace*{-10pt}{\scriptsize (c)}\hspace*{-2pt}
		\includegraphics[scale=0.38]{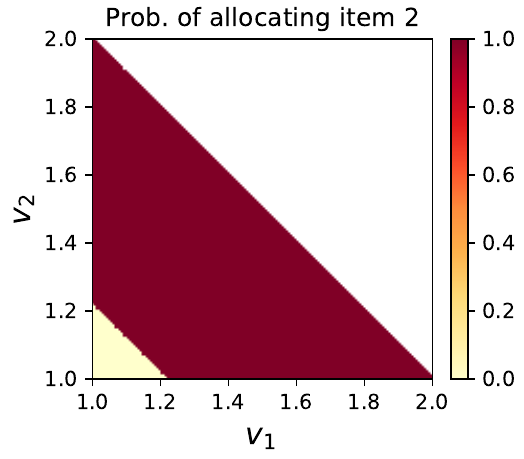}
	\end{minipage}
	\begin{subfigure}{0.49\textwidth}
		\centering
		\hspace*{-10pt}
		\includegraphics[scale=0.38]{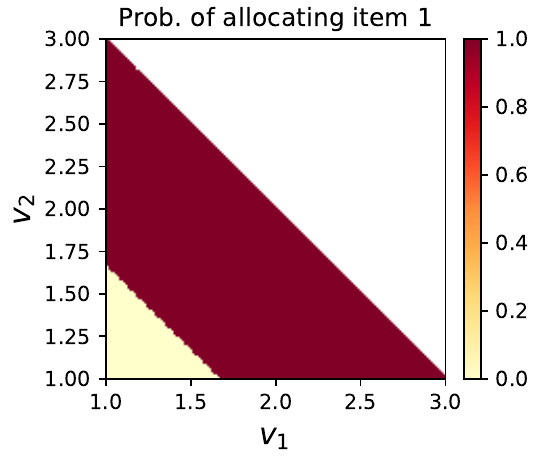}
		\hspace*{-10pt}{\scriptsize (d)}\hspace*{-2pt}
		\includegraphics[scale=0.38]{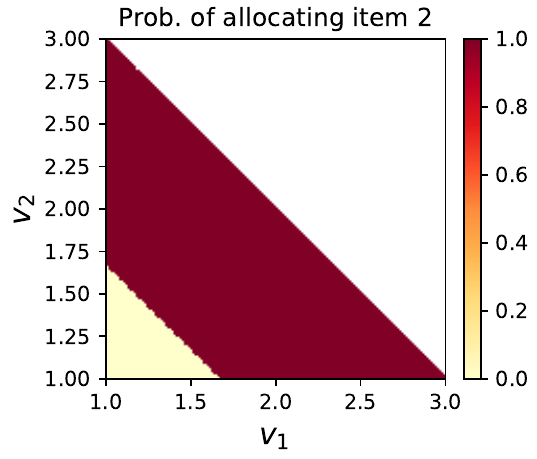}
	\end{subfigure}
	
	\centering
	\begin{minipage}{0.49\textwidth}
		\centering
		\hspace*{-10pt}
		\includegraphics[scale=0.38]{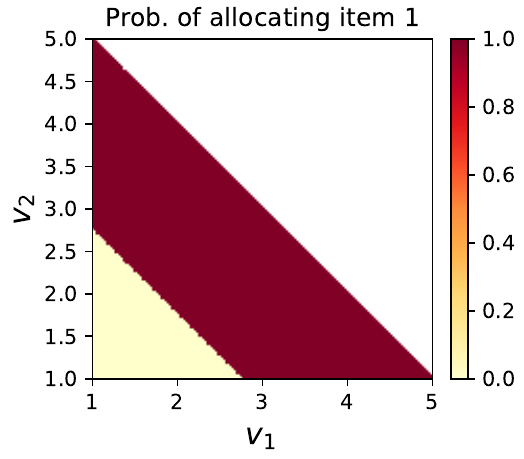}
		\hspace*{-10pt}{\scriptsize (e)}\hspace*{-2pt}
		\includegraphics[scale=0.38]{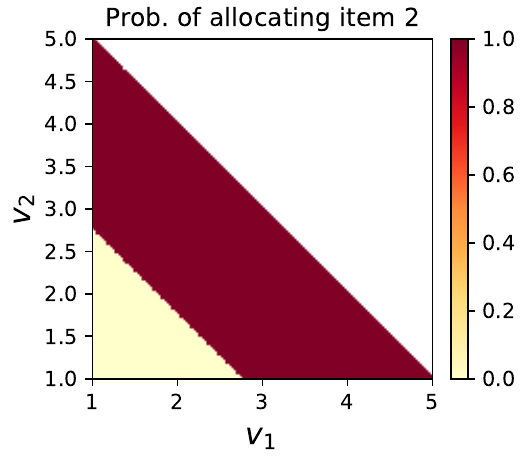}
	\end{minipage}
	\begin{subfigure}{0.49\textwidth}
		\centering
		\hspace*{-10pt}
		\includegraphics[scale=0.38]{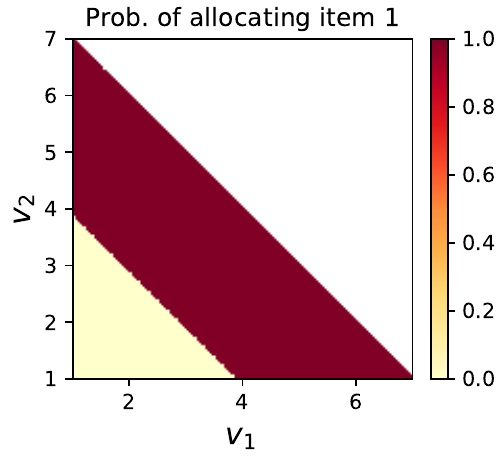}
		\hspace*{-10pt}{\scriptsize (f)}\hspace*{-2pt}
		\includegraphics[scale=0.38]{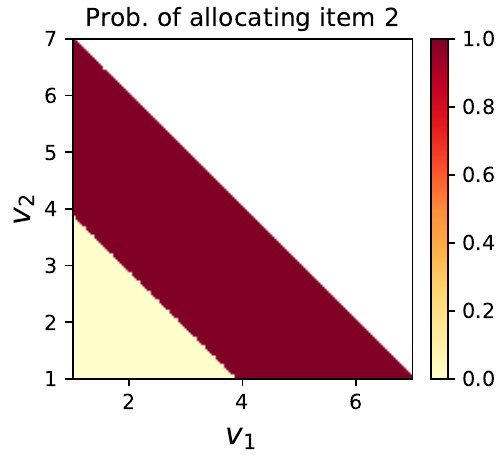}
	\end{subfigure}
	
	\centering
	\begin{minipage}{0.49\textwidth}
		\centering
		\hspace*{-10pt}
		\includegraphics[scale=0.38]{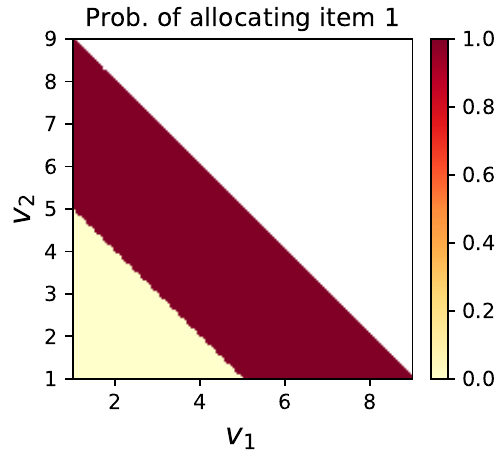}
		\hspace*{-10pt}{\scriptsize (g)}\hspace*{-2pt}
		\includegraphics[scale=0.38]{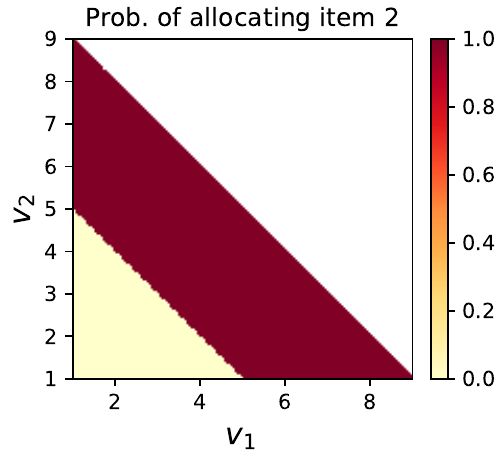}
	\end{minipage}
	\begin{subfigure}{0.49\textwidth}
		\centering
		\hspace*{-10pt}
		\includegraphics[scale=0.38]{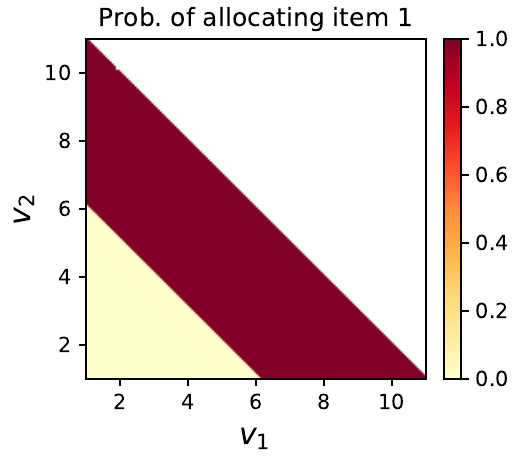}
		\hspace*{-10pt}{\scriptsize (h)}\hspace*{-2pt}
		\includegraphics[scale=0.38]{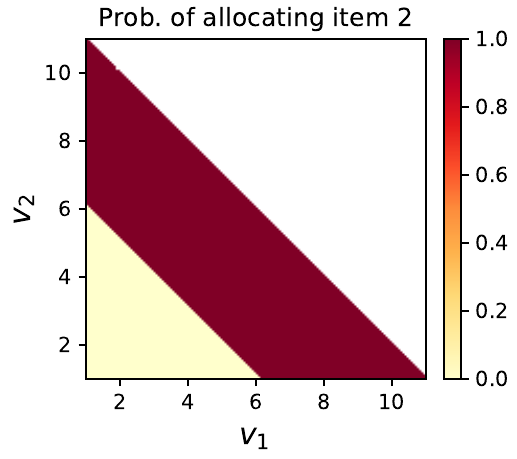}
	\end{subfigure}

	\caption{%
		Allocation rules learned by RochetNet for Setting~\ref{exp:triangle-3}. 
		The panels describe the probability that the bidder is allocated item 1 (left) and item 2 (right) for different values $c = 1.25, 1.5, 2.0, 3.0, 5.0, 7.0, 9.0$ and $11.0$.
		\label{fig:alloc-triangle-setting-3}}
	\vspace{-10pt}
\end{figure}

\section{Comparison to Linear Programming}\label{APP:lp_comparison}

\begin{figure}

\begin{subfigure}{0.99\textwidth}
\centering
\begin{tabular}{|cc|c|c|c|c|c|}
\hline
\multicolumn{2}{|c|}{\multirow{2}{*}{bins / value}} & \multirow{2}{*}{Parameters} & \multirow{2}{*}{$\mathit{rev}$} & \multirow{2}{*}{$\mathit{rgt(mean)}$} & \multirow{2}{*}{\textit{IR viol.}} & Run-time \\
 & &  &  &  &  & (in hours) \\ \hline
\multicolumn{2}{|c|}{\multirow{2}{*}{D = 2}} & \multirow{2}{*}{96} & 1.495 & 0.187 & 0.374 & \multirow{2}{*}{8e-6} \\
 &&  & 0 & 0 & 0 &  \\ \hline
\multicolumn{2}{|c|}{\multirow{2}{*}{D = 3}} & \multirow{2}{*}{486} & 0.994 & 0.046 & 0.090 & \multirow{2}{*}{2e-5} \\
 &&  & 0 & 0.020 & 0 &  \\ \hline
\multicolumn{2}{|c|}{\multirow{2}{*}{D = 4}} & \multirow{2}{*}{1536} & 1.073 & 0.024 & 0.040 & \multirow{2}{*}{9e-5} \\
 &&  & 0.762 & 0.005 & 0 &  \\ \hline
\multicolumn{2}{|c|}{\multirow{2}{*}{D = 5}} & \multirow{2}{*}{3750} & 0.978 & 0.013 & 0.022 & \multirow{2}{*}{3e-4} \\
 &&  & 0.706 & 0.002 & 0 &  \\ \hline
 \multicolumn{2}{|c|}{\multirow{2}{*}{D = 6}} & \multirow{2}{*}{7776} & 0.987 & 0.008 & 0.012 & \multirow{2}{*}{3e-4} \\
 &&  & 0.799 & 0.002 & 0 &  \\ \hline
 \multicolumn{2}{|c|}{\multirow{2}{*}{D = 7}} & \multirow{2}{*}{14406} & 0.967 & 0.006 & 0.009 & \multirow{2}{*}{0.003} \\
 &&  & 0.816 & 0.002 & 0 &  \\ \hline
\multicolumn{2}{|c|}{\multirow{2}{*}{D = 8}} & \multirow{2}{*}{24576} & 0.953 & 0.004 & 0.006 & \multirow{2}{*}{0.007} \\
 &&  & 0.825 & 0.002 & 0 &  \\ \hline
\multicolumn{2}{|c|}{\multirow{2}{*}{D = 9}} & \multirow{2}{*}{39366} & 0.941 & 0.003 & 0.005 & \multirow{2}{*}{ 0.016} \\
 &&  & 0.827 & 0.001 & 0 &  \\ \hline
\multicolumn{2}{|c|}{\multirow{2}{*}{D = 10}} & \multirow{2}{*}{60000} & 0.936 & 0.003 & 0.004 & \multirow{2}{*}{0.45} \\
 &&  & 0.839 & 0.001 & 0 &  \\ \hline
 \multicolumn{2}{|c|}{\multirow{2}{*}{D = 11}} & \multirow{2}{*}{87846} & 0.927 & 0.002 & 0.003 & \multirow{2}{*}{61.9} \\
 &&  & 0.837 & 0.001 & 0 &  \\ \hline
\multicolumn{2}{|c|}{\multirow{2}{*}{D = 12}} & \multirow{2}{*}{124416} & \multirow{2}{*}{$-$} & \multirow{2}{*}{$-$} & \multirow{2}{*}{$-$} & \multirow{2}{*}{\textgreater~216} \\
 &&  &  &  &  &  \\ \hline
\multicolumn{7}{c}{\scriptsize(a)} \\ [-3pt] 
\multicolumn{7}{c}{} \\ \hline
\multicolumn{1}{|c|}{\multirow{2}{*}{$R$}} & \multirow{2}{*}{$K$} & \multirow{2}{*}{Parameters} & \multirow{2}{*}{$\mathit{rev}$} & \multirow{2}{*}{$\mathit{rgt(mean)}$} & \multirow{2}{*}{\textit{IR viol.}} & Training time \\ 
\multicolumn{1}{|c|}{}& & & & & & (in hours) \\
\hline
\multicolumn{1}{|c|}{1} & 100 & 2111 & 0.843 & \textless{}0.0005 & 0 & 3.4 \\ 
\multicolumn{1}{|c|}{1} & 50 & 1061 & 0.844 & \textless{}0.0005 & 0 & 3.5 \\ 
\multicolumn{1}{|c|}{1} & 200 & 4211 & 0.841 & \textless{}0.0005 & 0 & 3.5 \\ 
\multicolumn{1}{|c|}{2} & 100 & 22311 & 0.878 & \textless{}0.0005 & 0 & 3.7 \\ 
\multicolumn{1}{|c|}{2} & 200 & 84611 & 0.874 & \textless{}0.0005 & 0 & 3.7 \\ 
\multicolumn{1}{|c|}{1} & 500 & 10511 & 0.841 & \textless{}0.0005 & 0 & 3.8 \\ 
\multicolumn{1}{|c|}{2} & 50 & 6161 & 0.874 & \textless{}0.0005 & 0 & 3.8 \\ 
\multicolumn{1}{|c|}{2} & 500 & 511511 & 0.87 & \textless{}0.0005 & 0 & 3.8 \\ 
\multicolumn{1}{|c|}{4} & 200 & 245411 & 0.88 & \textless{}0.0005 & 0 & 4.2 \\ 
\multicolumn{1}{|c|}{4} & 50 & 16361 & 0.882 & \textless{}0.0005 & 0 & 4.4 \\ 
\multicolumn{1}{|c|}{4} & 100 & 62711 & 0.884 & \textless{}0.0005 & 0 & 4.4 \\ 
\multicolumn{1}{|c|}{4} & 500 & 1513511 & 0.887 & \textless{}0.0005 & 0 & 5.1 \\ 
\multicolumn{1}{|c|}{8} & 200 & 567011 & 0.896 & \textless{}0.0005 & 0 & 5.2 \\
\multicolumn{1}{|c|}{8} & 50 & 36761 & 0.885 & \textless{}0.0005 & 0 & 5.6 \\
\multicolumn{1}{|c|}{8} & 100 & 143511 & 0.877 & \textless{}0.001 & 0 & 5.6 \\\hline
\multicolumn{7}{c}{\scriptsize(b)} \\ [-5pt]
\end{tabular}
\end{subfigure}
\caption{Test revenue, test regret,  test IR violation, and running-time for the setting of Section~\ref{sec:lpdcp},
with two additive bidders and two items, with bidder item values sampled independently $U[0,1]$.
(a)  LP-based method, with varying levels of  discretization  ($D$), first row for {\tt -nearest} and second for {\tt -down} rounding;and (b)  RegretNet, with varying number of hidden layers ($R$) and hidden units ($K$).
\label{fig:2x3_comp}}
\end{figure}

\ssredit{ In Figure~\ref{fig:2x3_comp}, we report additional details on
the performance of the LP-based approach as we vary the discretization in the LP
and the number of parameters in RegretNet (varying the number of hidden layers and hidden units). For LP, the number of parameters is given by the number of output variables used to define the objective and the constraints. For RegretNet, the number of parameters are computed by counting the number of learnable weights in the allocation and payment network. The results are reported for the  setting  in Section~\ref{sec:lpdcp},
with two additive bidders and two items, with bidder item values sampled independently $U[0,1]$.  For the {\tt -nearest} rounding strategy, the LP-based approach yields a higher revenue than RegretNet, but this is misleading and would not be attainable in practice because it has higher regret and suffers from substantial IR violations. If we instead compute the allocation and payment  in the LP through {\tt -down} rounding, the IR violation is zero but the revenue is much lower. Increasing the amount of discretization in the LP leads to more accurate results with lower regret (and lower IR violations with {\tt -nearest}),  but the number of parameters and the run time also increase exponentially. For the setting with 12 bins per value, the LP did not terminate despite running for 9 days on an AWS EC2 instance with 48 cores and 96GB memory.} 
\ssredit{In contrast, RegretNet learns a mechanism in this setting with negligible regret and zero IR violations in} \dcpadd{at most six hours} \ssredit{for most configurations. In Figure~\ref{fig:2x3_comp}, we report the test revenue and test regret achieved by RegretNet for different hidden layer $R$ and hidden units $K$ configurations.} 

\section{Omitted Proofs}

\new{We present formal proofs for all theorems and lemmas that are stated in the body of the paper or in other appendices.} 
\new{We first introduce some notation.}
We denote the inner product between vectors $a,b \in \R^d$ as $\langle a, b \rangle \,=\, \sum_{i=1}^d a_i b_i$. \new{We denote the $\ell_1$ norm for a vector $x$ by $\Vert x \Vert_1$ and  the induced $\ell_{1}$ norm
for a matrix $A \in \R^{k\times t}$ 
 by $\|A\|_{1} \,=\, \max_{1\leq j \leq t} \sum_{i=1}^k A_{ij}$.}

\subsection{Proof of Lemma~\ref{lem:quantile-regret}}\label{app:quantile-regret}

Let $f_i(v;w) := \max_{v'_i \in V_i}\,u^w_i(v_i; (v'_i, v_{-i})) - u^w_i(v_i;(v_i, v_{-i}))$. Then we have $\mathit{rgt}_i(w) = \E_{v\sim F}[f_i(v;w)]$. Rewriting the expected value, we have
\begin{align*}
\mathit{rgt}_i(w) &= \int_{0}^\infty \PP(f_i(v;w)\geq x) dx \geq \int_{0}^{\mathit{rgt}^q_i(w)} \PP(f_i(v;w)\geq x) dx \geq q\cdot\mathit{rgt}^q_i(w),
\end{align*}
where the last inequality holds because for any $0<x<\mathit{rgt}^q_i(w)$, $\PP(f_i(v;w) \geq x) \geq \PP(f_i(v;w) \geq \mathit{rgt}^q_i(w)) = q$. \qed

\subsection{Proof of Theorem~\ref{THM:GBOUND}}\label{app:gbound}

We present the proof for auctions with general, randomized allocation rules. 
A randomized allocation rule  $g_i: V \rightarrow [0,1]^{2^M}$ maps valuation profiles to a vector of allocation probabilities for bidder $i$, where $g_{i,S}(v) \in [0,1]$ denotes the probability that the allocation rule assigns subset of items $S \subseteq M$ to bidder $i$, and 
$\sum_{S \subseteq M}g_{i,S}(v) \leq 1$. This encompasses both the allocation rules for the combinatorial setting, and the  allocation rules for the additive and unit-demand settings, which only output allocation probabilities for individual items.  The payment function  $p: V \rightarrow R^n$ maps valuation profiles to 
a payment for each bidder $p_i(v) \in \R$. 
For ease of exposition, we omit the superscripts ``$w$''. Recall that $\M$ is a class of auctions consisting of allocation and payment rules $(g,p)$. As noted in the theorem statement, we will assume w.l.o.g.~that  for each bidder $i$, $v_i(S) \leq 1,\, \forall S \subseteq M$.

\subsubsection{\new{Definitions}}
Let  $\mathcal{U}_i$ be the class of utility functions for bidder $i$ defined on auctions in $\M$, i.e.,
\begin{align*}
\mathcal{U}_i \,=\, \big\{
&u_i: V_i \times V \rightarrow \R\,\big|\,u_i(v_i, b) \,=\, v_i(g(b)) \,-\, p_i(b) \text{ for some }(g,p) \in \M
\big\}.
\end{align*}
and let $\mathcal U$ be the class of profile of utility functions defined on $\M$, i.e., the class of tuples $(u_1,\ldots,u_n)$
where each $u_i: V_i\times V \rightarrow \R$
and 
$u_i(v_i, b) \,=\, v_i(g(b)) \,-\, p_i(b), 
\forall i \in N$
for some $(g,p) \in \M$. 

We will sometimes find it useful to represent the utility function as an inner product, i.e., treating $v_i$ as a real-valued vector of length $2^M$, we may write ${u}_i(v_i, b) =  \langle v_i, {g}_{i}(b)\rangle - {p}_i(b)$.

Let $\mathrm{rgt}\circ \mathcal{U}_i$ be the class of all regret functions for bidder $i$ 
defined  on utility functions in $\mathcal{U}_i$, i.e.,
\begin{align*}
\mathrm{rgt}\,\circ\, \mathcal{U}_i =
\Big\{
 &f_i: V \rightarrow \R \,\Big|\, f_i(v) \,=\, 
 \max_{v'_i} u_i(v_i, (v'_i, v_{-i})) \,-\, u_i(v_i, v)\text{ for some  }
u_i \in \mathcal{U}_i
\Big\},
\end{align*}
and as before, let $\mathrm{rgt}\,\circ\, \mathcal{U}$ be defined as 
the class of profiles of regret functions.

Define the $\ell_{\infty, 1}$ distance between two utility functions $u$ and $u'$ as 
\begin{align*}
\max_{v, v'} \sum_i \vert u_i(v_i,(v'_i,v_{-i})) - u_i(v_i,(v'_i,v_{-i}))\vert
\end{align*}
and \new{let} $\mathcal{N}_\infty(\mathcal{U}, \epsilon)$ \new{denote} the minimum number of balls of radius $\epsilon$ to cover $\mathcal{U}$ under this distance. Similarly, define the distance between $u_i$ and $u'_i$ as $\max_{v, v'_i}\vert u_i(v_i,(v'_i, v_{-i})) - u'_i(v_i,(v'_i, v_{-i}))\vert$,  and let $\mathcal{N}_\infty(\mathcal{U}_i, \epsilon)$ denote the minimum number of balls of radius $\epsilon$ to cover $\mathcal{U}_i$ under this distance. Similarly, we define covering numbers $\mathcal{N}_\infty(\mathrm{rgt}\,\circ\,\mathcal{U}_i, \epsilon)$ and $\mathcal{N}_\infty(\mathrm{rgt}\,\circ\,\mathcal{U}, \epsilon)$ for the function classes $\mathrm{rgt}\,\circ\,\mathcal{U}_i$ and $\mathrm{rgt}\,\circ\,\mathcal{U}$ respectively.

Moreover, we denote the class of allocation functions as $\mathcal{G}$ and for
each bidder $i$,
$\mathcal{G}_i \,=\, \{g_i: V\rightarrow 2^M\,|\,g \in  \mathcal{G}\}$. 
Similarly, we denote the class of payment functions by $\mathcal{P}$ and 
$\mathcal{P}_i
\,=\,
 \{p_i: V\rightarrow \R\,|\,p \in  \mathcal{P}\}$. We denote the covering number of $\mathcal{P}$ as $\mathcal{N}_\infty(\mathcal{P}, \epsilon)$  under the $\ell_{\infty,1}$ distance and the covering number for $\mathcal{P}_i$ using $\mathcal{N}_\infty(\mathcal{P}_i, \epsilon)$ under the \new{$\ell_{\infty, 1}$} distance.
 
 \subsubsection{\new{Auxiliary Lemma}}

 \new{We will use  a lemma from~\citet{Shalev-Shwartz14}. 
 Let  $\mathcal{F}$ denote a class of bounded functions $f: {Z} \rightarrow [-c,c]$ defined on an input space $Z$, for some $c>0$. Let $D$ be a distribution over $Z$ and $\mathcal{S} \,=\, \{z_1,\ldots,z_L\}$ be a sample drawn i.i.d.\ from $D$. We are interested in the gap between the  expected value of a  function $f$  and the average value of the function on sample $S$, and would like to bound this gap uniformly for all functions in $\mathcal{F}$. For this, we measure  the capacity of the function class $\mathcal{F}$
using the empirical Rademacher complexity on  sample $S$, defined below:}
\[
 \hat{\mathcal{R}}_{L}(\mathcal{F}):= \frac{1}{L}\E_{\sigma}\left[\sup_{f\in\mathcal{F}}\sum_{z_i\in S}\sigma_i f(z_i)\right],
 \]
 where $\sigma \in \{-1,1\}^L$ and each $\sigma_i$
 is drawn i.i.d  from a uniform distribution on $\{-1,1\}$. \new{We then have:}
\begin{lem}
[\cite{Shalev-Shwartz14}]
\label{lem:rademacher+gbound}
 Let $\mathcal{S} \,=\, \{z_1,\ldots,z_L\}$ be a sample drawn i.i.d. from some distribution $D$ over $Z$. Then with probability of at least $1-\delta$ over draw of $\mathcal{S}$ from $D$, for all $f\in\mathcal{F}$, 
\begin{align*}
\E_{z\sim D}[f(z)] \leq \frac{1}{L}\sum_{\ell=1}^L f(z_\ell) \,+\, 2\hat{\mathcal{R}}_{L}(\mathcal{F}) \,+\, 4c\sqrt{\frac{2\log(4/\delta)}{L}},
\end{align*}
\end{lem}

\subsubsection{\new{Generalization Bound for Revenue}}
\label{sec:revenue-gbound}
\allowdisplaybreaks
\new{We first prove the generalization bound for revenue.  } %
\new{
For this, we define the following {auxiliary} function class, where each
$f: V \rightarrow \R_{\geq 0}$measures the total payments from some mechanism in $\M$:
$$
\mathrm{rev}\circ \M \,=\, \big\{f: V \rightarrow \R_{\geq 0}\,\big|\,
f(v) \,=\, \textstyle\sum_{i=1}^n p_i(v), \text{ for some }(g,p) \in \M
\big\}.
$$
Note each function $f$ in this class corresponds to a mechanism $(g,p)$  in $\M$, and the expected value $\E_{v \sim D}[f(v)]$ gives the expected revenue from that mechanism. The proof then follows by an application of the uniform convergence bound in Lemma \ref{lem:rademacher+gbound} to the above function class, and by further bounding the Rademacher complexity term in this bound by the covering number of the auction class $\M$.}

\new{
Applying Lemma \ref{lem:rademacher+gbound} to  the auxiliary function class $\mathrm{rev}\circ \M$}, we get with probability of at least $1-\delta$ over draw of $L$ valuation profiles $S$ from $D$, for any $f \in \mathrm{rev}\circ \M$, \ZFadd{there exists a distribution-independent constant $C > 0$ such that,}
\begin{align}
\E_{v\sim F}\Big[-\sum_{i\in N}p_i(v)\Big] 
\leq  &~-\frac{1}{L}\sum_{\ell=1}^L \sum_{i=1}^n p_i(v^{(\ell)})\allowdisplaybreaks[0]
\nonumber
\\
&~+ 2\hat{R}_L(\mathrm{rev}\circ \M)
~+ Cn\sqrt{\frac{\log(1/\delta)}{L}}.
\label{eq:gen-bound-revenue-class}
\end{align}

\new{All that remains is to bound the above empirical Rademacher \new{complexity}  $\hat{R}_L(\mathrm{rev}\circ \M)$ in terms of the covering number of the payment class $\mathcal{P}$ and in turn in terms of the covering number of the auction class $\M$.} %
Since we assume that the auctions in $\M$ satisfy individual rationality
and $v(S) \leq 1, \forall S  \subseteq M$, we have for any $v$, $p_i(v)\leq 1$.

\new{By the definition of the covering number for the payment class, there exists a cover $\hat{\mathcal{P}}$ for $\mathcal{P}$ of size $|\hat{\mathcal{P}}|\leq \mathcal{N}_\infty(\mathcal{P}, \epsilon)$} such that for any $p\in\mathcal{P}$, there is a ${f_p}\in\hat{\mathcal{P}}$ with $\max_v\sum_i\vert p_i(v)-{f_p}_i(v)\vert\leq \epsilon$. We thus have:  
\begin{align}
\hat{\mathcal{R}}_L(\mathrm{rev}\circ \M	)
=&~\frac{1}{L}\E_\sigma\left[\sup_p\sum_{\ell=1}^{L}\sigma_\ell\cdot\sum_i p_i(v^{(\ell)})\right]\allowdisplaybreaks[0]
\nonumber
\\
= &~\frac{1}{L}\E_\sigma\left[\sup_p\sum_{\ell=1}^{L}\sigma_\ell\cdot\sum_i {f_p}_i(v^{(\ell)})\right] + 
\frac{1}{L}\E_\sigma\left[\sup_p\sum_{\ell=1}^{L}\sigma_\ell\cdot\sum_i p_i(v^{(\ell)}) - {f_p}_i(v^{(\ell)})\right]\allowdisplaybreaks[0]
\nonumber\\
\leq &~\frac{1}{L}\E_\sigma\left[\sup_{\hat{p}\in\hat{\mathcal{P}}}\sum_{\ell=1}^{L}\sigma_\ell\cdot\sum_i \hat{p}_i(v^{(\ell)})\right] + \frac{1}{L}\E_\sigma \Vert\sigma\Vert_1 \epsilon\allowdisplaybreaks[0]
\nonumber\\
\leq &~\sqrt{\sum_{\ell} (\sum_i \hat{p}_i(v^{\ell}))^2} \sqrt{\frac{2\log(\mathcal{N}_\infty(\mathcal{P}, \epsilon))}{L}} + \epsilon
\nonumber
\\ 
\leq &~2n\sqrt{\frac{2\log(\mathcal{N}_\infty(\mathcal{P}, \epsilon))}{L}} + \epsilon,
\label{eq:Radamacher-covering-bound}
\end{align}
\new{where the second-last inequality follows from Massart's lemma, and}
the last inequality \new{holds} because
$$\sqrt{\sum_{\ell} \left(\sum_i \hat{p}_i(v^{\ell})\right)^2}\leq \sqrt{\sum_{\ell} \left(\sum_i p_i(v^{\ell}) + n\epsilon\right)^2} \leq 2n\sqrt{L}.$$

\new{We further observe that $\mathcal{N}_\infty(\mathcal{P}, \epsilon) \leq \mathcal{N}_\infty(\mathcal{M}, \epsilon)$. By the definition of the covering number for the auction class $\M$, there exists a cover $\hat{\M}$ for $\mathcal{M}$ of size $|\hat{\M}| \leq  \mathcal{N}_\infty(\mathcal{M}, \epsilon)$}
such that for
 any $(g, p)\in \mathcal{M}$, there is a $(\hat{g}, \hat{p}) \in \hat{\M}$ \new{such that} for all $v$, %
\begin{align*}
\sum_{i,j} \vert g_{ij}(v) - \hat{g}_{ij}(v)\vert + \sum_i\vert p_i(v) - \hat{p}_i(v)\vert \leq \epsilon.
\end{align*} 
\new{This also implies that  $\sum_i\vert p_i(v) - \hat{p}_i(v)\vert\leq \epsilon$, and shows the existence of a cover for $\mathcal{P}$ of size at most  $\mathcal{N}_\infty(\mathcal{M}, \epsilon)$. }

\new{
Substituting the bound on the Radamacher complexity term in \eqref{eq:Radamacher-covering-bound} in \eqref{eq:gen-bound-revenue-class} and using the fact that $\mathcal{N}_\infty(\mathcal{P}, \epsilon) \leq \mathcal{N}_\infty(\mathcal{M}, \epsilon)$, we get:}
\begin{align*}
\E_{v\sim F}\Big[-\sum_{i\in N}p_i(v)\Big] 
\leq  &~-\frac{1}{L}\sum_{\ell=1}^L \sum_{i=1}^n p_i(v^{(\ell)})\allowdisplaybreaks[0]+ 2\cdot\inf_{\epsilon>0}\Big\{\epsilon +\,2n\sqrt{\frac{2\log(\cN_\infty(\M, \,\epsilon))}{L}}\Big\} \allowdisplaybreaks[0] + Cn\sqrt{\frac{\log(1/\delta)}{L}}\,,
\end{align*}
\new{which completes the proof.}

\subsubsection{\new{Generalization Bound for Regret}}\label{sec:regret-gbound}

We move to the second part, namely a generalization bound for regret, which
is the more challenging  part of the proof. 
We first define  the class of sum regret functions:
\[
\overline{\mathrm{rgt}}\,\circ\,\mathcal{U} = 
\left\{f: V \rightarrow \R
\,\bigg|\,
f(v) \,=\, \sum_{i=1}^n r_i(v)
\text{ for some }
(r_1, \ldots, r_n) \in 
\mathrm{rgt} \circ \mathcal{U}
\right\}.
\]
The proof then proceeds in three steps:

(1) bounding the covering number for each regret class $\mathrm{rgt}\circ\mathcal{U}_i$ in terms of the covering number for individual utility classes $\mathcal{U}_i$

(2) bounding the covering number for the combined utility class $\mathcal{U}$ in terms of the covering number for $\M$

(3) bounding the covering number for the sum regret class $\overline{\mathrm{rgt}}\,\circ\,\mathcal{U}$ in terms of the covering number for the (combined) utility class $\M$.

An application of Lemma \ref{lem:rademacher+gbound} then completes the proof. We prove each of the above steps below.

\begin{step}\label{step:one}
$\mathcal{N}_\infty(\mathrm{rgt}\circ\mathcal{U}_i, \epsilon) \leq \mathcal{N}_\infty(\mathcal{U}_i, \epsilon/2)$.
\end{step}

\begin{proof}
By the definition of covering number $\mathcal{N}_\infty(\mathcal{U}_i, \epsilon)$, 
there exists  a cover $\hat{\U}_i$ with size at most $\mathcal{N}_\infty(\mathcal{U}_i, \epsilon/2)$ such that
for any $u_i\in \mathcal{U}_i$, there is a $\hat{u}_i \in \hat{\mathcal{U}}_i$ with
\begin{align*}
\sup_{v, v'_i} \vert u_i(v_i, (v'_i, v_{-i})) - \hat{u}_i(v_i, (v'_i, v_{-i}))\vert\leq \epsilon/2.
\end{align*}

For any $u_i \in \U_i$, taking $\hat{u}_i\in\hat{\mathcal{U}}_i$ satisfying the above condition, then for any $v$,
\begin{align*}
&\bigg\vert\max_{v'_i\in V}\big(u_i(v_i, (v'_i, v_{-i})) - u_i(v_i, (v_i, v_{-i}))\big) - \max_{\bar{v}_i\in V}\big(\hat{u}_i(v_i, (\bar{v}_i, v_{-i})) - \hat{u}_i(v_i, (v_i, v_{-i}))\big)\bigg\vert\\
\leq~& \bigg\vert\max_{v'_i}u_i(v_i, (v'_i, v_{-i})) - \max_{\bar{v}_i}\hat{u}_i(v_i, (\bar{v}_i, v_{-i})) + \hat{u}_i(v_i, (v_i, v_{-i})) - u_i(v_i, (v_i, v_{-i}))\bigg\vert\\
\leq~&\left\vert\max_{v'_i}u_i(v_i, (v'_i, v_{-i})) - \max_{\bar{v}_i}\hat{u}_i(v_i, (\bar{v}_i, v_{-i}))\right\vert + \bigg\vert \hat{u}_i(v_i, (v_i, v_{-i})) - u_i(v_i, (v_i, v_{-i}))\bigg\vert\\
\leq~&\left\vert\max_{v'_i}u_i(v_i, (v'_i, v_{-i})) - \max_{\bar{v}_i}\hat{u}_i(v_i, (\bar{v}_i, v_{-i}))\right\vert + \epsilon/2
\end{align*}

Let $v^*_i \in \arg\max_{v'_i} u_i(v_i, (v'_i, v_{-i}))$ and $\hat{v}^*_i \in \arg\max_{\bar{v}_i} \hat{u}_i(v_i, (\bar{v}_i, v_{-i}))$, then
\begin{equation*}%
\begin{aligned}
\max_{v'_i}u_i(v_i, (v'_i, v_{-i})) &
= u_i(v^*_i, v_{-i}) \leq \hat{u}_i(v^*_i, v_{-i}) + \epsilon/2 \leq \hat{u}_i(\hat{v}^*_i, v_{-i}) + \epsilon/2 = \max_{\bar{v}_i}\hat{u}_i(v_i, (\bar{v}_i, v_{-i})) + \epsilon,\\
\max_{\bar{v}_i}\hat{u}_i(v_i, (\bar{v}_i, v_{-i})) 
&= \hat{u}_i(\hat{v}^*_i, v_{-i}) \leq u_i(\hat{v}^*_i, v_{-i}) + \epsilon/2 \leq  u_i(v^*_i, v_{-i}) + \epsilon/2 = \max_{v'_i}u_i(v_i, (v'_i, v_{-i}))  + \epsilon/2\,.
\end{aligned}
\end{equation*}

Thus, for all $u_i \in \mathcal{U}_i$, there exists $\hat{u}_i \in \hat{\mathcal{U}}_i$ such that for any valuation profile $v$,
\begin{align*}
&\bigg\vert\max_{v'_i}\big(u_i(v_i, (v'_i, v_{-i})) - u_i(v_i, (v_i, v_{-i}))\big) - \max_{\bar{v}_i}\big(\hat{u}_i(v_i, (\bar{v}_i, v_{-i})) - \hat{u}_i(v_i, (v_i, v_{-i}))\big)\bigg\vert \leq \epsilon,
\end{align*}
which implies $\mathcal{N}_\infty(\mathrm{rgt}\circ \mathcal{U}_i, \epsilon) \leq \mathcal{N}_\infty(\mathcal{U}_i, \epsilon/2)$.

This completes the proof \new{of} Step~\ref{step:one}.
\end{proof}

\begin{step}\label{step:two}
\new{For all  $i\in N$,} $\mathcal{N}_\infty(\mathcal{U}, \epsilon)\leq \mathcal{N}_\infty(\mathcal{M}, \epsilon/n)$.
\end{step}

\begin{proof}
Recall that the utility function of bidder $i$ is $u_i(v_i, (v'_i, v_{-i})) = \langle v_i, g_{i}(v'_i, v_{-i})\rangle - p_i(v'_i, v_{-i})$.
There exists a set $\hat{\mathcal{M}}$ with $|\hat{\mathcal{M}}|\leq\mathcal{N}_\infty(\mathcal{M}, \epsilon/n)$ 
such that, %
there exists 
 $(\hat{g}, \hat{p}) \in \hat{M}$ such that
\begin{align*}
\sup_{v\in  V}\sum_{i,j} \vert g_{ij}(v) - \hat{g}_{ij}(v)\vert + \Vert p(v) - \hat{p}(v)\Vert_1 \leq \epsilon/n.
\end{align*}
We denote $\hat{u}_i(v_i, (v'_i, v_{-i})) =  \langle v_i, \hat{g}_{i}(v'_i, v_{-i})\rangle - \hat{p}_i(v'_i, v_{-i})$, where we treat $v_i$ as a real-valued vector  of length $2^M$. %

For all $v \in V, v'_i \in V_i$,
\begin{eqnarray*}
\lefteqn{\left\vert u_i(v_i, (v'_i, v_{-i})) - \hat{u}_i(v_i, (v'_i, v_{-i}))\right\vert}\\
& \leq& \left\vert \langle v_i, g_{i}(v'_i, v_{-i})\rangle -  \langle v_i, \hat{g}_{i}(v'_i, v_{-i})\rangle \right\vert + \left\vert p_i(v'_i, v_{-i}) -\hat{p}_i(v'_i, v_{-i}) \right\vert\\
& \leq& \Vert v_i\Vert_\infty \cdot \Vert g_{i}(v'_i, v_{-i}) -  \hat{g}_{i}(v'_i, v_{-i})\Vert_1 + \left\vert p_i(v'_i, v_{-i}) -\hat{p}_i(v'_i, v_{-i}) \right\vert\\
& \leq& \sum_j \vert g_{ij}(v'_i, v_{-i}) - \hat{g}_{ij}(v'_i, v_{-i})\vert + \left\vert p_i(v'_i, v_{-i}) -\hat{p}_i(v'_i, v_{-i}) \right\vert\\
&\leq& \ZFadd{\epsilon / n}
\end{eqnarray*}

Therefore, for any $u\in \mathcal{U}$, take $\hat{u} =(\hat{g}, \hat{p})\in \hat{\mathcal{M}}$, for all $v, v'$,
\begin{align*}
&\sum_i \vert u_i(v_i, (v'_i, v_{-i})) - \hat{u}_i(v_i, (v'_i, v_{-i}))\vert \\
&~\leq  \sum_{ij} \vert g_{ij}(v'_i, v_{-i}) - \hat{g}_{ij}(v'_i, v_{-i})\vert + \sum_i\left\vert p_i(v'_i, v_{-i}) -\hat{p}_i(v'_i, v_{-i}) \right\vert\\
&~\leq \epsilon
\end{align*}


This completes the proof \new{of} Step~\ref{step:two}.
\end{proof}

\begin{step}\label{step:three}
$\mathcal{N}_\infty(\overline{\mathrm{rgt}}\circ \mathcal{U}, \epsilon)\leq \mathcal{N}_\infty(\M, \frac{\epsilon}{2n})$.
\end{step}

\begin{proof}
By definition of $\mathcal{N}_\infty(\mathcal{U}, \epsilon)$, 
there exists $\hat{\U}$ with size at most $\mathcal{N}_\infty(\mathcal{U}, \epsilon)$,
such that,  for any $u\in \mathcal{U}$, there exists $\hat{u}$ \new{such that} for all $v, v' \in V$,
\begin{align*}
\sum_{i} \vert u_i(v_i, (v'_i, v_{-i})) - \hat{u}_i(v_i, (v'_i, v_{-i}))\vert \leq \epsilon.
\end{align*}
Therefore for all $v \in V$, $\vert\sum_i u_i(v_i, (v'_i, v_{-i})) - \sum_i \hat{u}_i(v_i, (v'_i, v_{-i}))\vert \leq \epsilon$, from which it follows that $\mathcal{N}_\infty(\overline{\mathrm{rgt}}\circ \mathcal{U}, \epsilon)\leq \mathcal{N}_\infty(\mathrm{rgt}\circ \mathcal{U}, \epsilon)$. Following Step 1, it is easy to show $\mathcal{N}_\infty(\mathrm{rgt}\circ\mathcal{U}, \epsilon)\leq\mathcal{N}_\infty(\mathcal{U}, \epsilon/2)$. 

\new{Together} with Step~\ref{step:two} \new{this} completes the proof of Step~\ref{step:three}.
\end{proof}

Based on the same arguments \new{as} in Section~\ref{sec:revenue-gbound}, \new{we can thus bound} the empirical Rademacher \new{complexity} \new{as:}
\begin{align*}
\hat{\mathcal{R}}_L(\overline{\mathrm{rgt}}\,\circ\,\mathcal{U})
&\leq 
\inf_{\epsilon>0} \left(\epsilon + 
2n\sqrt{\frac{2\log\mathcal{N}_\infty(\overline{\mathrm{rgt}}\circ\mathcal{U}, \epsilon)}{L}}\right)\\
&\leq \inf_{\epsilon>0} \left(\epsilon + 2n\sqrt{\frac{2\log\mathcal{N}_\infty(\mathcal{M}, \frac{\epsilon}{2n})}{L}}\right).
\end{align*}

Applying Lemma~\ref{lem:rademacher+gbound}, completes the proof \new{of the} generalization bound \new{for} regret. \qed
%

%
\subsection{Proof of Lemma~\ref{LEM:UNIT_DEMAND_DS}}
\label{APP:UNIT_DEMAND_DS}
First, given the property of the softmax function and the min operation, $\varphi^{DS}(s, s')$ ensures that the row sums and column sums for the resulting allocation matrix do not exceed 1.  In fact, for any doubly stochastic allocation $z$, there exists scores $s$ and $s'$, for which the min of normalized scores recovers $z$ (e.g.  $s_{ij} = s'_{ij} = log(z_{ij}) + c$ for any $c \in \R$). \qed
\subsection{Proof of Lemma~\ref{LEM:CA_DS}}\label{APP:CA_DS}
Similar to Lemma ~\ref{LEM:UNIT_DEMAND_DS}, $\varphi^{CF}(s, s^{(1)},\ldots,s^{(m)})$ trivially satisfies the combinatorial feasibility (constraints \eqref{eq:combinatorial-constraints-1}--\eqref{eq:combinatorial-constraints-2}). For any allocation $z$ that satisfies the combinatorial feasibility, the following scores
\begin{align*}
\forall j~=~ 1,\cdots, m, ~~~~ s_{i, S} & ~=~ s^{(j)}_{i, S} ~=~  \log(z_{i, S}) + c,
\end{align*}
makes $\varphi^{CF}(s, s^{(1)},\ldots,s^{(m)})$ recover $z$. \qed

\subsection{Proof of Theorem \ref{thm:cover_regretnet}}\label{APP:THM_COVER_REGRETNET}

\new{In Theorem~\ref{thm:cover_regretnet}, we only show the bounds on $\Delta_L$ %
for RegretNet with additive and unit-demand bidders. %
We restate this theorem so that it also bounds $\Delta_L$ for the general combinatorial valuations setting (with combinatorial feasible allocation).}
\new{Recall that the $\ell_1$ norm for a vector $x$ is denoted by $\Vert x \Vert_1$ and  the induced $\ell_{1}$ norm
for a matrix $A \in \R^{k\times t}$ 
 is denoted by $\|A\|_{1} \,=\, \max_{1\leq j \leq t} \sum_{i=1}^k A_{ij}$.}
 
\begin{thm}\label{thm:cover_regretnet_full}
For RegretNet with $R$ hidden layers, $K$ nodes per hidden layer, $d_g$ parameters in the allocation component, $d_p$ parameters in the payment component, and 
the vector of all model parameters
$\|w\|_{1} \leq W$,
the following are the bounds on the term $\Delta_{L}$ for different bidder valuation types:

(a) additive valuations:

$%
\Delta_{L} \leq O\big(\sqrt{R(d_g+d_p)
\log(LW\max\{K, mn\})
/
{L}
}
\big)$,\\[-12pt]

(b) unit-demand valuations:

$\displaystyle
\Delta_{L} \leq O\big(
\sqrt{R(d_g+d_p)
\log(LW\max\{K, mn\})
/
{L}
}\big)$,\\[-12pt]

(c) combinatorial valuations (with combinatorial feasible allocation):

$\displaystyle
\Delta_{L} \leq 
O\big(\sqrt{R(d_g+d_p)
\log(LW\max\{K, n\,2^m\})
/
{L}}
\big)$. 
\end{thm}

We first bound the covering number for a general feed-forward neural network and specialize it
to the three architectures we present in Section \ref{sec:regretnet} \new{and Appendix~\ref{sec:ca-architecture}}. %
\begin{lem}\label{lem:cover_fully_connect}
Let $\mathcal{F}_k$ be a class of feed-forward neural networks that maps an input vector $x \in \R^{d_0}$ to an output vector $y \in \R^{d_k}$, with each layer $\ell$ containing $T_\ell$ nodes and computing $z \mapsto \phi_\ell(w^\ell z)$, 
where each $w^\ell \in \mathbb{R}^{T_{\ell} \times T_{\ell-1}}$ 
and $\phi_\ell: \R^{T_\ell} \rightarrow [-B, +B]^{T_\ell}$. 
Further let, for each network in $\mathcal{F}_k$, let the parameter matrices $\Vert w^\ell\Vert_{1} \leq W$ and $\|\phi_\ell(s) - \phi_\ell(s')\|_1 \leq \Phi\|s-s'\|_1$ for any
$s, s' \in \R^{T_{\ell-1}}$.
\begin{align*}
\mathcal{N}_\infty(\mathcal{F}_k, \epsilon) \leq \left\lceil \frac{2Bd^2 W(2\Phi W)^k}{\epsilon}\right\rceil^d,
\end{align*}
where $T = \max_{\ell \in [k]}T_\ell$ and $d$ is the total number of parameters in a network. 
\end{lem}

\begin{proof}
We shall construct an $\ell_{1,\infty}$ cover for $\mathcal{F}_k$ by
 discretizing each of the $d$ parameters along $[-W, +W]$ at scale $\epsilon_0/d$, where we will choose  $\epsilon_0 > 0$ at the end of the proof. We will use $\hat{\mathcal{F}}_k$ to denote the subset of neural networks in $\mathcal{F}_k$ 
whose parameters are in the range $\{-(\lceil Wd/\epsilon_0 \rceil-1)\,\epsilon_0/d,  \ldots, -\epsilon_0/d, 0, \epsilon_0/d,\ldots, \lceil Wd/\epsilon_0 \rceil \epsilon_0/d\}$. The size of $\hat{\mathcal{F}}_k$ is at most $\lceil 2dW/\epsilon_0\rceil^d$. We shall now show that $\hat{\mathcal{F}}_k$ is an $\epsilon$-cover for ${\mathcal{F}}_k$.

We use mathematical induction on the number of layers $k$. We wish to show that for any $f \in \mathcal{F}_k$ there exists a $\hat{f} \in \hat{\mathcal{F}}_k$ such 
that:
\begin{equation*}
\Vert f(x) - \hat{f}(x)\Vert_1 \leq Bd\epsilon_0(2\Phi W)^k.
\end{equation*}
For $k = 0$, the statement holds trivially. Assume that the statement is true for $\mathcal{F}_k$. We now show that the statement holds for $\mathcal{F}_{k+1}$. 

A function $f \in \mathcal{F}_{k+1}$ can be written as $f(z) = \phi_{k+1}(w_{k+1} H(z))$ for some $H \in \mathcal{F}_{k}$. Similarly, a function $\hat{f} \in \hat{\mathcal{F}}_{k+1}$ can be written as $\hat{f}(z) = \phi_{k+1}(\hat{w}_{k+1} \hat{H}(z))$ for some $\hat{H} \in \hat{\mathcal{F}}_{k}$ and $\hat{w}_{k+1}$ is a matrix of entries in $\{-(\lceil Wd/\epsilon_0 \rceil-1)\,\epsilon_0/d,  \ldots, -\epsilon_0/d, 0, \epsilon_0/d,\ldots, \lceil Wd/\epsilon_0 \rceil \epsilon_0/d\}$. Also, for any parameter matrix $w^\ell \in \mathbb{R}^{T_{\ell} \times T_{\ell-1}}$, there is a matrix $\hat{\w}^\ell$ with discrete entries s.t.\
\begin{equation}
\Vert w_\ell - \hat{w}_\ell\Vert_{1} = \max_{1\leq j\leq T_{\ell-1}}\sum_{i=1}^{T_\ell} \vert w^\ell_{\ell,i,j} - \hat{w}_{\ell,i,j}\vert \leq T_\ell \epsilon_0/d \leq \epsilon_0.
\label{eq:wbound}
\end{equation}

We then have:
\begin{align*}
&\Vert f(x) - \hat{f}(x)\Vert_1\\
&\quad= \Vert \phi_{k+1}(w_{k+1} H(x)) - \phi_{k+1}(\hat{w}_{k+1}\hat{H}(x))\Vert_1\\
&\quad\leq \Phi \Vert w_{k+1} H(x) - \hat{w}_{k+1}\hat{H}(x)\Vert_1\\
&\quad\leq  \Phi \Vert w_{k+1} H(x) - w_{k+1}\hat{H}(x)\Vert_1 +  \Phi \Vert w_{k+1} \hat{H}(x) - \hat{w}_{k+1}\hat{H}(x)\Vert_1\\
&\quad\leq \Phi \Vert w_{k+1}\Vert_{1} \cdot \Vert H(x) - \hat{H}(x)\Vert_1 + \Phi \Vert w_{k+1} - \hat{w}_{k+1}\Vert_{1} \cdot \Vert \hat{H}(x)\Vert_1\\
&\quad\leq \Phi W  \Vert H(x) - \hat{H}(x)\Vert_1 + \Phi \new{T_k B}\Vert w_{k+1} - \hat{w}_{k+1}\Vert_{1}\\
&\quad\leq Bd\epsilon_0 \Phi W(2\Phi W)^k + \Phi Bd\epsilon_0\\
&\quad\leq Bd\epsilon_0 (2\Phi W)^{k+1},
\end{align*}
where the second line follows from our assumption on $\phi_{k+1}$, and the sixth line follows from our inductive hypothesis and from \eqref{eq:wbound}. By choosing $\epsilon_0 =\frac{\epsilon}{B(2\Phi W)^k}$, we complete the proof.
\end{proof}

We next bound the covering number of the auction class in terms of the
covering number for the class of allocation networks and for the class of payment networks. 
Recall that the payment networks computes a fraction  $\alpha: \R^{m(n+1)}\rightarrow [0, 1]^n$
and computes a payment 
$p_i(b) = \alpha_i(b)\cdot \langle v_i, g_i(b)\rangle$ for each bidder $i$.
Let $\mathcal{G}$ be the class of allocation networks
and   $\fP$ be the class of fractional payment functions used to construct auctions in $\mathcal{M}$.  
Let $\mathcal{N}_\infty(\mathcal{G}, \epsilon)$ and $\mathcal{N}_\infty(\mathcal{\fP}, \epsilon)$ be the corresponding covering numbers w.r.t.\ the $\ell_\infty$ norm. 
Then:
\begin{lem}
$\mathcal{N}_\infty(\mathcal{M}, \epsilon)\leq \mathcal{N}_\infty(\mathcal{G}, \epsilon/3)\cdot\mathcal{N}_\infty(\fP, \epsilon/3)$
\end{lem}

\begin{proof}
Let $\hat{\G} \subseteq \mathcal{G}$, $\hat{\fP} \subseteq \fP$ be $\ell_\infty$ covers for $\G$ and $\fP$, i.e.\
for any $g \in \G$ and $\fp \in \fP$, there exists $\hat{g} \in \hat{\G}$ and $\hat{\alpha} \in \hat{\fP}$ with
\begin{align}
\sup_b \sum_{i,j}\vert g_{ij}(b) - \hat{g}_{ij}(b)\vert \leq \epsilon/3\\
\sup_b\sum_{i}\vert \fp_i(b) - \hat{\alpha}_i(b)\vert \leq \epsilon/3.
\label{eq:gpcover}
\end{align}
We now show that the class of mechanism 
$\hat{\M} \,=\, \{(\hat{g},\hat{\alpha})\,|\, \hat{g} \in \hat{\G},\,
\text{and}\, \hat{p}(b) \,=\, \hat{\alpha}_i(b)\cdot\langle v_i, \hat{g}_i(b)\rangle\}$ 
is an $\epsilon$-cover for $\M$ under the $\ell_{1,\infty}$ distance. 
For any mechanism in $(g,p) \in \M$, let $(\hat{g}, \hat{p}) \in \hat{\M}$ be a mechanism in $\hat{\M}$ that satisfies \eqref{eq:gpcover}. We have:
\begin{align*}
&\sum_{i,j}\vert g_{ij}(b) - \hat{g}_{ij}(b)\vert + \sum_{i}\vert p_i(b) - \hat{p}_i(b)\vert\\
&\quad\leq \epsilon/3 + \sum_i \left\vert \alpha_i(b) \cdot\langle b_i, g_{i,\cdot}(b)\rangle - \hat{\alpha}_i(b)\cdot\langle b_i, \hat{g}_i(b)\rangle\right\vert\\
&\quad\leq \epsilon/3 + \sum_i \Big( \vert (\alpha_i(b) - \hat{\alpha}_i(b))\cdot\langle b_i, g_{i}(b)\rangle\vert \\
&\hspace*{95pt}+ \vert\hat{\alpha}_i(b) \cdot(\langle b_i, g_i(b)\rangle -\langle b_i, \hat{g}_{i,\cdot}(b))\rangle\vert \Big)\\
&\quad\leq \epsilon/3 + \sum_i\vert \alpha_i(b) - \hat{\alpha}_i(b)\vert + \sum_i \Vert b_i\Vert_{\infty}\cdot \Vert g_i(b) - \hat{g}_i(b)\Vert_1 \indent\indent\indent\\
&\quad\leq 2\epsilon/3 + \sum_{i,j}\vert g_{ij}(b) - \hat{g}_{ij}(b)\vert \leq \epsilon,
\end{align*}
where in the third inequality we use $\langle b_i, g_{i}(b)\rangle\leq 1$. 
The size of the cover $\hat{\M}$ is $|\hat{\G}| |\hat{\fP}|$, which completes the proof.
\end{proof}

We are now ready to prove covering number bounds for the three architectures in Section \ref{sec:regretnet} \new{and Appendix~\ref{sec:ca-architecture}}. 

%
\begin{proof}[Proof of Theorem~\ref{thm:cover_regretnet_full}]
All three architectures use the same feed-forward architecture for computing fractional payments, consisting of $K$ hidden layers with tanh activation functions. We also have by our assumption that the $\ell_1$ norm of the vector of all model parameters is at most $W$, for each $\ell = 1,\ldots,R+1$, $\|w_\ell\|_1 \leq W$.
Using that fact that the tanh activation functions are 1-Lipschitz and bounded in $[-1,1]$, and there are at most $\max\{K, n\}$ number of nodes in any layer of the payment network, we have by an application of Lemma \ref{lem:cover_fully_connect} the following bound on the covering number of the fractional payment networks $\fP$ used in each case:
\begin{align*}
\mathcal{N}_\infty(\mathcal{A}, \epsilon) \leq  \left\lceil \frac{\max(K, n)^2 (2W)^{R+1}}{\epsilon}\right\rceil^{d_p}
\end{align*}
where $d_p$ is the number of parameters in payment networks.  

For the covering number of allocation networks $\mathcal{G}$, we consider each architecture separately. In each case, we bound the Lipschitz constant for the activation functions used in the layers of the allocation network and followed by an application of Lemma \ref{lem:cover_fully_connect}. For ease of exposition, we omit the dummy scores used  in the final layer of neural network architectures.

{\bf Additive bidders.}  
The output layer computes $n$ allocation probabilities for each
 item $j$ using a softmax function. The 
  activation function $\phi_{R+1}:\R^n \rightarrow\R^n$ for the final layer 
  for input $s \in \R^{n\times m}$
 can be described as:
 $\phi_{R+1}(s)  = [\mathrm{softmax}(s_{1,1},$ $\ldots, s_{n,1}), \ldots,
 \mathrm{softmax}(s_{1,m},\ldots,s_{n,m})]$,
where $\mathrm{softmax}: \R^n \rightarrow [0,1]^n$ is defined 
for any $u \in \R^n$ as $\mathrm{softmax}_i(u) \,=\, e^{u_i} / \sum_{k=1}^n e^{u_k}$.

We then have for any $s, s' \in \R^{n\times m}$,
\begin{align}
&\|\phi_{R+1}(s) - \phi_{R+1}(s')\|_1 \nonumber\\
&\quad\new{\leq}
\sum_{j}
\left\Vert \mathrm{softmax}(s_{1,j},\ldots,s_{n,j}) - \mathrm{softmax}(s'_{1,j},\ldots,s'_{n,j})\right\Vert_1
\nonumber
\\
&\quad\leq
\sqrt{n}
\sum_{j}
\left\Vert \mathrm{softmax}(s_{1,j},\ldots,s_{n,j}) - \mathrm{softmax}(s'_{1,j},\ldots,s'_{n,j})\right\Vert_2
\nonumber
\\
&\quad\leq \sqrt{n}\frac{\sqrt{n-1}}{n} 
\sum_{j}
\sqrt{\sum_{i}
\|s_{ij} - s'_{ij}\|^2}
\nonumber
\\
&\quad\leq
\sum_{j}\sum_{i}
|s_{ij} - s'_{ij}|
\label{eq:phi-lip}
\end{align}
where the third  step follows by bounding the Frobenius norm of the Jacobian  %
of the  softmax function.

The  hidden layers $\ell = 1,\ldots,R$  are standard feed-forward layers with tanh activations. Since the tanh activation function is 1-Lipschitz,  $\|\phi_\ell(s) - \phi_\ell(s')\|_1 \leq \|s-s'\|_1$.
We also have by our assumption that the $\ell_1$ norm of the vector of all model parameters is at most $W$, for each $\ell = 1,\ldots,R+1$, $\|w_\ell\|_1 \leq W$. Moreover, the output of each hidden layer node is in $[-1,1]$, the output layer nodes is in $[0,1]$, and the maximum number of nodes in any layer (including the output layer) is at most $\max\{K, mn\}$. 

By an application of Lemma \ref{lem:cover_fully_connect} with $\Phi=1$, $B=1$, and \new{$d = \max\{K,mn\}$} %
we have
$$
\displaystyle
\mathcal{N}_\infty(\mathcal{G}, \epsilon)\leq  \left\lceil \frac{\max\{K, mn\}^2 (2W)^{R+1}}{\epsilon}\right\rceil^{d_g},$$ where $d_g$ is the number of parameters in allocation networks.

{\bf Unit-demand bidders.}  
The output layer $n$ allocation probabilities for each item $j$ as an element-wise minimum of two softmax functions. The 
  activation function 
$\phi_{R+1}:\R^2n \rightarrow\R^n$  
  for the final layer 
  for two sets of scores $s, \bar{s} \in \R^{n\times m}$
 can be described as:
 $$\phi_{R+1,i,j}(s,s')  = \min\{\textrm{softmax}_{j}(s_{i,1}, \ldots, s_{i,m}), \, \textrm{softmax}_{i}(s'_{1,j}, \ldots, s'_{n,j})\}.$$

We then have for any $s, \tilde{s}, s', \tilde{s}' \in \R^{n\times m}$,
\begin{align*}
&\|\phi_{R+1}(s,\tilde{s}) - \phi_{R+1}(s',\tilde{s}')\|_1 \\
&\quad\new{\leq}\,
\sum_{i,j}
\Big|
\min\{\textrm{softmax}_{j}(s_{i,1}, \ldots, s_{i,m}), \, \textrm{softmax}_{i}(\tilde{s}_{1,j}, \ldots, \tilde{s}_{n,j})\}
\\[-5pt]
& \hspace{1.5cm}
\,-\,
\min\{\textrm{softmax}_{j}(s'_{i,1}, \ldots, s'_{i,m}), \, \textrm{softmax}_{i}(\tilde{s}'_{1,j}, \ldots, \tilde{s}'_{n,j})\}
\Big|\\
&\quad\leq\,
\sum_{i,j}
\Big|
\max\{\textrm{softmax}_{j}(s_{i,1}, \ldots, s_{i,m})
\,-\, \textrm{softmax}_{j}(s'_{i,1}, \ldots, s'_{i,m}),\\[-5pt]
& \hspace{2.9cm}
 \textrm{softmax}_{i}(\tilde{s}_{1,j}, \ldots, \tilde{s}_{n,j})
\,-\, \textrm{softmax}_{i}(\tilde{s}'_{1,j}, \ldots, \tilde{s}'_{n,j})\}
\Big|\\
&\quad\leq\,
\sum_{i}
\big\|
\textrm{softmax}(s_{i,1}, \ldots, s_{i,m})
\,-\, \textrm{softmax}(s'_{i,1}, \ldots, s'_{i,m})
\big\|_1\\[-5pt]
& \hspace{1.5cm}
\,+\,
\sum_{j}
\big\|
 \textrm{softmax}(\tilde{s}_{1,j}, \ldots, \tilde{s}_{n,j})
\,-\, \textrm{softmax}(\tilde{s}'_{1,j}, \ldots, \tilde{s}'_{n,j})\}
\big\|_1\\
&\quad\leq\,
\sum_{i,j}
|s_{ij} - s'_{ij}|
\,+\,
\sum_{i,j}
|\tilde{s}_{ij} - \tilde{s}'_{ij}|,
\end{align*}
where the last step can be derived in the same way as \eqref{eq:phi-lip}.

As with additive bidders, using additionally hidden layers $\ell = 1,\ldots,R$  are standard feed-forward layers with tanh activations, we have 
from Lemma \ref{lem:cover_fully_connect}  with $\Phi=1$, $B=1$ and $d = \max\{K,mn\}$, %
 $$\mathcal{N}_\infty(\mathcal{G}, \epsilon)\leq  \left\lceil \frac{\max\{K, mn\}^2 (2 W)^{R+1}}{\epsilon}\right\rceil^{d_g}.$$

{\bf Combinatorial bidders.} 
The output layer outputs an allocation probability for each bidder $i$ and bundle of items $S \subseteq M$. The 
  activation function 
$\phi_{R+1}:\R^{(m+1)n2^m}  \rightarrow\R^{n2^m}$  
  for this layer 
  for $m+1$ sets of scores $s, s^{(1)},\ldots, s^{(m)} \in \R^{n\times 2^m}$ is given by:
\begin{align*}
 \phi_{R+1,i,S}(s, s^{(1)}, \ldots, s^{(m)})  =\,
\min\Big\{
&\textrm{softmax}_{S}(s_{i,S'}: S' \subseteq M),\,\textrm{softmax}_{S}(s^{(1)}_{i,S'}: S' \subseteq M),\,
\ldots,\\
&\textrm{softmax}_{S}(s^{(m)}_{i,S'}: S' \subseteq M)
 \Big\},
\end{align*}
where $\textrm{softmax}_{S}(a_{S'}: S' \subseteq M)
\,=\, e^{a_S}/ \sum_{S' \subseteq M} e^{a_{S'}}$.

We then have for any $s, s^{(1)}, \ldots, s^{(m)}, s', s'^{(1)}, \ldots, s'^{(m)} \in \R^{n\times 2^m}$,
\begin{align*}
\lefteqn{
\|\phi_{R+1}(s, s^{(1)}, \ldots, s^{(m)}) - \phi_{R+1}(s', s'^{(1)}, \ldots, s'^{(m)})\|_1
}
\\
&\new{\leq}\,
\sum_{i,S}
\Big|
\min\Big\{
\textrm{softmax}_{S}(s_{i,S'}: S' \subseteq M),\,\\[-12pt]
&\hspace*{1.6cm} 
\textrm{softmax}_{S}(s^{(1)}_{i,S'}: S' \subseteq M),\,
\ldots,
\textrm{softmax}_{S}(s^{(m)}_{i,S'}: S' \subseteq M)
 \Big\}
\\
& \hspace{0.9cm}
\,-\,
\min\Big\{
\textrm{softmax}_{S}(s'_{i,S'}: S' \subseteq M),\,\\
&\hspace*{1.6cm} 
\textrm{softmax}_{S}(s'^{(1)}_{i,S'}: S' \subseteq M),\,
\ldots,
\textrm{softmax}_{S}(s'^{(m)}_{i,S'}: S' \subseteq M)
 \Big\}
 \Big|
\\
& \leq\,
\sum_{i,S}
\max\Big\{
\big|
\textrm{softmax}_{S}(s_{i,S'}: S' \subseteq M)
\,-\,
\textrm{softmax}_{S}(s'_{i,S'}: S' \subseteq M)
\big|,\\[-10pt]
& \hspace{1.3cm}
\big|
\textrm{softmax}_{S}(s^{(1)}_{i,S'}: S' \subseteq M)
\,-\,
\textrm{softmax}_{S}(s'^{(1)}_{i,S'}: S' \subseteq M)
\big|, \,\ldots\,\\
& \hspace{1.3cm}
\big|
\textrm{softmax}_{S}(s^{(m)}_{i,S'}: S' \subseteq M)
\,-\,
\textrm{softmax}_{S}(s'^{(m)}_{i,S'}: S' \subseteq M)
\big|
\Big\}\\
& \leq\,
\sum_{i}
\big\|
\textrm{softmax}(s_{i,S'}: S' \subseteq M)
\,-\,
\textrm{softmax}(s'_{i,S'}: S' \subseteq M)
\big\|_1\\[-5pt]
& \hspace{1.2cm}
\,+\,
\sum_{i,j}
\big\|
\textrm{softmax}(s^{(j)}_{i,S'}: S' \subseteq M)
\,-\,
\textrm{softmax}(s'^{(j)}_{i,S'}: S' \subseteq M)
\big\|_1
\\
&\leq\,
\sum_{i,S}
|s_{i,S} - s'_{i,S}|
\,+\,
\sum_{i,j,S}
|s^{(j)}_{i,S} - s'^{(j)}_{i, S}|,
\end{align*}
where the last step can be derived in the same way as \eqref{eq:phi-lip}.

As with additive bidders, using additionally hidden layers $\ell = 1,\ldots,R$  are standard feed-forward layers with tanh activations, we have from Lemma \ref{lem:cover_fully_connect}  with $\Phi=1$, $B=1$ and $d = \max\{K,n\cdot 2^m\}$
$$\mathcal{N}_\infty(\mathcal{G}, \epsilon)\leq  \left\lceil \frac{\max\{K,n\cdot 2^m\}^2 (2 W)^{R+1}}{\epsilon}\right\rceil^{d_g}$$ where $d_g$ is the number of parameters in allocation networks.
\end{proof}

We now bound $\Delta_{L}$ for the three architectures using the covering number bounds we derived above. In particular, we upper bound the the `inf' over $\epsilon > 0$ by substituting a specific value of $\epsilon$: 

(a) For additive bidders, choosing $\epsilon = \frac{1}{\sqrt{L}}$, we get $$\Delta_{L} \leq O\left(\sqrt{R(d_p+d_g)\frac{\log(W\max\{K, mn\}L)}{L}}\right)$$

(b) For unit-demand bidders, choosing $\epsilon=\frac{1}{\sqrt{L}}$, we get  $$\Delta_{L} \leq O\left(\sqrt{R(d_p+d_g)\frac{\log((W\max\{K, mn\} L)}{L}}\right)$$

(c) For combinatorial bidders, choosing $\epsilon=\frac{1}{\sqrt{L}}$, we get  $$\Delta_{L} \leq O\left(\sqrt{R(d_p+d_g)\frac{\log(W\max\{K, n\cdot 2^m\} L)}{L}}\right).$$
\qed

\subsection{Proof of Theorem~\ref{THM:NEW_UNIFORM_TRIANGLE}}\label{app:duality-framework}

We apply the duality theory of~\citet{DaskalakisDT13} to verify the optimality of our proposed mechanism (motivated by empirical results of RochetNet). For the completeness of presentation, we provide a brief introduction of their approach here. 

\begin{figure}
\centering

\begin{tikzpicture}[scale=2.0,line cap=round,line join=round,>=triangle 45,x=1.0cm,y=1.0cm]
\clip(-0.5,0.5) rectangle (6.5,3.);
\draw [line width=0.4pt] (0.,1.)-- (6.,1.);
\draw [line width=0.4pt] (6.,1.)-- (0.,2.5);
\draw [line width=0.4pt] (0.,2.5)-- (0.,1.);
\draw [line width=0.8pt] (2.,1.)-- (0.,1.5);
\draw [line width=0.8pt] (2.,1.)-- (2.,2.);
\draw [line width=0.4pt,dash pattern=on 1pt off 1pt] (2.235294117647059,1.9411764705882353)-- (2.,1.);
\draw [line width=0.4pt] (2.904805183366678,1.7737987041583305)-- (2.3,1.);
\draw [line width=0.4pt] (2.6020607869702186,1.8494848032574454)-- (2.5,1.);
\draw [line width=0.4pt,dash pattern=on 1pt off 1pt] (2.904805183366678,1.7737987041583305)-- (2.5,1.);
\draw (0.5183734654275157,1.3149879513414549) node[anchor=north west] {$R_1$};
\draw (0.7,1.8441433146404131) node[anchor=north west] {$R_2$};
\draw (3.260778652611692,1.4484271299124964) node[anchor=north west] {$R_3$};
\draw (-0.3,0.9) node[anchor=north west] {$(0,1)$};
\draw (1.7,0.9) node[anchor=north west] {$(\frac{c}{3}, 1)$};
\draw (5.8,0.9) node[anchor=north west] {$(c,1)$};
\draw (-0.5,2.5) node[anchor=north west] {(0,2)};
\draw (-0.5,1.5) node[anchor=north west] {$(0,\frac{4}{3})$};
\draw [->,line width=0.8pt] (0.16406943956647332,2.3732986779393714) -- (0.16406943956647332,1.7015014340989547);
\draw [->,line width=0.8pt] (0.5597856242943913,2.249062201338747) -- (0.5597856242943913,1.5772649574983297);
\draw [->,line width=0.8pt] (1.0935423385785599,2.092616267841662) -- (1.0935423385785599,1.4208190240012466);
\draw [->,line width=0.8pt] (1.562880139069812,1.96377844025583) -- (1.562880139069812,1.291981196415413);

\begin{scriptsize}
\draw [fill=black] (0.,1.) circle (0.5pt);
\draw[color=black] (-0.03,0.9) node {$A$};
\draw [fill=black] (6.,1.) circle (0.5pt);
\draw[color=black] (6.058400051618369,0.9) node {$B$};
\draw [fill=black] (0.,2.5) circle (0.5pt);
\draw[color=black] (-0.09360621560519478,2.545849339884684) node {$C$};
\draw [fill=black] (2.,1.) circle (0.5pt);
\draw[color=black] (1.995407131679396,0.9) node {$D$};
\draw [fill=black] (0.,1.5) circle (0.5pt);
\draw[color=black] (-0.08900486461998643,1.5611602290501008) node {$E$};
\draw [fill=black] (2.,2.) circle (0.5pt);
\draw[color=black] (2.032217939561063,2.048903433482184) node {$F$};
\draw [fill=black] (2.235294117647059,1.9411764705882353) circle (0.5pt);
\draw[color=black] (2.2668868398066886,1.9890858706744758) node {$G$};
\draw [fill=black] (2.6020607869702186,1.8494848032574454) circle (0.5pt);
\draw[color=black] (2.6349949186233568,1.8970588509703088) node {$H$};
\draw [fill=black] (2.904805183366678,1.7737987041583305) circle (0.5pt);
\draw[color=black] (2.9386840836471078,1.8234372352069759) node {$I$};
\draw[color=black] (2.3,0.9) node {$J$};
\draw [fill=black] (2.3,1.) circle (0.5pt);
\draw[color=black] (2.5,0.9) node {$K$};
\draw [fill=black] (2.5,1.) circle (0.5pt);
\end{scriptsize}
\end{tikzpicture}
\caption{The transport of transformed measure of each region in Setting~\ref{exp:triangle-1}.\label{FIG:NEW_UNIFORM_TRIANGLE}}. 
\end{figure}

Let $f(v)$ be the joint valuation distribution of $v=(v_1, v_2, \cdots, v_m)$, $V$ be the support of $f(v)$ and define the measure $\mu$ with the following density, 
\begin{eqnarray}
\1_{v=\bar{v}} + \1_{v\in \partial V}\cdot f(v)(v\cdot \hat{n}(v)) - (\nabla f(v)\cdot v + (m+1)f(v)) 
\end{eqnarray}
where $\bar{v}$ is the ``base valuation'', i.e. $u(\bar{v})=0$, $\partial V$ denotes the boundary of $V$, $\hat{n}(v)$ is the outer unit normal vector at point $v\in \partial V$, and $m$ is the number of items. \new{Let $\Gamma_+(X)$ be the unsigned (Radon) measures on $X$. Consider an unsinged measure $\gamma\in \Gamma_+(X\times X)$, let $\gamma_1$ and $\gamma_2$ be the two marginal measures of $\gamma$, i.e. $\gamma_1(A) = \gamma(A\times X)$ and $\gamma_2(A) = \gamma(X\times A)$ for all measurable sets $A\subseteq X$. We say measure $\alpha$ dominates $\beta$ if and only if for all (non-decreasing, convex) functions $u$, $\int u\,d\alpha \geq \int u\,d\beta$. Then by strong duality theory we have}
\begin{equation}\label{eq:lp-duality-revenue-maximization}
\new{\sup_u \int_V u\,d\mu = \inf_{\gamma \in \Gamma_+(V, V), \gamma_1 - \gamma_2\succeq\mu}\int_{V\times V}\Vert v - v'\Vert_1 \, d\gamma,}
\end{equation}
\new{and both the supremum and infimum are achieved. Based on "complementary slackness" of linear programming, the optimal solution of Equation~\ref{eq:lp-duality-revenue-maximization} needs to satisfy the following conditions.}
\begin{cor}[\cite{DaskalakisEtAl17}]\label{cor:optimal-solution-condition}
\new{Let $u^*$ and $\gamma^*$ be feasible for their respective problems in Equation~\ref{eq:lp-duality-revenue-maximization}, then $\int u^*\,d\mu = \int \Vert v - v'\Vert_1\,d\gamma^*$ if and only if the following two conditions hold:}
\new{\begin{equation*}
\begin{aligned}
&\int u^*\,d(\gamma^*_1 - \gamma^*_2) = \int u^*\,d\mu \\
&\int u^*(v) - u^*(v') = \Vert v - v'\Vert_1, \gamma^*\text{-almost surely.}
\end{aligned}
\end{equation*}}
\end{cor}

Then we prove the utility function $u^*$ induced by the mechanism for setting~\ref{exp:triangle-1} is optimal. \new{Here we only focus on Settiong~\ref{exp:triangle-1} with $c > 1$, for $c \leq 1$ the proof is analogous and we omit here}\footnote{\new{It is fairly similar to the proof for setting $c > 1$. If $c\leq 1$, there are only two regions to discuss, in which $R_1$ and $R_2$ are the regions correspond to allocation $(0,0)$ and $(1,1)$, respectively. Then we show the optimal $\gamma^* = \bar{\gamma}^{R_1} + \bar{\gamma}^{R_2}$ where $\bar{\gamma}^{R_1}= 0$ for region $R_1$ and show $\gamma^{R_2}$ only "transports" mass of measure downwards and leftwards in region $R_2$, which is analgous to the analysis for $\gamma^{R_3}$ for setting $c > 1$.}}.
The transformed measure $\mu$ of the valuation distribution is composed of:
\begin{enumerate}
\item A point mass of $+1$ at $(0,1)$.
\item Mass $-3$ uniformly distributed throughout the triangle area (density $-\frac{6}{c}$).
\item Mass $-2$ uniformly distributed on lower edge of triangle (density $-\frac{2}{c}$).
\item Mass $+4$ uniformly distributed on right-upper edge of triangle (density $+\frac{4}{\sqrt{1+c^2}}$).
\end{enumerate}
It is straightforward to verify that $\mu(R_1)=\mu(R_2)=\mu(R_3)=0$. We will show there exists an optimal measure $\gamma^*$ for the dual program of Theorem 2 (Equation 5) in~\citet{DaskalakisDT13}. $\gamma^*$ can be decomposed into $\gamma^* =\gamma^{R_1} + \gamma^{R_2} + \gamma^{R_3}$ with $\gamma^{R_1} \in \Gamma_+(R_1\times R_1), \gamma^{R_2} \in \Gamma_+(R_2\times R_2), \gamma^{R_3} \in \Gamma_+(R_3\times R_3)$. We will show the feasibility of $\gamma^*$, such that 
\begin{equation}
\begin{array}{ccc}
\gamma^{R_1}_1 - \gamma^{R_1}_2 \succeq \mu|_{R_1};
& \gamma^{R_2}_1 - \gamma^{R_2}_2 \succeq \mu|_{R_2};&\gamma^{R_3}_1 - \gamma^{R_3}_2 \succeq \mu|_{R_3}.
\end{array}
\end{equation}
Then we show the conditions of Corollary 1 in~\citet{DaskalakisDT13} hold for each of the measures $\gamma^{R_1}, \gamma^{R_2}, \gamma^{R_3}$ separately, such that $\int u^* d(\gamma^A_1 - \gamma^A_2) = \int_A u^*d\mu$ and $u^*(v) - u^*(v') = \Vert v - v'\Vert_1$ hold $\gamma^{A}$-almost surely for any $A=R_1, R_2,$ and $R_3$. We visualize the transport of measure $\gamma^*$ in Figure~\ref{FIG:NEW_UNIFORM_TRIANGLE}.

\textbf{Construction of $\gamma^{R_1}$. } $\mu_+|_{R_1}$ is concentrated on a single point $(0,1)$ and $\mu_-|_{R_1}$ is distributed throughout a region which is coordinate-wise greater than $(0,1)$, then it is obviously to show $0\succeq\mu|_{R_1}$. We set $\gamma^{R_1}$ to be zero measure, and we get $ \gamma^{R_1}_1-\gamma^{R_1}_2 = 0 \succeq \mu|_{R_1}$. In addition, $u^*(v)=0, \forall v\in R_1$, then the conditions in Corollary 1 in~\citet{DaskalakisDT13} hold trivially.

\textbf{Construction of $\gamma^{R_2}$. } $\mu_+|_{R_2}$ is uniformly distributed on upper edge $CF$ of the triangle and $\mu_-|_{R_2}$ is uniformly distributed in $R_2$. Since we have $\mu(R_2) = 0$,  we construct $\gamma^{R_2}$ by ``transporting'' $\mu_+|_{R_2}$ into $\mu_-|_{R_2}$ downwards, that is $\gamma^{R_2}_1 = \mu_+|_{R_2}, \gamma^{R_2}_2=\mu_-|_{R_2}$. Therefore, $\int u^* d(\gamma^{R_2}_1 - \gamma^{R_2}_2) = \int u^* d\mu$ holds trivially. The measure $\gamma^{R_2}$ is only concentrated on the pairs $(v, v')$ such that $v_1 = v'_1, v_2 \geq v'_2$. Thus for such pairs $(v' v')$, we have $u^*(v) - u^*(v') = (\frac{v_1}{c} + v_2 - \frac{4}{3}) - (\frac{v_1}{c} + v'_2 - \frac{4}{3}) = ||v - v'||_1$.

\textbf{Construction of $\gamma^{R_3}$. } It is intricate to directly construct $\gamma^{R_3}$ analytically, however, we will show there the optimal measure $\gamma^{R_3}$ only transports mass from $\mu_+|_{R_3}$ to  $\mu_-|_{R_3}$ leftwards and downwards. Let's consider a point $H$ on edge $BF$ with coordinates $(v^H_1, v^H_2)$. Define the regions $R^H_{L}=\{v'\in R_3|v'_1\leq v^H_1\}$ and $R^H_U=\{v'\in R_3|v'_2 \geq v^H_2\}$. Let $\ell(\cdot)$ represent the length of segment, then we have $\ell(FH) < \frac{2}{3\sqrt{c^2+1}}$. Thus,
\begin{align*}
\mu(R^H_U) &= \frac{4\ell(FH)}{\sqrt{c^2 +1}} - \frac{6}{c}\cdot\frac{\ell^2(FH)c}{2(c^2+1)} = \frac{\ell(FH)}{\sqrt{c^2 +1}}\cdot\left(4- \frac{3\ell(FH)}{\sqrt{c^2+1}}\right) > 0\\
\mu(R^H_L) &= \frac{4\ell(FH)}{\sqrt{c^2 +1}} - \frac{2}{c}\cdot\frac{\ell(FH)c}{\sqrt{c^2+1}} - \frac{6}{c}\cdot \left(\frac{2\ell(FH)c}{3\sqrt{c^2+1}} - \frac{\ell^2(FH)c}{2(c^2+1)}\right)\\
&= \frac{\ell(FH)}{\sqrt{c^2+1}}\cdot \left(\frac{3\ell(FH)}{\sqrt{c^2+1}}-2\right) < 0
\end{align*}
Thus, there exists a unique line $l_H$ with positive slope that intersects $H$ and separate $R_3$ into two parts, $R^H_U$ (above $l_H$) and $R^H_B$ (below $l_H$), such that $\mu_+(R^H_U) = \mu_-(R^H_U)$. We will then show for any two points on edge $BF$, $H$ and $I$, lines $l_H$ and $l_I$ will not intersect inside $R_3$. In Figure~\ref{FIG:NEW_UNIFORM_TRIANGLE}, on the contrary, we assume $l_H=HK$ and $l_I =IJ$ intersects inside $R_3$. Given the definition of $l_H$ and $l_I$, we have
\begin{eqnarray*}
\mu_+(FHKD) = \mu_-(FHKD); &  \mu_+(FIJD) = \mu_-(FIJD)
\end{eqnarray*}
Since $\mu_+$ is only distributed along the edge $BF$, we have 
\begin{eqnarray*}
\mu_+(FIKD) = \mu_+(FIJD) = \mu_-(FIJD)
\end{eqnarray*}
Notice $\mu_-$ is only distributed inside $R_3$ and edge $DB$, thus $\mu_-(FIKD) > \mu_-(FIJD)$. Given the above discussion, we have
\begin{equation}\label{eq:measure_triangle_one_side}
\begin{aligned}
\mu_+(HIK) &= \mu_+(FIJD) - \mu_+(FHKD) = \mu_-(FIJD) - \mu_-(FHKD) \\
&< \mu_-(FIKD) - \mu_-(FHKD) = \mu_-(HIK)
\end{aligned}
\end{equation}

On the other hand, let $S(HIK)$ be the area of triangle $HIK$, $DG$ be the altitude of triangle $DBF$ w.r.t $BF$, and $h$ be the altitude of triangle $HJK$ w.r.t the base $HI$.
\begin{align*}
\mu_-(HJK) &= \frac{6}{c}\cdot S(HIK) = \frac{6}{c}\cdot \frac{1}{2} \ell(HI)h \leq \frac{3}{c}\cdot\ell(HI)\cdot\ell(DG)\\
&=\frac{3}{c} \cdot \frac{2c}{3\sqrt{c^2+1}}\cdot \ell(HI)  = \frac{2}{\sqrt{c^2+1}}\cdot \ell(HI)\\
&<\frac{2}{\sqrt{c^2+1}}\cdot \ell(HI) = \mu_+(HIK),
\end{align*}
which is a contradiction of Equation~\ref{eq:measure_triangle_one_side}. Thus, we show $l_H$ and $l_I$ doesn't intersect inside $R_3$. Let $\gamma^{R_3}$ be the measure that transport mass from $\mu_+|_{R_3}$ to $\mu_-|_{R_3}$ through lines $\{l_H|H\in BF\}$. Then we have $\gamma^{R_3}_1 = \mu_+|_{R_3}, \gamma^{R_3}_2 = \mu_-|_{R_3}$, which leads to $\int u^*d(\gamma^{R_3}_1 - \gamma^{R_3}_2) =\int u^* d\mu$. The measure $\gamma^{R_3}$ is only concentrated on the pairs $(v, v')$, s.t. $v_1 \geq v'_1$ and $v_2 \geq v'_2$. Therefore, for such pairs $(v, v')$, we have $u^*(v) - u^*(v') = (v_1 + v_2 - \frac{c}{3}-1) - (v'_1 + v'_2 - \frac{c}{3}-1) = (v_1 - v'_1) + (v_2 - v'_2) = ||v-v'||_1$.

Finally, we show there must exist an optimal measure $\gamma$ for the dual program of Theorem 2 in~\citet{DaskalakisDT13}. \qed

\end{document}